\documentclass[envcountsame]{llncs}
\usepackage{graphicx}
\usepackage{calc}
\usepackage[text={\textwidth+1.3cm, \textheight+1cm},centering]{geometry}
\usepackage[binary-units]{siunitx}
\usepackage{amsmath}
\usepackage{mathtools}
\usepackage{amsfonts}
\usepackage{tikz}
\usepackage{tkz-graph}
\usetikzlibrary{calc}
\usepackage[utf8]{inputenc}
\usepackage[numbers,square]{natbib}
\usepackage{float}
\usepackage[linesnumbered, boxruled]{algorithm2e}
\usepackage{enumerate}
\usepackage{environ}
\usepackage{multirow}
\usepackage{longtable}
\usepackage{booktabs}
\usepackage{xstring}
\usepackage{breakurl}
\PassOptionsToPackage{hyphens}{url}\usepackage[breaklinks]{hyperref}

\newcommand{\bigO}{\mathcal{O}}
\newcommand{\bigL}{\mathcal{L}}
\newcommand{\smt}{\operatorname{smt}}
\newcommand{\mst}{\operatorname{mst}}
\newcommand{\tsp}{\operatorname{tsp}}

\newcommand{\smtvertex}{l}
\newcommand{\term}{R}
\newcommand{\rootterm}{r_0}
\newcommand{\sources}{{\term \setminus \{ \rootterm \}}}
\newcommand{\sourcesalone}{{(\sources)}}
\newcommand{\nonneg}{\mathbb{R}_{\geq 0}}

\newcommand{\pos}{\mathbb{R}_{> 0}}
\newcommand{\back}{b}

\newcommand{\jterm}[1]{\bigL_{#1}}
\newcommand{\jplustterm}[1]{\jterm{#1}}
\newcommand{\onetree}{\bigL_{\text{1-tree}}}

\newcommand{\tspbound}{\bigL_{TSP}}
\newcommand{\dist}{d}
\newcommand{\polylog}{\operatorname{polylog}}

\newcommand{\graphminus}{-}
\newcommand{\graphplus}{+}
\newcommand{\np}{NP}

\newcommand{\distgraph}[1]{G_{#1}}

\newcommand{\numb}[1]{\num[round-mode=places,round-precision=\prec]{#1}}
\newcommand{\numbP}[1]{\num[round-mode=places,round-precision=2]{#1}}

\newcommand{\explicitqed}{\qed}
\newcommand{\subgraphname}{H}
\newcommand{\simpathwidth}{p}

\newcommand{\nosol}{ -- }
\newcommand{\notime}{timeout}
\newcommand{\nomem}{memout}

\newcommand{\testset}[3]
{
\begin{longtable}{l r r r r r}
Instance & $|V|$ & $|E|$ & $|\term|$ & Opt & Time [s] \\
\midrule \midrule \endhead
#2
\caption{Results on the testset #1. #3}
\end{longtable}
\vspace{\extraspace}
}

\newcommand{\typeinternal}[1] {Type: #1}

\newcommand{\typeVlsi}
{\typeinternal{VLSI-derived grid graphs with holes.}}
\newcommand{\typeRandom}
{\typeinternal{Random sparse graphs with random costs.}}
\newcommand{\typeGridtwoD}
{\typeinternal{2D grid graphs.}}
\newcommand{\typeGridthreeD}
{\typeinternal{3D grid graphs.}}
\newcommand{\typePUC}
{\typeinternal{Artificial instances designed to be hard for existing solvers.}}
\newcommand{\typeIncidenceCost}
{\typeinternal{Random graphs with so-called incidence costs, designed to defy preprocessing.}}
\newcommand{\typeGroup}
{\typeinternal{Group Steiner tree instances arising from VLSI design modeled as Steiner tree instances by connecting each terminal to the vertices if its group by edges of very high cost.}}
\newcommand{\typeFST}
{\typeinternal{Rectilinear instances after FST-preprocessing by GeoSteiner.}}
\newcommand{\typeObstacle}
{\typeinternal{Instances of the Obstacle-avoiding rectilinear Steiner tree problem after FST-preprocessing by ObSteiner and merging the FSTs into a single graph.}}
\newcommand{\typeViennaSimple}
{\typeinternal{Real-world telecommunication networks after a ``simple'' preprocessing routine.}}
\newcommand{\typeViennaAdvance}
{\typeinternal{Real-world telecommunication networks after an ``advanced'' preprocessing routine. We report the cost of an optimum solution in the original instance, computed as the sum of an
optimum solution in the reduced instance and the fixed cost induced by the reductions.}}
\newcommand{\typeGap}
{\typeinternal{Artificial instances arising from generalizations of Steiner tree LP gap examples.}}
\newcommand{\typePUCn}
{\typeinternal{Unweighted instances of the PUC testset, which contains artificial instances designed to be hard for existing solvers.}}
\newcommand{\typePfourE}
{\typeinternal{Complete graphs with Euclidean costs.}}
\newcommand{\typePfourZ}
{\typeinternal{Complete graphs with random costs.}}
\newcommand{\typePsixE}
{\typeinternal{Sparse graphs with Euclidean costs.}}
\newcommand{\typePsixZ}
{\typeinternal{Sparse graphs with random costs.}}
\newcommand{\typeArt}
{\typeinternal{Artificial instances.}}

\newcommand{\lowerboundalgo}[1]
{\IfEqCase{#1}{
{ (0) }{$\equiv 0$}%
{ (1) }{$\jterm{2}$}%
{ (2) }{$\jplustterm{2}$}%
{ (3) }{$\jterm{3}$}%
{ (4) }{$\jplustterm{3}$}%
{ (5) }{$\onetree$}%
{ (6) }{$\tspbound$}%
{ (7) }{$\max(\jplustterm{2},\onetree)$}%
{ (8) }{$\max(\jplustterm{2},\tspbound)$}%
{ (9) }{$\max(\jplustterm{3},\onetree)$}%
{ (10) }{$\max(\jplustterm{3},\tspbound)$}%
} & #1
}
\newcommand{\candset}{\mathcal{J}}
\newcommand{\candsubset}{{\candset'}}
\newcommand{\candelem}{J}

\newcommand{\doubleEdgeColor}[6]{
\path (#1) edge [color=#3,thick,
#4,
bend left=10] (#2);
\path (#1) edge [color=#5,thick,
#6,
bend right=10] (#2);}

\begin{document}

\title{Dijkstra meets Steiner: a fast exact goal-oriented Steiner tree algorithm}


\author{Stefan Hougardy,
        Jannik Silvanus, \and
        Jens Vygen
}


\institute{Research Institute for Discrete Mathematics,
University of Bonn\\
\email{\{hougardy, silvanus, vygen\}@or.uni-bonn.de}
}

\date{September 8, 2015}

\maketitle

\begin{abstract}
We present a new exact algorithm for the Steiner tree problem in edge-weighted graphs.
Our algorithm improves the classical dynamic programming approach by Dreyfus and Wagner.
We achieve a significantly better practical performance via pruning and future costs,
a generalization of a well-known concept to speed up shortest path computations.
Our algorithm matches the best known worst-case run time and has a fast, often superior, practical performance:
on some large instances originating from VLSI design, previous best run times are improved upon by orders of magnitudes.
We are also able to solve larger instances of the $d$-dimensional rectilinear Steiner tree problem for $d \in \{3, 4, 5\}$,
whose Hanan grids contain up to several millions of edges.

\end{abstract}

\section{Introduction}
We consider the well-known Steiner tree problem (in graphs): Given an undirected graph $G$, costs $c : E(G) \rightarrow \nonneg$ and a terminal set $\term \subseteq V(G)$, find a tree $T$ in $G$ such that $\term \subseteq V(T)$ and
$c(E(T))$ is minimum.
The decision version of the Steiner tree problem is one of the classical \np-complete problems \cite{karp};
it is even \np-complete in the special case that $G$ is bipartite with $c \equiv 1$.
Furthermore, it is \np-hard to approximate the Steiner tree problem within a factor of $\frac{96}{95}$ \cite{apxhard}.
The currently best known  approximation algorithm by Byrka et al.\ \cite{byrkaln4} uses polyhedral methods to achieve a $1.39$-approximation.
The Steiner tree problem has many applications, in particular in VLSI design \cite{vlsi}, where electrical connections are realized by Steiner trees.

From now on, we will refer to $|V(G)|$ by $n$, $|E(G)|$ by $m$ and $|\term|$ by $k$.
Dreyfus and Wagner \cite{dreyfus} applied dynamic programming to the Steiner tree problem to obtain an exact algorithm with a run time of $\bigO(n(n\log n + m) +3^kn + 2^kn^2)$ if implemented using Fibonacci heaps \cite{fibheap}.
In 1987, Erickson, Monma and Veinott  \cite{veinott} improved the run time to $\bigO(3^k n + 2^k (n \log n + m))$ using a very similar approach.
In 2006, Fuchs et al.\ \cite{asc} proposed an algorithm with a run time of $\bigO((2+\delta)^k n^{\left({{\ln(\frac{1}{\delta})}/{\delta}}\right)^\zeta})$ for every sufficiently small $\delta > 0$
and $\zeta > \frac{1}{2}$, improving the exponential dependence on $k$ from $3^k$ to $(2+\delta)^k$.
Vygen \cite{vygen11} developed an algorithm with a worst-case run time of $\bigO(nk 2^{k + \log_2(k)\log_2(n)})$, which is the fastest known algorithm if $f(n) < k < g(n)$
for some $f(n) = \polylog(n)$ and $g(n) = \frac{n}{2} - \polylog(n)$. However, for $k < 4 \log n$, the run time obtained by Erickson, Monma and Veinott \cite{veinott} is still the best known. See \cite{vygen11} for a more detailed analysis of the run times mentioned above.

For graphs with treewidth $t$, one can solve the Steiner tree problem in time $\bigO(n 2^{\bigO(t)})$ \cite{bodlaender13}. An implementation of this algorithm was evaluated in \cite{bodlander13impl}.
Polzin and Vahdati Daneshmand \cite{polzin06} proposed an algorithm with a worst-case run time
of $\bigO(n 2^{\simpathwidth \log \simpathwidth + 3\simpathwidth + \log \simpathwidth})$ where $\simpathwidth$ is a parameter closely related to the pathwidth of $G$. They use this algorithm as a subroutine in their successful reduction-based Steiner tree solver \cite{disspolzin, dissvahdati}.

Except for the last mentioned algorithm, these results have played a very limited role in practice. Instead, empirically successful algorithms rely on preprocessing and reduction techniques, heuristics and branching:
First, reductions \cite{Bea84, DV89a, PV01c, UdAR99} are applied to reduce the size of the graph and the number of terminals, guaranteeing that optimum solutions of the reduced instance correspond to optimum solutions of the original instance. These reductions are not limited to simple local edge elimination tests, but may also rely on linear programming formulations and optimum solutions of partial instances.
Primal and dual \cite{dAUW01, wongdualascent} heuristics yield good upper and lower bounds, in many cases even resulting in a provably optimum solution. If these methods do not already solve the instance, enumerative algorithms are used.
To this end, various authors \cite{dAUW01, CGR92, KM98} perform branch and cut. However, the solver by Polzin and Vahdati Daneshmand \cite{disspolzin, dissvahdati}, which achieved the best results so far, uses a
branch and bound approach, where high effort is put into single branching nodes.

We propose a dynamic programming based algorithm with a worst-case run time of $\bigO(3^k n + 2^k (n \log n + m))$, matching the best known bound for small $k$, and which is fast in practice.
Good practical performance is achieved by effectively pruning partial solutions and using \emph{future cost} estimates.
The latter are motivated by the similarity of our algorithm with Dijkstra's algorithm \cite{dijkstra} and the well-known speed-up technique for Dijkstra's algorithm (first described by Hart, Nilsson and Raphael \cite{astar}) which uses reduced edge costs.
More precisely, given an instance $(G, c, s, t)$ of the shortest path problem, where we assume $G$ to be a directed graph,
we use a feasible potential $\pi$, which is a function $\pi: V(G) \rightarrow \nonneg$
with
 \begin{align}
  \pi(t) &= 0 \label{feasiblepotential:1} \\
\shortintertext{and}
  \pi(v) &\leq \pi(w) + c((v, w)) \label{feasiblepotential:2}
\intertext{for all $(v,w) \in E(G)$. Define reduced costs $c_{\pi}$ by}
c_{\pi}(e)&:=c(e)+\pi(w)-\pi(v) \geq 0 \label{reducededgecosts}
 \end{align}
for every $e=(v,w)\in E(G)$. Then, run Dijkstra's algorithm on the instance $(G, c_{\pi}, s, t)$.
The numbers $\pi(v)$ are lower bounds on the
distance from $v$ to $t$ (we also say that $\pi(v)$ estimates the \emph{future cost} at $v$).
Moreover, the cost of every $s$-$t$-path changes by the same amount when going from $c$ to $c_{\pi}$, namely $-\pi(s)$, so any shortest $s$-$t$-path
in $(G,c_{\pi})$ is a shortest $s$-$t$-path in $(G,c)$.
If the future costs are good lower bounds, only vertices close to a shortest path will be labeled by Dijkstra's algorithm before $t$ is
labeled permanently (i.e., the distance to $t$ is known) and the algorithm can be stopped. This can lead to huge speedups.

The rest of this paper is organized as follows:
In Section\nobreakspace \ref {algosection}, we generalize this future cost idea from paths to Steiner trees  and describe our algorithm.
Examples of future cost estimates are given in Section\nobreakspace \ref {seclb}.
Section\nobreakspace \ref {secpruning} introduces a pruning technique to further improve practical performance.
Section\nobreakspace \ref {secresults} contains  implementation details and computational results.

\section{The Algorithm}
\label{algosection}
Now, let $(G, c, \term)$ be an instance of the Steiner tree problem, where $G$ is an undirected graph, $c : E(G) \rightarrow \nonneg$ is a cost function on the edges and $\term$ is the set of terminals to be connected.
As usual, for a set $X \subseteq V(G)$, we denote by $\smt(X)$, short for Steiner minimal tree, the cost of an optimum Steiner tree for the terminal set $X$.
Our algorithm uses an arbitrary root terminal $\rootterm \in \term$.
We will call the terminals in $\sources$ source terminals.
The algorithms by Dreyfus and Wagner\nobreakspace \cite{dreyfus} and Erickson et al.\ \cite{veinott} as well as our algorithm use dynamic programming to compute
$\smt(\{v\} \cup I)$ for $(v,I) \in V(G) \times  2^\sources$.
Then, at termination, $\smt(\{\rootterm\} \cup \sourcesalone)$ is the cost of an optimum Steiner tree.

The former two algorithms work as follows: For each $i$ from 1 to $|\sources|$, they consider
all $I \subseteq \sources$ with $|I| = i$ one after another and then compute $\smt(\{v\} \cup I)$ for all $v \in V(G)$.
This way, it is guaranteed that when computing $\smt(\{v\} \cup I)$, the values $\smt(\{w\} \cup I^\prime)$ for all $w \in V(G)$ and $I^\prime \subset I$ are already known.
However, this leads to an exponential best case run time and memory consumption of $\Omega(2^k n)$.
In contrast, our new algorithm considers all subsets of source terminals simultaneously, using a labeling technique similar to Dijkstra's algorithm.
This way, we do not necessarily have to compute $\smt(\{v\} \cup I)$ for all pairs $(v,I)$.

Our new algorithm labels from the source terminals towards the root $\rootterm$. More precisely, the algorithm labels elements of $V(G) \times 2^\sources$.
Each label $(v, I)$ represents an optimum Steiner tree for $\{v\} \cup I$ found so far.
As in Dijkstra's algorithm, each iteration selects one label $(v,I)$ and declares it to be permanent.
Each time a label $(v,I)$ becomes permanent, all neighbors $w$ of $v$ are checked and updated if the Steiner tree represented by $(v,I)$ plus the edge $\{v,w\}$ leads to a better solution for $(w,I)$ than previously known. This operation is well-known from Dijkstra's algorithm.
In addition, for all sets $J \subseteq \sourcesalone \setminus I$ it is checked whether the Steiner trees for $(v,I)$ and $(v,J)$ combined to a tree for $(v,I \cup J)$ lead to a better solution than previously known.

To allow a simpler presentation, we restrict ourselves to instances without edges of zero cost, as these can be contracted in a trivial preprocessing step.

Now, we introduce the notion of valid lower bounds, which are used by the algorithm to estimate the future cost of a label $(v, I)$.

\begin{definition}
Let $(G, c, \term)$ be an instance of the Steiner tree problem and $\rootterm \in \term$.
A function $\bigL : V(G) \times 2^\term \rightarrow \mathbb{R}_{\geq 0}$ is called a \emph{valid lower bound} if
 \begin{align*}
  \bigL(\rootterm,\{\rootterm\}) = 0
 \end{align*}
and
 \begin{align*}
  \bigL(v,I) \leq \bigL(w,I^\prime) + \smt((I \setminus I^\prime) \cup\{v,w\})
 \end{align*}
 for all $v,w \in V(G)$ and $\{\rootterm\} \subseteq I^\prime \subseteq I \subseteq \term$.
\end{definition}

Note that the values $\bigL(v,I)$ for $\rootterm \notin I$ do not affect whether $\bigL$ is a valid lower bound.
Also note that by choosing $I^\prime = \{\rootterm\}$ and $w = \rootterm$, we have $\bigL(v,I) \leq \smt(I \cup\{v\})$, so a valid lower bound by definition indeed is a
lower bound on the cost of an optimum Steiner tree.
Moreover, if $e = \{v,w\} \in E(G)$ is an edge, by choosing $I^\prime = I$, we have
\begin{align*}
  \bigL(v,I) \leq \bigL(w,I) + \smt(\{v,w\}) \leq \bigL(w,I) + c(e).
\end{align*}
This shows that valid lower bounds generalize feasible potentials as defined in (\ref{feasiblepotential:1}) and (\ref{feasiblepotential:2}). In fact, our algorithm applied to the case $|\term| = 2$
is identical to Dijkstra's algorithm using future costs in the very same way.

\renewcommand{\algorithmcfname}{}
\renewcommand{\thealgocf}{}
\SetAlgoCaptionSeparator{}

\begin{figure}[t!]

\begin{algorithm}[H]
\SetKwInOut{Input}{Input}\SetKwInOut{Output}{Output}
\SetKwFunction{backtrack}{backtrack}

\Input{A connected undirected graph $G$, costs $c : E(G) \rightarrow \pos$, a terminal set $\term \subseteq V(G)$, a root terminal $\rootterm \in \term$,
and (an oracle to compute) a valid lower bound $\bigL : V(G) \times 2^\term \rightarrow \mathbb{R}_{\geq 0}$.}

\Output{The edge set of an optimum Steiner tree for $\term$ in $G$.}
\BlankLine
$\mathrlap{\smtvertex(v,I)}\phantom{\smtvertex(s,\{s\})}
  := \infty$ for all $(v,I) \in V(G) \times 2^\sources$\label{dijkstrasteiner:firstline}\;
$\smtvertex(s,\{s\}) := 0$ for all $s \in \sources$\;
$\mathrlap{\smtvertex(v,\emptyset)}\phantom{\smtvertex(s,\{s\})}
 := 0$ for all $v \in V(G)$\;
$\mathrlap{\back(v,I)}\phantom{\smtvertex(s,\{s\})}
 := \emptyset$ for all $(v,I) \in V(G) \times 2^\sources$\;
$\mathrlap{N}\phantom{\smtvertex(s,\{s\})}
 := \{(s,\{s\})\ |\ s \in \sources \}$\;
$\mathrlap{P}\phantom{\smtvertex(s,\{s\})}
:= V(G) \times \{\emptyset\}$\label{dijkstrasteiner:afterinit}\;
\While{$(\rootterm,\sources) \notin P$}{
Choose $(v, I) \in N$ minimizing $\smtvertex(v,I) + \bigL(v, \term \setminus I)$\label{dijkstrasteiner:firstloop}\;
$N := N \setminus \{(v, I)\}$\;
$\mathrlap{P}\phantom{N} := P \cup \{(v, I)\}$\;
\For{\emph{\textbf{all}} \upshape{edges} $e = \{v,w\}$ \upshape{incident to} $v$}{
\If{$\smtvertex(v,I) +  c(e) < \smtvertex(w, I)$ \emph{\textbf{and}} $(w, I) \notin P$ \label{dijkstrasteiner:edge}}
{
$\mathrlap{\smtvertex(w,I)}\phantom{\back(w,I)} := \smtvertex(v, I) + c(e)$\; \label{dijkstrasteiner:update:neighbor}
$\back(w,I) := \{(v, I)\}$\;
$\mathrlap{N}\phantom{\back(w,I)}  := N \cup \{(w,I)\}$\; \label{dijkstrasteiner:afterupdate:neighbor}
}
}
\For{\emph{\textbf{all}} $\emptyset \neq J \subseteq \sourcesalone \setminus I$ \upshape{with} $(v,J) \in P$\label{dijkstrasteiner:superset}}{
\If{$\smtvertex(v,I) + \smtvertex(v, J) < \smtvertex(v, I \cup J)$ \emph{\textbf{and}} $(v, I \cup J) \notin P$}
{
$\mathrlap{\smtvertex(v,I \cup J) }\phantom{\back(v,I \cup J)}
:= \smtvertex(v, I) +\smtvertex(v, J)$\;
$\back(v,I \cup J) := \{(v, I), (v, J)\}$\;
$\mathrlap{N}\phantom{\back(v,I \cup J)}
 := N \cup \{(v,I \cup J)\}$\;\label{dijkstrasteiner:afterupdate:superset}
}
}\label{dijkstrasteiner:endloop} \label{dijkstrasteiner:endupdatesupersets}
}
 \KwRet{\backtrack{$\rootterm, \sources$}}\;
\BlankLine
\SetNlSty{phantom}{}{}
  \SetKwProg{myproc}{Procedure}{}{}
  \myproc{\backtrack{$v, I$}}{
   \setcounter{AlgoLine}{26}
\SetNlSty{textbf}{}{}
   \eIf{$\back(v,I) = \{(w,I)\}$}
    {\KwRet{$\{\{v,w\}\} \cup \text{\backtrack{$w,I$}}$}\;}
    {\KwRet{$\bigcup_{(w,I^\prime) \in \back(v,I)} \text{\backtrack{$w,I^\prime$}}$\label{dijkstrasteiner:algo:backtrack:e}}\;}
   }
\caption{\textbf{Dijkstra-Steiner Algorithm}}
\label{dijkstrasteiner}
\end{algorithm}
\caption{The Dijkstra-Steiner Algorithm}
\label{dijkstrasteiner:figure}
\end{figure}

Our new algorithm is described in Figure \ref{dijkstrasteiner:figure}.
For each label $(v,I) \in V(G) \times 2^\sources$, the algorithm stores the cost $\smtvertex(v,I) \in \nonneg \cup \{ \infty \}$ of the cheapest Steiner tree for $\{v\} \cup I$ found so far as well as backtracking data $\back(v,I) \subseteq V(G) \times 2^\sources$ which is used to construct the Steiner tree represented by this label.
If $\back(v,I)$ is not empty, it will always either be of the form $\back(v,I) = \{(w, I)\}$ where $w$ is a neighbor of $v$ or of the form
$\back(v,I) = \{(v, I_1), (v, I_2)\}$ where $I_1$ and $I_2$ form a partition of $I$ (into nonempty disjoint sets).
In the first case, i.e., $\back(v,I) = \{(w, I)\}$, the Steiner tree represented by the
label $(v, I)$ contains exactly one edge incident to $v$, which is $\{v,w\}$.
In the second case, i.e., $\back(v,I) = \{(v, I_1), (v, I_2)\}$, the Steiner tree represented by $(v, I)$
can be split into two Steiner trees for
the terminal sets $I_1 \cup \{v\}$ and $I_2 \cup \{v\}$, where $I_1$ and $I_2$ form a partition of $I$.
To be precise, it may happen that the subgraph of $G$ corresponding to some label $(v, I)$ contains cycles.
However, since we ruled out edges of zero cost, this can only be the case as long as the label is not permanent,
because -- as we will show -- permanent labels correspond to optimal Steiner trees.

We incorporate the valid lower bound $\bigL$ into the algorithm as follows.
First note that instead of running Dijkstra's algorithm on reduced edge costs $c_{\pi}$ as defined in (\ref{reducededgecosts}),
one can equivalently apply the following modification. For a vertex $v \in V(G)$, we denote by $\smtvertex(v)$ the cost of a label $v$ in Dijkstra's algorithm, i.e., the cost of a shortest path connecting $s$ with $v$ found so far. Then, in each iteration, instead of choosing a non-permanent vertex $v$ minimizing
$\smtvertex(v)$ to become permanent, choose a non-permanent vertex minimizing $\smtvertex(v) + \pi(v)$.
We generalize this approach by always choosing a non-permanent label minimizing $\smtvertex(v,I) + \bigL(v, \term \setminus I)$.

By $P \subseteq V(G) \times 2^\sources$ we denote the set of permanently labeled elements. We also maintain a set $N$ of non-permanent labels which are candidates to be selected.
The set $N$ contains exactly the labels
$(v,I) \in (V(G) \times 2^\sources) \setminus P$ with $\smtvertex(v,I) < \infty$.
Note that $\bigL \equiv 0$ is always a valid lower bound, which may serve as an example to help understanding the algorithm. Other examples of valid lower bounds will be discussed in Section\nobreakspace \ref {seclb}.


\begin{theorem} \label{dijkstrasteiner:proof:correct}
The Dijkstra-Steiner algorithm works correctly:
Given an instance $(G, c, \term)$  of the Steiner tree problem, $\rootterm \in \term$, and a valid lower bound $\bigL$,
it always returns the edge set of an optimum Steiner tree for $\term$ in $G$.
\end{theorem}
\begin{proof}
We will prove that the following invariants always hold when line\nobreakspace \ref {dijkstrasteiner:firstloop} is executed:
\begin{enumerate}[(a)]
\item For each $(v,I) \in N \cup P$:
\begin{enumerate}[({a}1)]
\item $\smtvertex(v,I) = \begin{cases}
     c(\{v,w\}) + \smtvertex(w,I) & \text{if } \back(v,I) = \{(w,I)\}, \\
     \sum_{(v,J) \in \back(v,I)} \smtvertex(v,J) & \text{otherwise},
   \end{cases}$ \label{dijkstrasteiner:proof:correct:a:1}
\item $I \cup \{v\} = \{v\} \cup \bigcup_{(w,J) \in \back(v,I)}J$, \label{dijkstrasteiner:proof:correct:a:2}
\item $\back(v,I) \subseteq P$, and
\texttt{backtrack}($v, I$) returns the edge set $F$ of a connected subgraph of $G$ containing $\{v\} \cup I$ with $c(F) \leq \smtvertex(v,I)$. \label{dijkstrasteiner:proof:correct:a:3}
\end{enumerate}
\label{dijkstrasteiner:proof:firstinv}
\item For each $(v,I) \in P$: \label{dijkstrasteiner:proof:secondinv} \label{dijkstrasteiner:proof:correct:b}
\begin{enumerate}[\phantom{({a}1)}]
\item $\smtvertex(v,I) = \smt(\{v\}\cup I)$.
\end{enumerate}

\item For each $(v,I) \in (V(G) \times 2^\sources) \setminus P$:
\begin{enumerate}[({c}1)]
\item $\smtvertex(v,I) \geq \smt(\{v\}\cup I)$, \label{dijkstrasteiner:proof:correct:c:1}
\item If $I = \{v\}$, then $\smtvertex(v,I) = 0$,  \label{dijkstrasteiner:proof:correct:c:2}
otherwise \\ $\smtvertex(v,I) \leq \min_{\{v,w\} \in E(G), (w,I) \in P}(\smtvertex(w,I) + c(\{v,w\}))$
and \\ $\smtvertex(v,I) \leq \min_{\emptyset \neq I' \subset I \text{ and }(v,I'), (v, I \setminus I') \in P}(\smtvertex(v, I') + \smtvertex(v, I \setminus I'))$.
\item If $\smtvertex(v,I) = \smt(\{v\}\cup I)$, then $(v,I) \in N$. \label{dijkstrasteiner:proof:correct:c:3}
\end{enumerate}\label{dijkstrasteiner:proof:thirdinv}
\item $N$ is not empty and $N \cap P = \emptyset$.\label{dijkstrasteiner:proof:fourthinv}\label{dijkstrasteiner:proof:lastinv}
\end{enumerate}
Assuming  (\ref{dijkstrasteiner:proof:firstinv}) --  (\ref{dijkstrasteiner:proof:lastinv}), the correctness of the algorithm directly follows:
Once we have $(\rootterm,\sources) \in P$,  (\ref{dijkstrasteiner:proof:secondinv})
implies $\smtvertex(\rootterm, \sources) = \smt(\{\rootterm\} \cup (\sources)) = \smt(\term)$. Furthermore,
(\ref{dijkstrasteiner:proof:correct:a:3}) implies that the algorithm returns the edge set $F$ of a connected subgraph of $G$ containing $\term$ with $c(F) \leq \smtvertex(\rootterm, \sources) = \smt(\term)$. Since there
are no edges of zero cost, $F$ indeed is the edge set of a tree.
Invariant (\ref{dijkstrasteiner:proof:fourthinv}) guarantees that in each iteration, a label $(v,I) \in N$ can be chosen
and that the algorithm terminates, since $|P|$ is increased in each iteration.

Clearly, after line\nobreakspace \ref {dijkstrasteiner:afterinit} these invariants hold.
We have to prove that lines\nobreakspace  \ref {dijkstrasteiner:firstloop} to\nobreakspace  \ref {dijkstrasteiner:endloop}
preserve (\ref{dijkstrasteiner:proof:firstinv}) --  (\ref{dijkstrasteiner:proof:lastinv}).
Below, we first verify that (\ref{dijkstrasteiner:proof:firstinv}) and (\ref{dijkstrasteiner:proof:thirdinv}) are preserved. Then, the main part of the proof shows that (\ref{dijkstrasteiner:proof:secondinv}) is preserved. Finally, we will see that the latter argument also shows that (\ref{dijkstrasteiner:proof:fourthinv}) is preserved.\\
Let $(v,I)$ be the label chosen in line\nobreakspace \ref {dijkstrasteiner:firstloop} in some iteration.
Clearly, lines\nobreakspace  \ref {dijkstrasteiner:firstloop} to\nobreakspace  \ref {dijkstrasteiner:endloop}
preserve (\ref{dijkstrasteiner:proof:correct:a:2}), (\ref{dijkstrasteiner:proof:correct:a:3}),\nobreakspace (\ref{dijkstrasteiner:proof:correct:c:2}) and (\ref{dijkstrasteiner:proof:correct:c:3}).
If a label $\smtvertex(w, I)$ or $\smtvertex(v, I \cup J)$ is decreased, it cannot be permanent,
so (\ref{dijkstrasteiner:proof:correct:a:1}) is maintained.

We now consider (\ref{dijkstrasteiner:proof:correct:c:1}). Since (\ref{dijkstrasteiner:proof:thirdinv}) held before the current iteration,
we have $\smtvertex(v,I)\geq \smt(\{v\} \cup I)$.
This directly implies
\begin{align*}
\smtvertex(v,I) +c(\{v,w\}) &\geq \smt(\{v\} \cup I) +c(\{v,w\}) \\ &\geq \smt(\{v,w\} \cup I) \\ &\geq \smt(\{w\} \cup I),
\end{align*}
so line \ref{dijkstrasteiner:update:neighbor} does not destroy (\ref{dijkstrasteiner:proof:correct:c:1}).
Also, if $\emptyset \neq J \subseteq \sourcesalone \setminus I$ is a set chosen in line\nobreakspace \ref {dijkstrasteiner:superset} leading to the change of $\smtvertex(v, I \cup J)$, we have
\begin{align*}
\smtvertex(v, I \cup J) &= \smtvertex(v,I) + \smtvertex(v, J) \\& \geq \smt(\{v\} \cup I) + \smt(\{v\} \cup J) \\ &\geq \smt(\{v\} \cup I \cup J),
\end{align*}
so\nobreakspace (\ref{dijkstrasteiner:proof:correct:c:1}) indeed is preserved.\\
In order to prove that (\ref{dijkstrasteiner:proof:correct:b}) is preserved, we now show that $\smtvertex(v,I) \leq \smt(\{v\} \cup I)$, which together with (\ref{dijkstrasteiner:proof:correct:c:1}) yields the desired result. Since $\smtvertex(v, \{v\}) = 0 = \smt(\{v\})$, we can assume $I \neq \{v\}$.
Let $T$ be an optimum Steiner tree for $\{v\} \cup I$ in $G$. Then, all leaves of $T$ are contained in $\{v\} \cup I$.
For a vertex $w \in V(T)$, let $T_w$ be the subtree of $T$ containing all
vertices $x$ for which the unique $x$-$v$-path in $T$ contains $w$. \\
We will now show that there is a vertex $w \in V(T)$, a nonempty terminal set $I^\prime \subseteq I \cap V(T_w)$ and a subtree $T^\prime $ of $T_w$ such that
\begin{enumerate}[(I)]
\item $(w,I^\prime) \in N$,
\item $\smtvertex(w,I^\prime) = c(T^\prime) = \smt(\{w\} \cup I')$,
\item $T^\prime$ is a subtree of $T_w$ containing $I^\prime \cup \{w\}$,
\item $T \graphminus T^\prime$ is a tree containing $(I \setminus I^\prime) \cup \{v,w\}$.
\end{enumerate}
Here, by $T \graphminus T^\prime$, we refer to the graph $\left((V(T) \setminus V(T^\prime)) \cup \{w\}, E(T) \setminus E(T^\prime)\right)$.
Assuming we have a triple $(w, I^\prime, T^\prime)$ satisfying (I) -- (IV), $\smtvertex(v,I) \leq \smt(\{v\} \cup I)$ can easily be proved:
Since $\bigL$ is a valid lower bound, we have
\begin{equation}
\begin{aligned}\label{dijkstrasteiner:proof:secondineq}
 \bigL(w, \term \setminus I^\prime) &\leq \bigL(v, \term \setminus I) + \smt((I \setminus I^\prime) \cup \{v,w\})
 \\ &\leq  \bigL(v, \term \setminus I)  + c(T \graphminus T^\prime)
 \\ &=  \bigL(v, \term \setminus I)  + c(T) - c(T^\prime).
\end{aligned}
\end{equation}
Adding (II) and (\ref{dijkstrasteiner:proof:secondineq}) yields
\begin{align*}
 \smtvertex(w, I^\prime) +  \bigL(w, \term \setminus I^\prime) \leq  c(T) + \bigL(v, \term \setminus I).
\end{align*}
By (I) and the choice of $(v,I)$ in line\nobreakspace \ref {dijkstrasteiner:firstloop} we have
\begin{align*}
 \smtvertex(v, I) +  \bigL(v, \term \setminus I) \leq  \smtvertex(w, I^\prime) + \bigL(w, \term \setminus I^\prime),
\end{align*}
so
\begin{align*}
 \smtvertex(v, I) \leq  c(T).
\end{align*}
It remains to be shown that there is such a triple $(w, I^\prime, T^\prime)$. We call $w \in V(T)$ \emph{proper} if $(w, I \cap V(T_w)) \in P$ before the execution of
line\nobreakspace \ref {dijkstrasteiner:endloop}.
If we have a leaf $w \in V(T) \setminus \{v\}$ which is not proper, we set $I^\prime = \{w\}$ and $T^\prime = (\{w\},\emptyset)$, satisfying (I) -- (IV).
Note that here (I) follows from  (\ref{dijkstrasteiner:proof:correct:c:3}).
Otherwise, since $v$ is not proper, there is a vertex $w$ in $T$ which is not proper but all neighbors $w_i$ of $w$ in $T_w$ are proper.
We define
\begin{align*}
\candset = \begin{cases}
  \{ I \cap V(T_{w_i}) \text{ $\vert$ $w_i$ is neighbor of $w$ in $T_w$} \} \cup \{\{w\}\} & \text{if $w \in I$,}\\
  \{ I \cap V(T_{w_i}) \text{ $\vert$ $w_i$ is neighbor of $w$ in $T_w$} \} & \text{if $w \notin I$.}\\
           \end{cases}
\end{align*}

Note that $\candset$ is a partition of $I \cap V(T_w)$.
Let $\candsubset$ be an inclusion-wise minimal subset of $\candset$ such that $\left(w, \bigcup_{\candelem \in \candsubset} \candelem \right) \notin P$.
Since $w$ is not proper and $(w, \emptyset) \in P$, $\candsubset$ exists and $\candsubset$ is not empty. For $\candelem \in \candsubset$, let $T_{\candelem}$ be the unique minimum connected subgraph of $T$
such that $(\{w\} \cup \candelem) \subseteq V(T_\candelem)$.
By optimality of $T$, $T_\candelem$ is an optimum Steiner tree for $(\{w\} \cup \candelem)$, as the only vertex in $T_\candelem$ which is incident to edges in $E(T) \setminus E(T_\candelem)$ is $w$.
We define
\begin{enumerate}[(i)]
 \item $I^\prime = \bigcup_{\candelem \in \candsubset} \candelem$,
 \item $V(T^\prime) = \bigcup_{\candelem \in \candsubset} V(T_{\candelem})$ and
 \item $E(T^\prime) =  \bigcup_{\candelem \in \candsubset} E(T_{\candelem})$.
\end{enumerate}
See Figure\nobreakspace \ref{algoproofconfigimage} for an illustration of this setting.
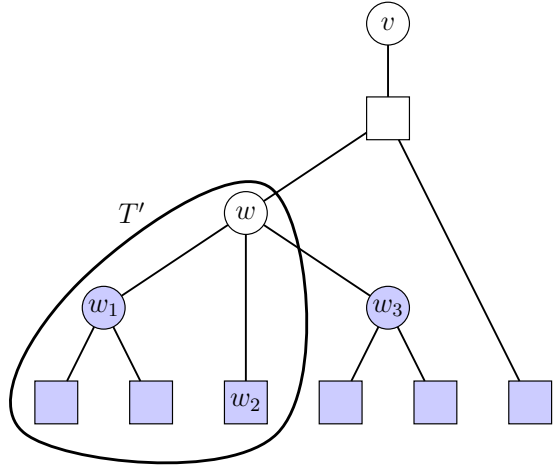
\begin{figure}[h]
\center
\scalebox{0.9}{
\begin{tikzpicture}[scale=0.7]
tikzstyle{every node}=[font=\large]
\def\propercolor{blue!20}
\SetVertexNormal[Shape=rectangle, FillColor=\propercolor]
\SetVertexNoLabel
\Vertex[x=0, y=0]{1}
\Vertex[x=2, y=0]{2}
\Vertex[x=4, y=0]{w2}
\Vertex[x=6, y=0]{3}
\Vertex[x=8, y=0]{4}
\Vertex[x=10, y=0]{5}
\SetVertexNormal[FillColor=\propercolor]
\Vertex[x=1, y=2]{w1}
\Vertex[x=7, y=2]{w3}
\SetVertexNormal
\Vertex[x=4, y=4]{w}
\SetVertexNormal[Shape=rectangle]
\Vertex[x=7, y=6]{6}
\SetVertexNormal
\Vertex[x=7, y=8]{v}
\node at (w1) {$w_1$};
\node at (w2) {$w_2$};
\node at (w3) {$w_3$};
\node at (w) {$w$};
\node at (v) {$v$};
\node (1sw) at (1.south west){};
\node (w1nw) at (w1.north west){};
\node (wne) at (w.north east){};
\node (w2se) at (w2.south east){};
\Edge(1)(w1)
\Edge(2)(w1)
\Edge(3)(w3)
\Edge(4)(w3)
\Edge(w1)(w)
\Edge(w2)(w)
\Edge(w3)(w)
\Edge(w)(6)
\Edge(5)(6)
\Edge(v)(6)

\draw [very thick] plot [smooth cycle, tension=0.7]  coordinates {(1sw.south west) (w1nw.north west)  (wne.north east)  (w2se.south east)};

\node at  (1.6,4) {$T'$};

\end{tikzpicture}}
\caption{A possible configuration with $\candsubset = \{ I \cap V(T_{w_1}), I \cap V(T_{w_2})\}$. Vertices in $I$ are drawn as squares, proper vertices are drawn in blue.}
\label{algoproofconfigimage}
\end{figure}

The conditions (III) and (IV) are satisfied by construction. We now show (II), which due do (\ref{dijkstrasteiner:proof:correct:c:3})
also implies (I).
First, we deal with the case $|\candsubset| = 1$.
If $\candsubset = \{\{w\}\}$, by (\ref{dijkstrasteiner:proof:correct:c:2}) clearly $\smtvertex(w, I^\prime) = 0 = c(T^\prime).$
If $\candsubset = \{ I \cap V(T_{w_i}) \}$ for some neighbor $w_i$ of $w$ in $T_w$, by\nobreakspace (\ref{dijkstrasteiner:proof:correct:c:2}), (\ref{dijkstrasteiner:proof:secondinv}), and the fact that $w_i$ is proper, we have
\begin{align*}
\smtvertex(w, I^\prime)
&\leq \smtvertex(w_i, I^\prime) + c(\{w_i, w\}) \\
& = \smt(\{w_i\} \cup I^\prime) + c(\{w_i, w\}) \\
& = c(T_{w_i}) + c(\{w_i, w\}) \\
& = c(T^\prime).
\end{align*}
Otherwise, i.e., $|\candsubset| > 1$, choose an element $\candelem \in \candsubset$
and
see again by\nobreakspace (\ref{dijkstrasteiner:proof:correct:c:2}), (\ref{dijkstrasteiner:proof:secondinv}) and the minimality of $\candsubset$ that
\begin{align*}
 \smtvertex(w, I^\prime)
 &\leq  \smtvertex(w, I^\prime \setminus \candelem) + \smtvertex(w, \candelem )
 \\ &= \smt(\{w\} \cup (I^\prime \setminus \candelem)) + \smt(\{w\} \cup \candelem)
 \\ &= c(T^\prime \graphminus T_{\candelem}) + c(T_{\candelem})
 \\ &= c(T^\prime).
\end{align*}
As $T$ is an optimum Steiner tree for $\{v\} \cup I$, by (IV) we get $c(T') = \smt(\{w\} \cup I')$, together with  (\ref{dijkstrasteiner:proof:correct:c:1}) showing (II).
By (\ref{dijkstrasteiner:proof:correct:c:3}), $(w,I') \in N$, so (I) holds as well, completing the proof that (\ref{dijkstrasteiner:proof:secondinv}) is preserved by lines\nobreakspace  \ref {dijkstrasteiner:firstloop} to\nobreakspace  \ref {dijkstrasteiner:endloop}.\\
To see that (\ref{dijkstrasteiner:proof:fourthinv}) is maintained, note that applying the previous argument to $(\rootterm, \sources)$
instead of $(v,I)$ always yields a label $(w,I') \in N$, so $N$ is not empty. Also, whenever a label $(v,I)$ is inserted to $P$, it is removed from $N$.
By (\ref{dijkstrasteiner:proof:secondinv}), from then on $\smtvertex(v,I)$ cannot be changed, so $(v,I)$ is never inserted again into $N$.
\explicitqed
\end{proof}

Note that in line\nobreakspace \ref {dijkstrasteiner:firstloop}, one can actually choose any label $(v, I) \in N$  with $\smtvertex(v,I) + \bigL(v, \term \setminus I) \leq \smtvertex(w, J) + \bigL(w, \term \setminus J)$
for all $(w,J) \in (V(G) \times 2^I) \cap N$. This is a generalization of the choice as specified in the algorithm. However, in our implementation, we always choose a label minimizing $\smtvertex(v,I) + \bigL(v, \term \setminus I)$.

\begin{theorem}
The Dijkstra-Steiner algorithm can be implemented to run in  $\bigO(3^k n + 2^k (n \log n + m) + 2^k n f_{\bigL})$ time, where $n = |V(G)|$, $m = |E(G)|$, $k = |\term|$, and $f_{\bigL}$ is an upper bound on the time required to evaluate $\bigL$. \label{dijkstrasteiner:proof:runtime}\end{theorem}
\begin{proof}
Since $|P|$ increases in each iteration, we have at most $n 2^{k-1}$ iterations.
We use a Fibonacci heap \cite{fibheap} to store all labels $(v, I) \in N$, which allows updates in constant amortized time.
Since the heap contains at most $2^{k-1}n$ elements, each execution of  line\nobreakspace \ref {dijkstrasteiner:firstloop} takes
$\bigO(\log(2^{k-1}n)) = \bigO(k + \log(n))$ amortized time.
Line\nobreakspace \ref {dijkstrasteiner:edge} is executed at most $2^km$ times and each execution takes $O(1)$ amortized time.
Furthermore, there are exactly $3^{k-1}$ pairs $I, J \subseteq \sources$ with $I \cap J = \emptyset$, since every element in $\sources$
can either be contained in $I$, $J$ or $\sourcesalone \setminus (I \cup J)$, independently of the others. Thus, line\nobreakspace \ref {dijkstrasteiner:superset} -- \ref{dijkstrasteiner:endupdatesupersets} take $\bigO(3^{k-1}n)$ time.
By caching values of $\bigL$, we can achieve that we query $\bigL$ at most once for each label, resulting in an additional run time of $\bigO(2^k n f_{\bigL})$.
The run time of the backtracking clearly is dominated by the previous tasks, since \texttt{backtrack} is called at most $\bigO(kn)$ times and requires effort linear in the size of its output.
We get a total run time of
\begin{align*}
\bigO(2^kn(\log n + k) + 2^km + 3^kn + 2^k n f_{\bigL}) =  \bigO(3^k n + 2^k (n \log n + m) + 2^k n f_{\bigL}),
\end{align*}
since $2^kk = \bigO(3^k)$.
\explicitqed
\end{proof}
In Section \ref{seclb}, we will see that there are non-trivial valid lower bounds $\bigL$ which can be used in the algorithm while still achieving a worst-case run time of $\bigO(3^k n + 2^k (n \log n + m))$, matching the bound of \cite{veinott}.

\section{Lower Bounds}
\label{seclb}
Roughly speaking, the larger the valid lower bound $\bigL$ is, the faster our algorithm will be. Before we describe examples of valid lower bounds, we note:
\begin{proposition}\label{lbcombprop}
 Let $\bigL$ and $\bigL^\prime$ be valid lower bounds. Then, $\max(\bigL, \bigL^\prime)$ also is a valid lower bound.
\end{proposition}
\begin{proof}
  Let $v,w \in V(G)$ and $I^\prime \subseteq I \subseteq \term$. Then
\begin{align*}
\bigL(v,I) &\leq \bigL(w,I^\prime) + \smt((I \setminus I^\prime) \cup\{v,w\}) \\
&\leq \max(\bigL(w,I^\prime), \bigL^\prime(w,I^\prime)) + \smt((I \setminus I^\prime) \cup\{v,w\})
\end{align*}
and
\begin{align*}
\bigL^\prime(v,I) &\leq \bigL^\prime(w,I^\prime) + \smt((I \setminus I^\prime) \cup\{v,w\}) \\
&\leq \max(\bigL(w,I^\prime), \bigL^\prime(w,I^\prime)) + \smt((I \setminus I^\prime) \cup\{v,w\}),
\end{align*}
so
\begin{align*}
\max(\bigL(v,I), \bigL^\prime(v,I)) &\leq \max(\bigL(w,I^\prime), \bigL^\prime(w,I^\prime)) + \smt((I \setminus I^\prime) \cup\{v,w\}).
\end{align*}
\explicitqed
\end{proof}

Proposition\nobreakspace \ref{lbcombprop} allows the combination of arbitrary valid lower bounds.

Now, we present three types of valid lower bounds.
A simple valid lower bound can be obtained by considering terminal sets of bounded cardinality:

\begin{definition}
Let $(G, c, \term)$ be an instance of the Steiner tree problem, $\rootterm \in \term$ and let $j \geq 1$ be an integer. Then, $\jplustterm{j}$ is defined as
\begin{align*}
\jplustterm{j}(v,I) = \max_{\{\rootterm\} \subseteq J \subseteq I \cup \{v\}, |J| \leq j+1}\smt(J)
\end{align*}
for $v \in V(G)$ and $\{\rootterm\} \subseteq I \subseteq \term$. For $v \in V(G)$ and $I \subseteq \term$ with $\rootterm \notin I$, set $\jplustterm{j}(v,I) = 0$.
\end{definition}

\begin{lemma}
Let $(G, c, \term)$ be an instance of the Steiner tree problem, $\rootterm \in \term$ and let $j \geq 1$ be an integer. Then, $\jplustterm{j}$ is a valid lower bound. Furthermore, we can implement $\jplustterm{j}$ such
that after a preprocessing time of $\bigO(3^k + (2k)^{j-1}n + k^{j-1}(n \log n + m))$, we can evaluate $\jplustterm{j}(v,I)$ for every
$v \in V(G)$ and $\{\rootterm\} \subseteq I \subseteq \term$ in time $\bigO({|I|}^{j-1})$, where $n = |V(G)|$, $m = |E(G)|$ and $k = |\term|$.
\end{lemma}

\begin{proof}
Let $v,w \in V(G)$ and $\{\rootterm\} \subseteq I^\prime \subseteq I \subseteq \term$.
To prove that $\jplustterm{j}$ is a valid lower bound, we have to show
\begin{align*}
 \max_{\substack{\{\rootterm\} \subseteq J \subseteq I \cup \{v\}\\ |J| \leq j + 1}}\smt(J) &\leq  \max_{\substack{\{\rootterm\} \subseteq J^\prime \subseteq I^\prime \cup \{w\} \\ |J^\prime| \leq j + 1}}\smt(J^\prime)
 &+ \smt((I \setminus I^\prime) \cup\{v,w\}).
\end{align*}
Consider the map $f : I \cup  \{v\} \rightarrow I^\prime \cup \{w\}$ with $f(x) = w$ for $x \notin I^\prime$ and $f(x) = x$ otherwise.
Let $J$ be a set with $\{\rootterm\} \subseteq J \subseteq I \cup \{v\}$ and $|J| \leq j + 1$. Set $J^\prime = \{f(x) : x \in J\}$. Then, clearly
$\{\rootterm\} \subseteq J^\prime \subseteq I^\prime \cup \{w\}$ and $|J^\prime| \leq |J| \leq j + 1$. Moreover,
\begin{align*}
 \smt(J)
 \leq \smt(J^\prime) + \smt((I \setminus I^\prime) \cup\{v,w\}).
\end{align*}
To achieve the given run time, we first compute $\smt(\{v\} \cup I)$ for all $v \in V(G)$ and $I \subseteq \term$ with $|I \setminus \{\rootterm\}| \leq j - 1$ in time $\bigO((2k)^{j-1}n + k^{j-1}(n \log n + m))$
using a modified variant of the Dijkstra-Steiner algorithm. More precisely, we do not use a lower bound, do not use a root terminal and consider terminal sets of increasing cardinality, very
similar to \cite{veinott}. There are $\bigO(k^{j-1})$ sets $I$ with $I \subseteq \term, |I \setminus \{\rootterm\}| \leq j - 1$. Since we consider terminal sets of increasing cardinality one after another,
we always have at most $n$ labels in the Fibonacci heap. As a set of cardinality $j$ has $2^j$ subsets,
there are $\bigO((2k)^{j-1}n)$ updates of supersets. \\
Then, to evaluate $\jplustterm{j}(v,I)$, we exploit
\begin{align*}
 \jplustterm{j}(v,I) &= \max_{\{\rootterm\} \subseteq J \subseteq I \cup \{v\}, |J| \leq j+1}\smt(J) \\
 &= \max \biggl( \max_{J \subseteq I, |J| \leq j-1}\smt(J \cup \{v,\rootterm\}), \max_{\{\rootterm\} \subseteq J \subseteq I, |J| \leq j+1}\smt(J) \biggr),
\end{align*}
where the first expression can be computed in $\bigO(|I|^{j-1})$ time using the precomputed values $\smt(\{v, \rootterm\} \cup J)$ and the second expression does not depend on $v$ and can be computed in advance for every $I \subseteq \term$ in $\bigO(3^k)$ time.
\explicitqed
\end{proof}

Thus, for small $j$, e.g., $j \leq 3$, $\jplustterm{j}$ can be efficiently computed. In experiments without pruning, this lower bound is useful on low-dimensional instances like planar rectilinear grid graphs. However, pruning has a much
larger impact and eliminates this effect.

We now present a more effective lower bound.
We will use the notion of \emph{1-trees}, which have long been studied \cite{onetree} as a lower bound for the traveling salesman problem.
Given a complete graph $H$ with metric edge costs and a special vertex $v \in V(H)$, a 1-tree for $v$ and $H$ is a spanning tree on $H - v$ together with two additional edges connecting $v$ with the tree.
A tour in $H$ is a connected 2-regular subgraph of $H$.
Since every tour consists of a path, which is a spanning tree, and a vertex connected to the endpoints of the path, every tour is a 1-tree.
Thus, a 1-tree of minimum cost is a lower bound on the cost of a tour of minimum cost. Since such a tour of minimum cost is at most twice as expensive as a minimum Steiner tree,
we can use 1-trees to get a lower bound for the Steiner tree problem.

For $v,w \in V(G)$, we denote by $\dist(v,w)$ the cost of a shortest path connecting $v$ and $w$ in $G$. Furthermore, for a set of vertices $X \subseteq V$, we
denote by $\distgraph{X}$ the \emph{distance graph} of $X$,
which is the subgraph of the metric closure of $G$ induced by $X$. Note that since the edge costs in $\distgraph{X}$ are the costs of shortest paths with respect to positive edge costs in $G$,
the edge costs in $\distgraph{X}$ are always metric.
Moreover, for $X \subseteq V(G)$, we denote by $\mst(X)$ the cost of a minimum spanning tree in $\distgraph{X}$.

\begin{definition}
Let $(G, c, \term)$ be an instance of the Steiner tree problem and $\rootterm \in \term$. Then, the \emph{1-tree bound $\onetree$} is defined as
\begin{align*}
\onetree(v,I) =\min_{\substack{i, j \in I : i \neq j \vee |I| = 1}} \frac{d(v,i) + d(v,j)}{2} + \frac{\mst(I)}{2}
\end{align*}
for $v \in V(G)$ and $\{\rootterm\} \subseteq I \subseteq \term$. For $v \in V(G)$ and $ I \subseteq \term \setminus \{\rootterm\}$, we set $\onetree(v,I) = 0$.
\end{definition}
\begin{lemma} \label{1treefeasible}
Let $(G, c, \term)$ be an instance of the Steiner tree problem and $\rootterm \in \term$. Then, $\onetree$ is a valid lower bound.
\end{lemma}

\begin{proof}
Let $v,w \in V(G)$ and $\{\rootterm\} \subseteq I^\prime \subseteq I \subseteq \term$. We will show
\begin{align*}
  2\onetree(v,I) \leq 2\onetree(w,I^\prime) + 2\smt((I \setminus I^\prime) \cup\{v,w\}),
\end{align*}
which translates to
\begin{align*}
 &\min_{i, j \in I: i \neq j \vee |I| = 1} (d(v,i) + d(v,j)) + \mst(I) \\
 &\leq \min_{i, j \in I': i \neq j \vee |I'| = 1}(d(w,i) + d(w,j)) + \mst(I') + 2\smt((I \setminus I^\prime) \cup\{v,w\}).
\end{align*}

Consider a minimum spanning tree $T_1$ in $\distgraph{I'}$ and a  Steiner tree $T_2$ for $(I \setminus I^\prime) \cup\{v,w\}$.
Furthermore, let $j_1,j_2 \in I'$ with $j_1 \neq j_2 \vee |I'| = 1$. This setting is illustrated in Figure\nobreakspace \ref {onetree:proof:firstimage}. \\
First, we construct a tour $C$ in $\distgraph{(I \setminus I^\prime) \cup\{v,w\}}$ of cost at most $2c(T_2)$ using the standard double tree argument: If we double each edge in $T_2$, the graph gets
Eulerian and we can find a Eulerian cycle. If we visit the vertices in the order the Eulerian cycle visits them and skip already visited vertices, we obtain a tour of at most the same cost exploiting that the edge costs in the distance graph are metric.\\
We can decompose the tour into two paths $P_1$ and $P_2$ in $\distgraph{I \cup \{v,w\}}$ with endpoints $v$ and $w$ such
that $(I \setminus I^\prime) \cup\{v,w\} = V(P_1) \cup V(P_2)$ and $c(P_1) + c(P_2) = c(C)$. Now, for $i \in \{1,2\}$, we define $P_i' = (V(P_i) \cup \{j_i\}, E(P_i) \cup \{\{w,j_i\}\})$,
which is the path obtained by appending the edge $\{w, j_i\}$ to $P_i$. \\
If $I_1 := (I\setminus I') \cap V(P_1)$ is empty, let $j_1'$ be $j_1$, else, let
$j_1'$ be the first terminal in $I_1$ when traversing $P_1$ from $v$ to $w$. \\
Similarly, if $I_2 := (I\setminus I') \setminus V(P_1) \subseteq V(P_2)$ is empty, set $j_2' = j_2$, else let $j_2'$ be the first terminal in $I_2$ when traversing $P_2$ from $v$ to $w$. \\
Since $I_1$ and $I_2$ are disjoint and do not contain $j_1, j_2 \in I'$, we have
\begin{align*}
j_1' = j_2' \implies  (j_1 = j_1' = j_2' = j_2 \wedge I\setminus I' = \emptyset),
\end{align*}
which together with
\begin{align*}
j_1 = j_2 \implies |I'| = 1
\end{align*}
implies that
\begin{align*}
j_1' = j_2' \implies |I| = 1.
\end{align*}
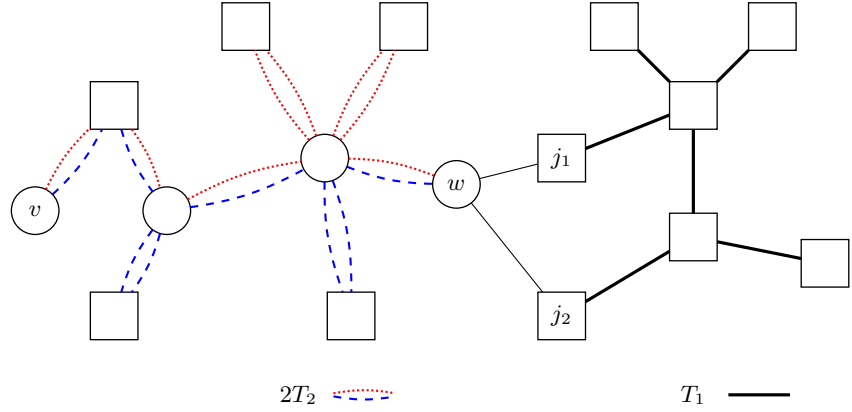
\begin{figure}[h]
\center
\scalebox{1}{
\begin{tikzpicture}[scale=0.7]

\def\firstcolor{red}
\def\secondcolor{blue}

\SetVertexNoLabel
\Vertex[x=0,y=0]{v}
\Vertex[x=2.5, y=0]{2}
\Vertex[x=5.5, y=1]{5}
\Vertex[x=8, y=0.5]{w}

\SetVertexNormal[Shape=rectangle]
\Vertex[x=6, y=-2]{7}
\Vertex[x=1.5, y=2]{1}
\Vertex[x=1.5, y=-2]{3}
\Vertex[x=4, y=3.5]{4}
\Vertex[x=7, y=3.5]{9}
\Vertex[x=10, y=1]{j1}
\Vertex[x=10, y=-2]{j2}
\Vertex[x=11, y=3.5]{10}
\Vertex[x=14, y=3.5]{14}
\Vertex[x=12.5, y=-0.5]{12}
\Vertex[x=15, y=-1]{13}
\Vertex[x=12.5, y=2]{11}

\node at (v) {$v$};
\node at (w) {$w$};
\node at (j1) {$j_1$};
\node at (j2) {$j_2$};
\node (leftcaption) at (5, -3.5) {2$T_2$};
\node (leftcaptionr) at (5.5, -3.5) {};
\node (leftcaptionrr) at (7, -3.5) {};
\doubleEdgeColor{leftcaptionr}{leftcaptionrr}{\firstcolor}{densely dotted}{\secondcolor}{dashed}

\doubleEdgeColor{v}{1}{\firstcolor}{densely dotted}{\secondcolor}{dashed}
\doubleEdgeColor{1}{2}{\firstcolor}{densely dotted}{\secondcolor}{dashed}
\doubleEdgeColor{2}{3}{\secondcolor}{dashed}{\secondcolor}{dashed}
\doubleEdgeColor{2}{5}{\firstcolor}{densely dotted}{\secondcolor}{dashed}
\doubleEdgeColor{5}{4}{\firstcolor}{densely dotted}{\firstcolor}{densely dotted}
\doubleEdgeColor{5}{9}{\firstcolor}{densely dotted}{\firstcolor}{densely dotted}
\doubleEdgeColor{5}{7}{\secondcolor}{dashed}{\secondcolor}{dashed}
\doubleEdgeColor{5}{w}{\firstcolor}{densely dotted}{\secondcolor}{dashed}


\tikzset{EdgeStyle/.style = {thin}}
\Edge(w)(j1)
\Edge(w)(j2)
\tikzset{EdgeStyle/.style = {very thick}}
\Edge(j1)(11)
\Edge(10)(11)
\Edge(11)(14)
\Edge(11)(12)
\Edge(13)(12)
\Edge(j2)(12)

\node (rightcaption) at (12.5,-3.5) {$T_1$};
\node (rightcaptionr) at (13,-3.5) {};
\node (rightcaptionrr) at (14.5,-3.5) {};
\Edge (rightcaptionr)(rightcaptionrr)

\end{tikzpicture}}
\caption{The minimum spanning tree $T_1$, the double Steiner tree $2T_2$ and the edges $\{w,j_1\}$ and $\{w, j_2\}$. Edges in $2 T_2$ contributing to $P_1$
are colored red and edges contributing to $P_2$ are colored blue. Vertices in $I$ are drawn as squares.}
\label{onetree:proof:firstimage}
\end{figure}

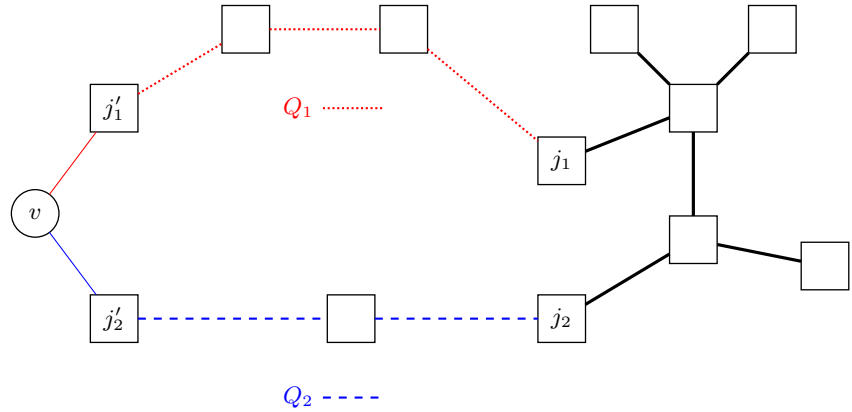
\begin{figure}[h]
\center
\scalebox{1}{
\begin{tikzpicture}[scale=0.7]

\def\firstcolor{red}
\def\secondcolor{blue}

\SetVertexNoLabel
\Vertex[x=0,y=0]{v}
\SetVertexNormal[Shape=rectangle]
\Vertex[x=6, y=-2]{7}
\Vertex[x=1.5, y=2]{1}
\Vertex[x=1.5, y=-2]{3}
\Vertex[x=4, y=3.5]{4}
\Vertex[x=7, y=3.5]{9}
\Vertex[x=10, y=1]{j1}
\Vertex[x=10, y=-2]{j2}
\Vertex[x=11, y=3.5]{10}
\Vertex[x=14, y=3.5]{14}
\Vertex[x=12.5, y=-0.5]{12}
\Vertex[x=15, y=-1]{13}
\Vertex[x=12.5, y=2]{11}

\node at (v) {$v$};
\node at (j1) {$j_1$};
\node at (j2) {$j_2$};
\node at (1) {$j_1'$};
\node at (3) {$j_2'$};
\node [color=\firstcolor] (q1) at (5, 2) {$Q_1$};
\node (q1r) at (5.3, 2) {};
\node (q1rr) at (6.8, 2) {};
\draw [thick, color=\firstcolor, densely dotted] (q1r) to (q1rr);
\node [color=\secondcolor] (q2) at (5, -3.5) {$Q_2$};
\node (q2r) at (5.3, -3.5) {};
\node (q2rr) at (6.8, -3.5) {};
\draw [thick, color=\secondcolor, dashed] (q2r) to (q2rr);

\tikzset{EdgeStyle/.style = {very thick}}
\Edge(j1)(11)
\Edge(10)(11)
\Edge(11)(14)
\Edge(11)(12)
\Edge(13)(12)
\Edge(j2)(12)

\tikzset{EdgeStyle/.style = {thin}}

\draw [thin, color=\firstcolor] (v) to (1);
\draw [thin, color=\secondcolor] (v) to (3);

\draw [thick, color=\firstcolor, densely dotted] (1) to (4) to (9) to (j1);
\draw [thick, color=\secondcolor, dashed] (3) to (7) to (j2);

\end{tikzpicture}}
\caption{$j_1'$, $j_2'$, $Q_1$ and $Q_2$ in the same setting as in Figure\nobreakspace \ref {onetree:proof:firstimage}.}
\label{onetree:proof:secondimage}
\end{figure}

Then, let $Q_i$ be the path in $\distgraph{I}$ obtained from
the subpath of $P_i'$ from $j_i'$ to $j_i$
by skipping $w$.
This is illustrated in Figure\nobreakspace \ref {onetree:proof:secondimage}.
We have
\begin{align*}
\dist(v,j_1') + \dist(v,j_2') + d(Q_1) + d(Q_2) \leq \dist(w,j_1) + \dist(w,j_2) + 2 c(T_2).
\end{align*}
Also, we have $(I \setminus I^\prime) \cup \{j_1, j_2\} \subseteq V(Q_1) \cup V(Q_2)$ and clearly, we can find a spanning tree in $\distgraph{I}$ with cost at most $d(Q_1) + d(Q_2) + d(T_1)$ by attaching $Q_1$ and
$Q_2$ to $T_1$. \\
Therefore,
\begin{align*}
 &\min_{i, j \in I: i \neq j \vee |I| = 1} (d(v,i) + d(v,j)) + \mst(I)
 \\&\leq \dist(v,j_1') + \dist(v,j_2') + d(Q_1) + d(Q_2) + d(T_1)
 \\&\leq \dist(w,j_1) + \dist(w,j_2) + d(T_1) + 2 c(T_2).
\end{align*}
\explicitqed
\end{proof}

Note that we could also use half the cost of a minimum spanning tree for $I \cup \{v\}$ to obtain a valid lower bound.
However, the 1-tree lower bound is always larger and hence leads to better run times.

\begin{lemma}
Let $(G, c, \term)$ be an instance of the Steiner tree problem and $\rootterm \in \term$. Then, we can implement $\onetree$ such
that after a preprocessing time of $\bigO(k(n \log n + m) + 2^kk^2)$, we can evaluate $\onetree$ for every
$v \in V(G)$ and $\{\rootterm\} \subseteq I \subseteq \term$ in time $\bigO({|I|})$, where $n = |V(G)|$, $m = |E(G)|$ and $k = |\term|$.
\end{lemma}

\begin{proof}
First, for every terminal $s \in \term$, we compute $\dist(v,s)$ for all $v \in V(G)$ in $\bigO(n \log n + m)$ time using Dijkstra's algorithm implemented with
a Fibonacci heap \cite{fibheap}.
For every $I \subseteq \term$, we compute $\mst(I)$ in $\bigO(|I|^2)$ time using Prim's algorithm \cite{prim}.
This results in a total preprocessing time of $\bigO(k(n \log n + m) + 2^kk^2)$.
Clearly,
\begin{align*}
 \min_{i, j \in I: i \neq j \vee |I| = 1} (\dist(v,i) + \dist(v,j))
\end{align*}
can be evaluated in $\bigO(|I|)$ time if $\dist(v,i)$ is known for all $v \in V(G)$ and $i \in I$.
\explicitqed
\end{proof}

Of course, in practice we do not compute minimum spanning trees for all sets of terminals in advance, but compute them dynamically when needed.

\begin{theorem}
Let $j \in \mathbb{N}$ be a constant.
Let $(G, c, \term)$ be an instance of the Steiner tree problem and $\rootterm \in \term$. Then, we can compute $\smt(\term)$ in time $\bigO(3^kn + 2^k(n \log n + m))$ using the Dijkstra-Steiner algorithm
with $\bigL = \max(\jplustterm{j}, \onetree)$.
\qed
\end{theorem}

The 1-tree lower bound exploits the fact that 1-trees can be used to compute lower bounds on the minimum cost of a tour, which in turn is at most twice as expensive as a minimum cost Steiner tree.
Using more preprocessing and evaluation time, we can eliminate the loss of approximating tours by 1-trees by using \emph{optimum tours} to get lower bounds.
While it may sound unreasonable to use optimum solutions for an \np-hard problem to speed up another algorithm, it turns out we can compute optimum tours for the union
of sets of terminals and at most one vertex quite fast if there are only few terminals. This is due to the fact that the cost of an optimum tour in $\distgraph{I \cup \{v\}}$ only depends on the distances from terminals to $v$ and shortest Hamiltonian paths with given endpoints in $\distgraph{I}$, which can be computed in advance.
For a set of vertices $X \subseteq V(G)$, we denote by $\tsp(X)$ the minimum cost of a Hamiltonian circuit in $\distgraph{X}$.
\begin{definition}
Let $(G, c, \term)$ be an instance of the Steiner tree problem and $\rootterm \in \term$. Then, the \emph{TSP bound $\tspbound$} is defined as
\begin{align*}
\tspbound(v,I) =\frac{\tsp(I \cup \{v\})}{2}
\end{align*}
for $v \in V(G)$ and $ I \subseteq \term$.
\end{definition}

\begin{lemma}
Let $(G, c, \term)$ be an instance of the Steiner tree problem and $\rootterm \in \term$. Then, $\tspbound$ is a valid lower bound. Moreover, after a preprocessing time of $\bigO(k(n \log n + m) + 2^kk^3)$, we can evaluate $\tspbound(v,I)$ in time
$\bigO(|I|^2)$ for all $v\in V(G)$ and $I \subseteq \term$.
\end{lemma}
\begin{proof}
Let $v,w \in V(G)$ and $\{\rootterm\} \subseteq I^\prime \subseteq I \subseteq \term$. We will show
\begin{align*}
  2\tspbound(v,I) \leq 2\tspbound(w,I^\prime) + 2\smt((I \setminus I^\prime) \cup\{v,w\}),
\end{align*}
which is equivalent to
\begin{align*}
  \tsp(I \cup \{v\}) \leq \tsp(I' \cup \{w\}) + 2\smt((I \setminus I^\prime) \cup\{v,w\}).
\end{align*}
First, we choose an optimal tour $C_1$ in $\distgraph{I' \cup \{w\}}$. Then, we construct a tour $C_2$ in $\distgraph{(I \setminus I^\prime) \cup\{v,w\}}$ of cost at most
$2\smt((I \setminus I^\prime) \cup\{v,w\})$ by doubling the edges of an optimum Steiner tree, finding a Eulerian walk and taking shortcuts. We have
$I \cup \{v\} = V(C_1) \cup V(C_2)$ and $w \in V(C_1) \cap V(C_2)$, so we can construct a tour in $\distgraph{I \cup \{v\}}$ by inserting $C_2$ into $C_1$ after $w$ and taking shortcuts,
which results in a tour of cost of at most
\begin{align*}
\tsp(I' \cup \{w\}) + 2\smt((I \setminus I^\prime) \cup\{v,w\}).
\end{align*}
We achieve the given run time using a dynamic programming approach very similar
to the TSP algorithm by Held and Karp \cite{heldkarptsp}. The idea is to compute shortest Hamiltonian paths in the distance graph of the terminals for all possible pairs of endpoints.
Then, one can evaluate $\tspbound(v,I)$ in $\bigO(|I|^2)$ time by enumerating all possible pairs of neighbors of $v$ in the tour.
\explicitqed
\end{proof}

\section{Pruning}
\label{secpruning}
In this section, we present techniques to speed up the algorithm further by discarding labels $(v,I)$ for which we can prove that they cannot contribute
to an optimum solution. This affects the number of iterations, since these labels then are not chosen in line\nobreakspace \ref {dijkstrasteiner:firstloop} of
the algorithm. Also, it speeds up the execution of line\nobreakspace \ref {dijkstrasteiner:superset}, since we only have to consider existing labels in the merge step.
First, we show how to identify labels that cannot contribute to an optimum solution. Then we show that we can indeed safely discard them in our algorithm.

\begin{definition} \label{subtreedef}
Let $(G, c, \term)$ be an instance of the Steiner tree problem, $(v,I) \in V(G) \times 2^\term$ and $T$ be a Steiner tree for $\term$. A tree $T_1$ is
said to be a \emph{$(v,I)$-subtree of $T$} if there exists a tree $T_2$ such that
\begin{enumerate}[(i)]
 \item $V(T_1) \cup V(T_2) = V(T)$,
 \item $V(T_1) \cap V(T_2) = \{v\}$,
 \item $T_1$ is a subtree of $T$ containing $\{v\} \cup I$ and
 \item $T_2$ is a subtree of $T$ containing $\{v\} \cup (\term \setminus I)$.
\end{enumerate}
\end{definition}
For a tree $T$ and a $(v,I)$-subtree $T_1$ of $T$, we will also refer to the corresponding subtree $T_2$ by $T - T_1$.

\begin{lemma} \label{boundprunelemma}
Let $(G, c, \term)$ be an instance of the Steiner tree problem and $\rootterm \in \term$. Let $\bigL$ be a valid lower bound and $U \geq \smt(\term)$.
Furthermore, let $v \in V(G)$, $I \subseteq \term \setminus \{\rootterm\}$ and $T_1$ be a tree in $G$ containing $\{v\} \cup I$ with
\begin{align*}
c(T_1) + \bigL(v, \term \setminus I) > U.
\end{align*}
Then, there is no optimum Steiner tree for $\term$ containing $T_1$ as a $(v,I)$-subtree.
\end{lemma}

\begin{proof}
Let $T$ be a Steiner tree for $\term$ such that $T_1$ is a $(v,I)$-subtree of $T$. Then,
\begin{align*}
c(T) &= c(T_1) + c(T \graphminus T_1) \\
       &\geq c(T_1) + \smt(\{v\} \cup (\term \setminus I)) \\
       &= c(T_1) + \smt(\{v,\rootterm\} \cup ((\term \setminus I) \setminus \{\rootterm\})) + \bigL(\rootterm,\{\rootterm\}) \\
       &\geq c(T_1) + \bigL(v,\term \setminus I) \\
       &> U \\
       &\geq \smt(\term).
\end{align*}
\explicitqed
\end{proof}

Lemma\nobreakspace \ref{boundprunelemma} is a trivial exploitation of the lower bound $\bigL$. Its effect on the run time of the algorithm is rather limited, since we only discard labels that would never have
been labeled permanently anyway. In contrast, the following lemma allows significant run time improvements of our algorithm, in particular on geometric instances. An application is illustrated in Figure\nobreakspace \ref{pruning:image}.

\begin{lemma} \label{setboundprunelemma}
Let $(G, c, \term)$ be an instance of the Steiner tree problem, $v \in V(G)$, $I \subset \term$ and $\emptyset \neq S \subseteq \term \setminus I$. Furthermore, let $T_1$ be a Steiner tree for
$\{v\} \cup I$ and $\subgraphname$ be a subgraph of $G$ with
\begin{enumerate}[(i)]
 \item $(I \cup S) \subseteq V(\subgraphname)$,
 \item each connected component of $\subgraphname$ contains a terminal in $S$ and
 \item $c(\subgraphname) < c(T_1)$.
\end{enumerate} Then, there is no optimum Steiner tree for $\term$ in $G$
containing $T_1$ as a $(v,I)$-subtree.
\end{lemma}

\begin{proof}
Let $T$ be a Steiner tree for $\term$ in $G$ containing $T_1$ as a $(v,I)$-subtree. Then, there exists a tree $T_2$ containing $\{v\} \cup (\term \setminus I)$ with $c(T) = c(T_1) + c(T_2)$. We construct a subgraph $T'$ of $G$ containing $\term$ by $T' = T_2 \graphplus \subgraphname$.
As $\subgraphname$ contains a path from every vertex in $\subgraphname$ to some vertex in $S$, $T_2$ is connected and $S \subseteq \term \setminus I \subseteq V(T_2)$, $T'$ is connected.
Thus,
\begin{align*}
 \smt(\term) &\leq c(T')
 \\ &\leq c(T_2) + c(\subgraphname)
 \\ &= c(T) - c(T_1) + c(\subgraphname)
 \\ &< c(T).
\end{align*}
\explicitqed
\end{proof}

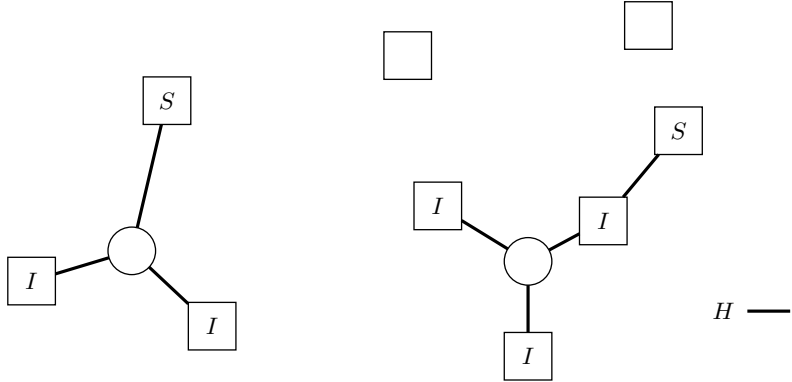
\begin{figure}[h]
\center
\scalebox{1}{
\begin{tikzpicture}[scale=2]

\def\firstcolor{red}
\def\secondcolor{blue}

\SetVertexNoLabel
\Vertex[x=-0.934,y=1.201]{10}
\Vertex[x=1.7,y=1.131]{11}
\SetVertexNormal[Shape=rectangle]
\Vertex[x=-1.6, y=1]{0}
\Vertex[x=-0.4, y=0.7]{1}
\Vertex[x=-0.7, y=2.2]{2}
\Vertex[x=1.1, y=1.5]{3}
\Vertex[x=1.7, y=0.5]{4}
\Vertex[x=2.2, y=1.4]{5}
\Vertex[x=2.7, y=2]{6}
\Vertex[x=2.5, y=2.7]{7}
\Vertex[x=0.9, y=2.5]{9}

\node at (0) {$I$};
\node at (1) {$I$};
\node at (2) {$S$};
\node at (3) {$I$};
\node at (4) {$I$};
\node at (5) {$I$};
\node at (6) {$S$};

\node (q1) at (3, 0.8) {$\subgraphname$};
\node (q1r) at (3.1, 0.8) {};
\node (q1rr) at (3.5, 0.8) {};
\draw [very thick] (q1r) to (q1rr);

\tikzset{EdgeStyle/.style = {very thick}}
\Edge(0)(10)
\Edge(1)(10)
\Edge(2)(10)
\Edge(3)(11)
\Edge(4)(11)
\Edge(5)(11)
\Edge(5)(6)

\tikzset{EdgeStyle/.style = {thin}}

\end{tikzpicture}}
\caption{By Lemma\nobreakspace \ref{setboundprunelemma}, no label for the set $I$ with cost strictly larger than $c(\subgraphname)$ can be part of an optimum solution. Terminals are drawn as squares, elements of $I$ and $S$ are labeled with the respective set.}
\label{pruning:image}
\end{figure}

In Section \ref{secresults}, we will explain how suitable graphs $\subgraphname$ can be found.
Lemmata\nobreakspace \ref {boundprunelemma} and\nobreakspace \ref{setboundprunelemma} allow us to identify labels that cannot contribute to an optimum solution.
Theorem\nobreakspace \ref {dijkstrasteiner:prunethm} shows that we can discard these labels without affecting the correctness of the algorithm. First, we prove an auxiliary lemma used in the proof of Theorem\nobreakspace \ref {dijkstrasteiner:prunethm}:

\begin{lemma} \label{replacelemma}
Let $(G, c, \term)$ be an instance of the Steiner tree problem and $\rootterm \in \term$. Let $O \subseteq V(G) \times 2^\sources$ be the set of pairs $(v,I)$ with the property that there is an optimum Steiner tree for $\term $ containing a $(v,I)$-subtree. Furthermore, let $(v,I) \in O$ and $T_1$ be a tree containing $\{v\} \cup I$ with $c(T_1) = \smt(\{v\} \cup I)$. Then, there
is an optimum Steiner tree $T$ for $\term$ containing $T_1$ as a $(v,I)$-subtree.
\end{lemma}
\begin{proof}
Since $(v,I) \in O$, there is an optimum Steiner tree $T'$ for $\term$ and a tree $T_1'$ which is a $(v,I)$-subtree of $T$. Set $T = (T' \graphminus T_1') \graphplus T_1$. Then,
since $T' \graphminus T_1'$ is a tree containing $\{v\} \cup (\term \setminus I)$ and $T_1$ is a tree containing $\{v\} \cup I$, $T$ is a connected subgraph of $G$ containing $\{v\} \cup \term$. Furthermore, we have
\begin{align*}
\smt(\term) \leq c(T) &\leq c(T' \graphminus T_1') + c(T_1) \\
&= c(T' \graphminus T_1') + \smt(\{v\} \cup I) \\
&= c(T') - c(T_1') + \smt(\{v\} \cup I) \\
&\leq c(T') \\
&= \smt(\term).
\end{align*}
Since we do not have edges of zero cost, this shows $T$ is an optimum Steiner tree and $T_1$ is a $(v,I)$-subtree of $T$.
\explicitqed
\end{proof}
We now formalize a general method of pruning:
\vspace{5pt}

\begin{procedure}[H]
\SetNlSty{phantom}{}{}
\newcommand{\mycapsty}[1]{\texttt{{#1}}}\SetAlCapNameSty{mycapsty}
\SetNlSty{textbf}{}{}   \setcounter{AlgoLine}{0}
\If{we can prove that there is no optimum Steiner tree $T$ for $\term$
  such that  {\tt backtrack}$(v,I)$ is the edge set of a $(v,I)$-subtree of $T$}
{
$N := N \setminus \{(v,I)\}$\;
}
\caption{{prune}($v, I$)}
\end{procedure}

\vspace{5pt}

Note that when considering a not permanently labeled element $(v,I) \in N$, we cannot guarantee that $\texttt{backtrack}(v,I)$ is the edge set of a tree, since it may contain cycles.
However, if it is not a tree, we can prune $(v,I)$ obviously.

\begin{theorem}
The Dijkstra-Steiner algorithm still works correctly if we modify it to execute \textup{prune($w, I)$} after line \ref{dijkstrasteiner:afterupdate:neighbor} and \textup{prune($v, I \cup J$)} after line \ref{dijkstrasteiner:afterupdate:superset}.
\label{dijkstrasteiner:prunethm}
\end{theorem}

\begin{proof}
Let $O \subseteq V(G) \times 2^\sources$ be the set of pairs $(v,I)$ with the property that there is an optimum Steiner tree  for $\term $ containing a $(v,I)$-subtree.
It suffices that the modified algorithm is correct on $O$, which we will now prove.\\
To this end, we modify invariants (\ref{dijkstrasteiner:proof:firstinv}) -- (\ref{dijkstrasteiner:proof:lastinv}) as defined in the proof of Theorem
\ref{dijkstrasteiner:proof:correct} by restricting (\ref{dijkstrasteiner:proof:firstinv}) -- (\ref{dijkstrasteiner:proof:thirdinv}) to labels
$(v,I) \in O$ and not changing (\ref{dijkstrasteiner:proof:fourthinv}).
We call these new invariants (\ref{dijkstrasteiner:proof:firstinv}') -- (\ref{dijkstrasteiner:proof:lastinv}'). \\
Since $(\rootterm,\sources) \in O$, the algorithm is correct assuming that these invariants hold. They clearly hold after the initialization.\\
Moreover, (\ref{dijkstrasteiner:proof:firstinv}'), (\ref{dijkstrasteiner:proof:correct:c:1}') and (\ref{dijkstrasteiner:proof:correct:c:2}') are clearly preserved. To see that (\ref{dijkstrasteiner:proof:correct:c:3}') is preserved, recall that prune only removes labels $(v,I)$ from $N$ if there is no
optimum Steiner tree $T$ for $\term$ such that  $\texttt{backtrack}(v,I)$ is the edge set of a $(v,I)$-subtree of $T$.
However, if $(v,I) \in O$ and $\smtvertex(v,I) = \smt(\{v\} \cup I)$, by (\ref{dijkstrasteiner:proof:correct:a:3}') and Lemma \ref{replacelemma} the label $(v,I)$ cannot be pruned, so (\ref{dijkstrasteiner:proof:correct:c:3}') is preserved as well.\\
The argument showing that (\ref{dijkstrasteiner:proof:secondinv}) is preserved remains unchanged: (\ref{dijkstrasteiner:proof:thirdinv}') can be applied to $(w,I')$, because if $(v,I) \in O$, then $(w, I')$ is in $O$ as well. This directly implies that the same argument as in the proof of Theorem \ref{dijkstrasteiner:proof:correct} can be used to prove that (\ref{dijkstrasteiner:proof:fourthinv}') is preserved, since $(\rootterm, \sources) \in O$.
\explicitqed
\end{proof}

Of course, in practice we just avoid the creation of such labels instead of removing them immediately after creation.
Moreover, whenever a label $(v,I)$ is selected in line \ref {dijkstrasteiner:firstloop}, we also try to prune it. This is not redundant, since after $\smtvertex(v,I)$ was updated the last time, bounds used to prune may have improved. Moreover, we could have pruned $(v,I)$ immediately after the last update of $\smtvertex(v,I)$ with an equivalent impact on the execution of the algorithm, so the algorithm still works correctly.

\section{Implementation and Results}
\label{secresults}
We implemented the algorithm using the C++ programming language.
In our implementation, we use a binary heap instead of a Fibonacci heap.
We represent terminal sets by bitsets using the canonical
bijection $2^\term \rightarrow \{0, \ldots, 2^{|\term|}-1\}$. For each vertex $v \in V(G)$, we maintain an array containing the labels $(v,I)$
with $\smtvertex(v,I) < \infty$ and a hash table storing for each label its index in the array, if it exists. This enables us to access labels
very quickly and traverse over the existing labels in linear time, which is important for an efficient implementation of the merge step:

To implement line\nobreakspace \ref {dijkstrasteiner:superset}, we have two options. Either we explicitly enumerate all sets $J \subseteq \sourcesalone \setminus I$ and check whether the label $(v,J)$ exists, or we traverse over all existing labels at $v$
and omit the labels $(v,J)$ with $J \cap I \neq \emptyset$. We always choose the option resulting in fewer sets to be considered.

We implement the pruning rule of Lemma\nobreakspace \ref {boundprunelemma} using a shortest-paths Steiner tree heuristic similar to Prim's algorithm \cite{prim}, maintaining and extending one component at a time. This takes $\bigO(k(n \log n + m))$ time. Then, we use the cost of that Steiner tree as an upper bound and apply Lemma\nobreakspace \ref {boundprunelemma} each time we create a new label.

To implement Lemma\nobreakspace \ref {setboundprunelemma}, we maintain an upper bound $U(I)$ on the cost of labels for each set $I \subseteq \sources$ of terminals, which is initially set to infinity. For each occurring set $I \subseteq \sources$, we compute the distance
$\dist(I, \term \setminus I) = \min_{x \in I, y \in \term \setminus I}\dist(x,y)$. Then, each time we extract a label $(v,I)$ from the heap, we update $U(I)$ by
\begin{align*}
U(I) := \min\left(U(I),\ \smtvertex(v,I) + \min(\dist(I, \term \setminus I), \dist(v, \term \setminus I))\right).
\end{align*}
Also, we keep track of the set $S(I)$ that was used to generate the currently best upper bound for the set $I$. In the routine described above, we always have $|S| = 1$. However, when merging two sets $I_1$ and $I_2$, we can use the sum of their upper bounds as an upper bound for the set $I_1 \cup I_2$ if $S(I_1) \cap I_2 = \emptyset$ or $S(I_2) \cap I_1 = \emptyset$, resulting in $S(I_1 \cup I_2) = (S(I_1) \cup S(I_2)) \setminus (I_1 \cup I_2)$.

Furthermore, we use the 1-tree bound as a lower bound. Of course, we do not compute minimum spanning trees for all subsets of terminals in advance. Instead, each time we consider a set $I$ for the first time, we compute $\mst(\term\setminus I)$.

Lacking a good selection strategy for general instances, we always choose the last terminal of the instance w.r.t. the order in the instance file as root terminal.
Note that our implementation is limited to instances with less than 64 terminals.

We now present computational results of our algorithm on instances of the 11th DIMACS implementation challenge \cite{dimacschallenge}.

\def\prec{2}
\setlength{\tabcolsep}{12pt}

\begin{table}[p]
\begin{tabular}{l r r r r r r}
Instance & $|V|$ & $|E|$ & $|\term|$ & Opt & Time [s] & Time PV [s] \\
\midrule \midrule \phantomsection \label{table::sub::vlsi}\\
\multicolumn{6}{l}{(a) VLSI-derived grid graphs with holes} & \\
diw0779    & 11821 & 22516   & 50    & 4440       & \numb{1.60282}       & \numbP{1.26} \\
diw0819    & 10553 & 20066   & 32    & 3399       & \numb{0.29050}       & \numbP{0.52} \\
diw0820    & 11749 & 22384   & 37    & 4167       & \numb{1.50418}       & \numbP{1.06} \\
lin23      & 3716  & 6750    & 52    & 17560      & \numb{11.07980}      & \numbP{0.54} \\
lin24      & 7998  & 14734   & 16    & 15076      & \numb{0.10387}       & \numbP{1.73} \\
lin30      & 19091 & 35644   & 31    & 27684      & \numb{0.68134}       & \numbP{14.74} \\
lin32      & 19112 & 35665   & 53    & 39832      & \numb{150.20080}     & \numbP{816.51} \\
lin34      & 38282 & 71521   & 34    & 45018      & \numb{12.92992}      & \numbP{1848.24} \\
lin35      & 38294 & 71533   & 45    & 50559      & \numb{26.90421}      & \numbP{1911.09} \\
lin36      & 38307 & 71546   & 58    & 55608      & \numb{47.58989}      &39931.77 \\ \\
\midrule \phantomsection \label{table::sub::obstacle}\\
\multicolumn{7}{l}{(b) Rectilinear obstacle-avoiding instances preprocessed by ObSteiner } \\
ind5       & 114   & 228     & 33    & 1341       & \numb{0.02059}      & \numbP{0.01} \\
rc03       & 109   & 202     & 50    & 54160      & \numb{0.15333}      & \numbP{0.00} \\
rt02       & 788   & 1938    & 50    & 45852      & \numb{0.70289}      & \numbP{1.99} \\
\\ \midrule \phantomsection \label{table::sub::group} \\
\multicolumn{6}{l}{(c) Group Steiner tree instances } & \\
wrp3-14    & 128   & 247     & 14    & 1400250    & \numb{3.77965}       & \numbP{0.01} \\
wrp3-15    & 138   & 257     & 15    & 1500422    & \numb{51.15627}      & \numbP{0.01} \\
wrp3-16    & 204   & 374     & 16    & 1600208    & \numb{11.62201}      & \numbP{0.03} \\
wrp3-17    & 177   & 354     & 17    & 1700442    & \numb{422.07924}     & \numbP{0.02} \\
wrp3-19    & 189   & 353     & 19    & 1900439    & \numb{1765.72262}    & \numbP{0.03} \\
\\ \midrule
\phantomsection \label{table::sub::inc} \\
\multicolumn{6}{l}{(d) Random graphs with so-called incidence costs } & \\
i160-141   & 160   & 2544    & 12    & 2549       & \numb{3.39813}     & \numbP{0.01}   \\
i320-111   & 320   & 1845    & 17    & 4273       & \numb{1706.67394}  & \numbP{0.03}   \\
i640-022   & 640   & 204480  & 9     & 1756       & \numb{4.03080}     & \numbP{0.52}   \\
i640-031   & 640   & 1280    & 9     & 3278       & \numb{0.06546}     & \numbP{0.00}   \\
i640-043   & 640   & 40896   & 9     & 1931       & \numb{0.88873}     & \numbP{0.13}   \\
\\ \midrule
\phantomsection \label{table::sub::puc} \\
\multicolumn{7}{l}{(e) Artificial instances from the hard PUC testset} \\
cc3-4p     & 64    & 288     & 8     & 2338       & \numb{0.00565}     & \numbP{1.99}   \\*
cc3-4u     & 64    & 288     & 8     & 23         & \numb{0.00819}     & \numbP{1.37}   \\*
cc3-5p     & 125   & 750     & 13    & 3661       & \numb{3.45007}     & \numbP{87.98}  \\
cc3-5u     & 125   & 750     & 13    & 36         & \numb{4.63731}     & \numbP{115.83} \\
cc6-2p     & 64    & 192     & 12    & 3271       & \numb{0.16628}     & \numbP{0.40}   \\*
cc6-2u     & 64    & 192     & 12    & 32         & \numb{0.27873}     & \numbP{0.90}   \\*

\end{tabular}
\caption{Results on various instance types.}
\label{sampletable}
\end{table}

Our results were achieved single-threaded on a computer with \SI{3.33}{\giga \hertz} Intel Xeon W5590 CPUs, which produced a score of 391 for the DIMACS benchmark.
We compare our results with those obtained by the state-of-the-art algorithm by Polzin and Vahdati Daneshmand \cite{polzindimacs}, which were obtained using one thread on a computer with a \SI{2.66}{\giga \hertz} Intel i7 920 CPU. This computer produced a score of 307 in the DIMACS benchmark.
The algorithm by Polzin and Vahdati Daneshmand successively performs various optimality-preserving reductions combined with
a branch and bound approach.

In Table\nobreakspace\ref{sampletable}, we show results on multiple instance classes \hyperref[table::sub::vlsi]{(a)} -- \hyperref[table::sub::puc]{(e)}.
For each instance, we give its name, the number of vertices, edges and terminals. Then, we state the cost of an optimum solution as reported by our algorithm and the run time in seconds.
Moreover, for each instance, we give the run time reported by Polzin and Vahdati Daneshmand. For the lin testset, Polzin and Vahdati Daneshmand improved run times by modifying their algorithm to use stronger reductions. With default settings, their algorithm did not solve lin36 within a time limit of 24 hours.

On instances \hyperref[table::sub::vlsi]{(a)} and \hyperref[table::sub::obstacle]{(b)}, both arising from rectilinear VLSI problems,
our algorithm is much faster than the worst-case bound tells.
This is primarily caused by the high impact of our pruning method.
In particular, on instances with large underlying graphs, our algorithm performs very well, beating the reduction-based solver.

In contrast, on group Steiner tree instances \hyperref[table::sub::group]{(c)}, our pruning implementation has no effect. Although these instances are based on VLSI-derived grid graphs with holes as well, they have been modified to
model the groups as terminals: For each group of the group Steiner tree instance, a new terminal is added to the graph and connected to the elements of the group by edges of very high cost.
By choosing the cost of these new edges sufficiently large, one can guarantee that an optimum Steiner tree in the new instance corresponds to an optimum group Steiner tree in the original instance and vice versa, since each terminal will be a leaf of the Steiner tree.
To prune a label $(v,I) \in V(G) \times 2^{\sources}$, our implementation needs to connect the terminal set $I$ to at least one additional terminal $s \in \term \setminus I$ with cost strictly less than the cost of the label $(v,I)$. However, on these instances, connecting to an additional terminal is always much more expensive than any path in the original graph.

On incidence cost instances \hyperref[table::sub::inc]{(d)}, where edges incident to terminals are assigned larger costs, a similar effect can be observed.
On other random graph instances, the 1-tree lower bound is very effective.

Although neither pruning nor future cost estimates do have a noticeable effect on instances from the hard PUC testset \hyperref[table::sub::puc]{(e)}, our algorithm performs very well on instances with few terminals.
This is caused by the strong worst-case run time guarantee, which, albeit exponential in the number of terminals, is quasilinear in the size of the graph.
Detailed further computational results can be found in Appendix \ref{appendix:graph}.

We also applied our algorithm to the $d$-dimensional rectilinear Steiner tree problem for $d \in \{3, 4, 5\}$. For $d = 2$, instances with thousands of terminals can be solved by the GeoSteiner algorithm \cite{warme2000exact}, which works by generating a candidate set of full Steiner trees for subsets of terminals and then concatenating a subset of these candidates to form an optimum Steiner tree. A full Steiner tree is a Steiner tree where each terminal has degree 1. The concatenation phase works by either solving a Steiner tree problem in graphs or a minimum spanning tree problem in hypergraphs.
Exploiting a result by Hwang \cite{hwang1976steiner}, the GeoSteiner algorithm only has to consider full Steiner trees following a special structure, which allows a significant reduction of the number of generated candidates.
In contrast, for higher dimensions, eliminating possible full Steiner trees is much harder, as Hwang's result does not apply and complicated full Steiner trees have to be considered \cite{mt_wulff}. An implementation of the GeoSteiner algorithm for higher dimension is only able to solve instances with up to around 15 terminals in dimension 3 and up to around 10 terminals in dimensions 4, 5 and 6 \cite{mt_wulff}. The situation is similar in the Euclidean case, where huge instances can be solved for $d=2$ using the GeoSteiner algorithm, but only instances with up to around 17 terminals for $d \geq 3$, using very different algorithms \cite{fonseca2014faster}.

It is well-known that to compute an optimum rectilinear Steiner tree in dimension 2, it suffices to compute an optimum Steiner tree in the so-called Hanan grid \cite{hanan1966steiner}.
The Hanan grid is the graph obtained by drawing axis-parallel lines through each terminal, taking intersections of these lines as vertex set and segments between intersections as edges.
This result was later generalized to arbitrary dimension by Snyder \cite{snyder1992exact}, leading to a grid graph with $\bigO(k^d)$ vertices and
$\bigO(d k^d)$ edges for a $d$-dimensional instance with $k$ terminals.

Using this reduction, we ran our algorithm on instances from the CARIOCA \cite{dimacschallenge} testset, which contains randomly generated instances for $d \in \{3, 4, 5\}$ with between 11 and 20 terminals.
To test our algorithm on larger instances, we generated new random instances (using the prefix ``bonn''), as other available testsets do not contain sufficiently many instances which are neither too small nor too large.
For these new instances, coordinates were chosen uniformly at random from $\{0, 1, \ldots, 999\}$.
In particular for the larger instances, coordinates may appear multiple times, leading to grid graphs with slightly less than $k^d$ vertices.
Coordinates of instances from the CARIOCA testset were scaled by $10^8$ to obtain integral coordinates.
For these instances, we chose a terminal as close as possible to the center of gravity of all terminals as root terminal, improving results compared to a random choice.

\begin{table}[H]
\begin{tabular}{l r r r r r r}
Instance & $d$ & $|V|$ & $|E|$ & $|\term|$ & Opt & Time [s] \\
   \midrule \midrule
carioca\_3\_11\_01 & 3 & 1331  & 3630    & 11    & 311221222  & \numb{0.02015}       \\
carioca\_3\_11\_02 & 3 & 1331  & 3630    & 11    & 466149453  & \numb{0.01888}       \\
carioca\_3\_20\_01 & 3 & 8000  & 22800   & 20    & 638376617  & \numb{1.61296}       \\
carioca\_3\_20\_02 & 3 & 8000  & 22800   & 20    & 477950448  & \numb{0.14770}       \\
bonn\_3\_40\_1 & 3 & 64000 & 187200  & 40    & 9024       & \numb{25.43975}      \\
bonn\_3\_40\_2 & 3 & 57798 & 168909  & 40    & 9633       & \numb{710.09162}     \\
bonn\_3\_55\_3 & 3 & 154548 & 455004 & 55    & 12138      & \numb{6201.10443}    \\
\midrule
carioca\_4\_11\_01 & 4 & 14641 & 53240   & 11    & 627022001  & \numb{0.42525}       \\
carioca\_4\_11\_02 & 4 & 14641 & 53240   & 11    & 636772154  & \numb{0.22314}       \\
carioca\_4\_20\_01 & 4 & 160000 & 608000  & 20    & 889180827  & \numb{82.72426}      \\
carioca\_4\_20\_02 & 4 & 160000 & 608000  & 20    & 822698792  & \numb{101.10654}      \\
\midrule
carioca\_5\_11\_01 & 5 & 161051 & 732050  & 11    & 925163690  & \numb{34.54747}      \\
carioca\_5\_11\_02 & 5 & 161051 & 732050  & 11    & 844673618  & \numb{13.01897}      \\
carioca\_5\_15\_01 & 5 & 759375 & 3543750 & 15    & 1011895745 & \numb{1046.36109}    \\
carioca\_5\_15\_02 & 5 & 759375 & 3543750 & 15    & 1067623193 & \numb{888.80817}    \\
carioca\_5\_18\_03 & 5 & 1889568 & 8922960 & 18    & 1177091608 & \numb{1081.01490}    \\
\end{tabular}
\caption{Results on $d$-dimensional rectilinear instances.}
\label{l1table}
\end{table}

Excerpts of experimental results are given in Table \ref{l1table}, the full results can be found in Appendix \ref{appendix:rect}.

In dimension 3,
we are able to solve all tested instances with up to 34 terminals. Many of the larger instances with up to 40 terminals are solved as well, and additionally one instance with 55 terminals.
In dimension 4, we solve all instances of the CARIOCA testset as well. Experiments with larger instances are not reported here, since we were only able to solve instances with up to 22 terminals. In dimension 5, all instances with up to 15 terminals are solved. The largest solved instance carioca\_5\_18\_03 has 18 terminals and nearly 9 million edges. Solving it required approximately \SI{20}{\giga \byte} of memory.

\section{Discussion}
Note that due to the dynamic programming nature of our algorithm, it can also be used to compute all optimum Steiner trees or even all Steiner trees up to a given cost. If we enumerate all Steiner trees up to a cost of $\text{Opt} + \Delta$, we have to relax
the pruning implementations by $\Delta$ and continue labeling until all labels $(v,I)$ with $\smtvertex(v,I) \leq \text{Opt} + \Delta$ are permanent. Also, we have to save all predecessors instead of only one optimum predecessor. Then, we can recursively combine Steiner trees for subsets of terminals. In practice, the additional effort is linear in the size of the output, allowing the enumeration of millions of near-optimum Steiner trees in seconds. See \cite{mt_silvanus} for details.

The dynamic programming idea by Dreyfus and Wagner has been used extensively to obtain Steiner tree algorithms with fast theoretical worst-case behavior. However, in the field of practical solving, it has rather been disregarded prior to this work.
Compared to other exact algorithms, our algorithm depends much less on effective preprocessing and performs well on large graphs. Our approach is very general and not limited to the lower bounds and pruning strategies proposed in this paper.


\bibliographystyle{spmpsci}      
\bibliography{steineralgo}   




\clearpage

\appendix


\section{Results on Graphic DIMACS Instances}\label{appendix:graph}

We present detailed computational results on DIMACS testsets. Our implementation is limited to instances with less than 64 terminals, so we exclude instances with more terminals.

The implementation of our algorithm is written in the C++ programming language and compiled using the GCC 4.8.2 compiler.

The experiments were performed single-threaded on a machine with \SI{3.33}{\giga \hertz} Intel Xeon W5590 CPUs and \SI{144}{\giga \byte} main memory which produced a score of 391.372724 using the DIMACS benchmark code.

For these experiments, we limited the memory consumption of the algorithm to \SI{100}{\giga \byte} and
set the time limit to \SI{7200}{\second}.
The reported run times do not include the time to read the instance file from disk.

\label{dimacsdata}
\def\prec{3}
\setlength{\tabcolsep}{12pt}

\def\extraspace{-8pt}

\testset{ALUE}{
alue2087   & 1244  & 1971    & 34    & 1049       & \numb{1.07535}       \\
alue2105   & 1220  & 1858    & 34    & 1032       & \numb{0.20420}       \\
alue7066   & 6405  & 10454   & 16    & 2256       & \numb{0.06756}       \\
alue7229   & 940   & 1474    & 34    & 824        & \numb{0.03575}       \\
}{\typeVlsi}

\testset{ALUT}{
alut0787   & 1160  & 2089    & 34    & 982        & \numb{1.76855}       \\
alut0805   & 966   & 1666    & 34    & 958        & \numb{0.25796}       \\
alut2764   & 387   & 626     & 34    & 640        & \numb{0.04196}       \\
}{\typeVlsi}

\testset{DIW}{
diw0234    & 5349  & 10086   & 25    & 1996       & \numb{0.09591}       \\
diw0250    & 353   & 608     & 11    & 350        & \numb{0.00198}       \\
diw0260    & 539   & 985     & 12    & 468        & \numb{0.00282}       \\
diw0313    & 468   & 822     & 14    & 397        & \numb{0.00283}       \\
diw0393    & 212   & 381     & 11    & 302        & \numb{0.00186}       \\
diw0445    & 1804  & 3311    & 33    & 1363       & \numb{0.06092}       \\
diw0459    & 3636  & 6789    & 25    & 1362       & \numb{0.04201}       \\
diw0460    & 339   & 579     & 13    & 345        & \numb{0.00364}       \\
diw0473    & 2213  & 4135    & 25    & 1098       & \numb{0.06537}       \\
diw0487    & 2414  & 4386    & 25    & 1424       & \numb{0.37314}       \\
diw0495    & 938   & 1655    & 10    & 616        & \numb{0.00648}       \\
diw0513    & 918   & 1684    & 10    & 604        & \numb{0.00647}       \\
diw0523    & 1080  & 2015    & 10    & 561        & \numb{0.00498}       \\
diw0540    & 286   & 465     & 10    & 374        & \numb{0.00175}       \\
diw0559    & 3738  & 7013    & 18    & 1570       & \numb{0.08933}       \\
diw0778    & 7231  & 13727   & 24    & 2173       & \numb{0.13894}       \\
diw0779    & 11821 & 22516   & 50    & 4440       & \numb{1.60282}       \\
diw0795    & 3221  & 5938    & 10    & 1550       & \numb{0.02380}       \\
diw0801    & 3023  & 5575    & 10    & 1587       & \numb{0.02423}       \\
diw0819    & 10553 & 20066   & 32    & 3399       & \numb{0.29050}       \\
diw0820    & 11749 & 22384   & 37    & 4167       & \numb{1.50418}       \\
}{\typeVlsi}

\testset{DMXA}{
dmxa0296   & 233   & 386     & 12    & 344        & \numb{0.00182}       \\
dmxa0368   & 2050  & 3676    & 18    & 1017       & \numb{0.01931}       \\
dmxa0454   & 1848  & 3286    & 16    & 914        & \numb{0.01237}       \\
dmxa0628   & 169   & 280     & 10    & 275        & \numb{0.00187}       \\
dmxa0734   & 663   & 1154    & 11    & 506        & \numb{0.00606}       \\
dmxa0848   & 499   & 861     & 16    & 594        & \numb{0.04061}       \\
dmxa0903   & 632   & 1087    & 10    & 580        & \numb{0.00836}       \\
dmxa1010   & 3983  & 7108    & 23    & 1488       & \numb{0.10421}       \\
dmxa1109   & 343   & 559     & 17    & 454        & \numb{0.00721}       \\
dmxa1200   & 770   & 1383    & 21    & 750        & \numb{0.06897}       \\
dmxa1304   & 298   & 503     & 10    & 311        & \numb{0.00206}       \\
dmxa1516   & 720   & 1269    & 11    & 508        & \numb{0.00333}       \\
dmxa1721   & 1005  & 1731    & 18    & 780        & \numb{0.01183}       \\
dmxa1801   & 2333  & 4137    & 17    & 1365       & \numb{0.04938}       \\
}{\typeVlsi}

\testset{GAP}{
gap1307    & 342   & 552     & 17    & 549        & \numb{0.05669}       \\
gap1413    & 541   & 906     & 10    & 457        & \numb{0.00655}       \\
gap1500    & 220   & 374     & 17    & 254        & \numb{0.00276}       \\
gap1810    & 429   & 702     & 17    & 482        & \numb{0.01304}       \\
gap1904    & 735   & 1256    & 21    & 763        & \numb{0.04172}       \\
gap2007    & 2039  & 3548    & 17    & 1104       & \numb{0.05121}       \\
gap2119    & 1724  & 2975    & 29    & 1244       & \numb{0.45977}       \\
gap2740    & 1196  & 2084    & 14    & 745        & \numb{0.01059}       \\
gap2800    & 386   & 653     & 12    & 386        & \numb{0.00262}       \\
gap2975    & 179   & 293     & 10    & 245        & \numb{0.00110}       \\
gap3036    & 346   & 583     & 13    & 457        & \numb{0.01294}       \\
gap3100    & 921   & 1558    & 11    & 640        & \numb{0.00660}       \\
}{\typeVlsi}

\def\extraspace{-5pt}

\testset{LIN}{
lin01      & 53    & 80      & 4     & 503        & \numb{0.00028}       \\
lin02      & 55    & 82      & 6     & 557        & \numb{0.00037}       \\
lin03      & 57    & 84      & 8     & 926        & \numb{0.00055}       \\
lin04      & 157   & 266     & 6     & 1239       & \numb{0.00097}       \\
lin05      & 160   & 269     & 9     & 1703       & \numb{0.00220}       \\
lin06      & 165   & 274     & 14    & 1348       & \numb{0.00382}       \\
lin07      & 307   & 526     & 6     & 1885       & \numb{0.00200}       \\
lin08      & 311   & 530     & 10    & 2248       & \numb{0.00232}       \\
lin09      & 313   & 532     & 12    & 2752       & \numb{0.00406}       \\
lin10      & 321   & 540     & 20    & 4132       & \numb{0.01713}       \\
lin11      & 816   & 1460    & 10    & 4280       & \numb{0.01250}       \\
lin12      & 818   & 1462    & 12    & 5250       & \numb{0.01713}       \\
lin13      & 822   & 1466    & 16    & 4609       & \numb{0.01719}       \\
lin14      & 828   & 1472    & 22    & 5824       & \numb{0.02261}       \\
lin15      & 840   & 1484    & 34    & 7145       & \numb{0.08825}       \\
lin16      & 1981  & 3633    & 12    & 6618       & \numb{0.03198}       \\
lin17      & 1989  & 3641    & 20    & 8405       & \numb{0.05116}       \\
lin18      & 1994  & 3646    & 25    & 9714       & \numb{0.47534}       \\
lin19      & 2010  & 3662    & 41    & 13268      & \numb{10.96578}      \\
lin20      & 3675  & 6709    & 11    & 6673       & \numb{0.02922}       \\
lin21      & 3683  & 6717    & 20    & 9143       & \numb{0.06244}       \\
lin22      & 3692  & 6726    & 28    & 10519      & \numb{0.12900}       \\
lin23      & 3716  & 6750    & 52    & 17560      & \numb{11.07980}      \\
lin24      & 7998  & 14734   & 16    & 15076      & \numb{0.10387}       \\
lin25      & 8007  & 14743   & 24    & 17803      & \numb{0.29050}       \\
lin26      & 8013  & 14749   & 30    & 21757      & \numb{0.31898}       \\
lin27      & 8017  & 14753   & 36    & 20678      & \numb{1.57743}       \\
lin29      & 19083 & 35636   & 24    & 23765      & \numb{1.74601}       \\
lin30      & 19091 & 35644   & 31    & 27684      & \numb{0.68134}       \\
lin31      & 19100 & 35653   & 40    & 31696      & \numb{34.06105}      \\
lin32      & 19112 & 35665   & 53    & 39832      & \numb{150.20080}     \\
lin34      & 38282 & 71521   & 34    & 45018      & \numb{12.92992}      \\
lin35      & 38294 & 71533   & 45    & 50559      & \numb{26.90421}      \\
lin36      & 38307 & 71546   & 58    & 55608      & \numb{47.58989}      \\
}{\typeVlsi}


\testset{MSM}{
msm0580    & 338   & 541     & 11    & 467        & \numb{0.00555}       \\
msm0654    & 1290  & 2270    & 10    & 823        & \numb{0.00693}       \\
msm0709    & 1442  & 2403    & 16    & 884        & \numb{0.00968}       \\
msm0920    & 752   & 1264    & 26    & 806        & \numb{0.04091}       \\
msm1008    & 402   & 695     & 11    & 494        & \numb{0.00790}       \\
msm1234    & 933   & 1632    & 13    & 550        & \numb{0.00664}       \\
msm1477    & 1199  & 2078    & 31    & 1068       & \numb{0.25943}       \\
msm1707    & 278   & 478     & 11    & 564        & \numb{0.00121}       \\
msm1844    & 90    & 135     & 10    & 188        & \numb{0.00082}       \\
msm1931    & 875   & 1522    & 10    & 604        & \numb{0.00341}       \\
msm2000    & 898   & 1562    & 10    & 594        & \numb{0.00357}       \\
msm2152    & 2132  & 3702    & 37    & 1590       & \numb{0.26882}       \\
msm2326    & 418   & 723     & 14    & 399        & \numb{0.00299}       \\
msm2492    & 4045  & 7094    & 12    & 1459       & \numb{0.03485}       \\
msm2525    & 3031  & 5239    & 12    & 1290       & \numb{0.01567}       \\
msm2601    & 2961  & 5100    & 16    & 1440       & \numb{0.04217}       \\
msm2705    & 1359  & 2458    & 13    & 714        & \numb{0.01737}       \\
msm2802    & 1709  & 2963    & 18    & 926        & \numb{0.02756}       \\
msm3277    & 1704  & 2991    & 12    & 869        & \numb{0.02142}       \\
msm3676    & 957   & 1554    & 10    & 607        & \numb{0.00556}       \\
msm3727    & 4640  & 8255    & 21    & 1376       & \numb{0.07875}       \\
msm3829    & 4221  & 7255    & 12    & 1571       & \numb{0.04156}       \\
msm4038    & 237   & 390     & 11    & 353        & \numb{0.00253}       \\
msm4114    & 402   & 690     & 16    & 393        & \numb{0.00363}       \\
msm4190    & 391   & 666     & 16    & 381        & \numb{0.00487}       \\
msm4224    & 191   & 302     & 11    & 311        & \numb{0.00232}       \\
msm4312    & 5181  & 8893    & 10    & 2016       & \numb{0.04230}       \\
msm4414    & 317   & 476     & 11    & 408        & \numb{0.00263}       \\
msm4515    & 777   & 1358    & 13    & 630        & \numb{0.01328}       \\
}{\typeVlsi}


\testset{TAQ}{
taq0023    & 572   & 963     & 11    & 621        & \numb{0.00676}       \\
taq0365    & 4186  & 7074    & 22    & 1914       & \numb{0.07876}       \\
taq0431    & 1128  & 1905    & 13    & 897        & \numb{0.01661}       \\
taq0631    & 609   & 932     & 10    & 581        & \numb{0.00823}       \\
taq0739    & 837   & 1438    & 16    & 848        & \numb{0.01852}       \\
taq0741    & 712   & 1217    & 16    & 847        & \numb{0.02359}       \\
taq0751    & 1051  & 1791    & 16    & 939        & \numb{0.03151}       \\
taq0891    & 331   & 560     & 10    & 319        & \numb{0.00314}       \\
taq0910    & 310   & 514     & 17    & 370        & \numb{0.00968}       \\
taq0920    & 122   & 194     & 17    & 210        & \numb{0.00261}       \\
taq0978    & 777   & 1239    & 10    & 566        & \numb{0.00471}       \\
}{\typeVlsi}



\testset{1R}{
1r111      & 1250  & 4704    & 6     & 28000      & \numb{0.00593}       \\*
1r112      & 1250  & 4704    & 6     & 28000      & \numb{0.00481}       \\*
1r113      & 1250  & 4704    & 6     & 26000      & \numb{0.00475}       \\*
1r121      & 1250  & 4704    & 6     & 36000      & \numb{0.00422}       \\
1r122      & 1250  & 4704    & 6     & 45000      & \numb{0.00566}       \\
1r123      & 1250  & 4704    & 6     & 40000      & \numb{0.00394}       \\
1r131      & 1250  & 4704    & 6     & 43000      & \numb{0.00509}       \\
1r132      & 1250  & 4704    & 6     & 37000      & \numb{0.00491}       \\
1r133      & 1250  & 4704    & 6     & 36000      & \numb{0.00436}       \\
1r211      & 1250  & 4704    & 31    & 77000      & \numb{0.34379}       \\
1r212      & 1250  & 4704    & 30    & 81000      & \numb{0.06563}       \\
1r213      & 1250  & 4704    & 29    & 70000      & \numb{0.72025}       \\
1r221      & 1250  & 4704    & 31    & 79000      & \numb{0.14724}       \\
1r222      & 1250  & 4704    & 31    & 68000      & \numb{0.05912}       \\
1r223      & 1250  & 4704    & 31    & 77000      & \numb{0.09650}       \\
1r231      & 1250  & 4704    & 30    & 80000      & \numb{0.14585}       \\
1r232      & 1250  & 4704    & 29    & 86000      & \numb{0.30656}       \\
1r233      & 1250  & 4704    & 31    & 71000      & \numb{1.55050}       \\
1r311      & 1250  & 4704    & 56    & \nosol{}   & \notime{}            \\
1r312      & 1250  & 4704    & 60    & 113000     & \numb{1276.76979}    \\
1r313      & 1250  & 4704    & 58    & 106000     & \numb{496.54942}     \\
1r321      & 1250  & 4704    & 59    & \nosol{}   & \notime{}            \\
1r322      & 1250  & 4704    & 60    & 113000     & \numb{1612.44318}    \\
1r323      & 1250  & 4704    & 60    & \nosol{}   & \notime{}            \\
1r331      & 1250  & 4704    & 58    & 103000     & \numb{1.10688}       \\
1r332      & 1250  & 4704    & 58    & 109000     & \numb{50.35087}      \\
1r333      & 1250  & 4704    & 58    & 113000     & \numb{1708.42476}    \\
}{\typeGridtwoD}

\def\extraspace{-5pt}

\testset{2R}{
2r111      & 2000  & 11600   & 9     & 28000      & \numb{0.01468}       \\
2r112      & 2000  & 11600   & 9     & 32000      & \numb{0.01335}       \\
2r113      & 2000  & 11600   & 9     & 28000      & \numb{0.01072}       \\
2r121      & 2000  & 11600   & 9     & 28000      & \numb{0.01131}       \\
2r122      & 2000  & 11600   & 9     & 29000      & \numb{0.01060}       \\
2r123      & 2000  & 11600   & 9     & 25000      & \numb{0.00945}       \\
2r131      & 2000  & 11600   & 9     & 27000      & \numb{0.01242}       \\
2r132      & 2000  & 11600   & 9     & 33000      & \numb{0.01588}       \\
2r133      & 2000  & 11600   & 9     & 29000      & \numb{0.01078}       \\
2r211      & 2000  & 11600   & 50    & \nosol{}   & \nomem{}             \\
2r212      & 2000  & 11600   & 49    & 80000      & \numb{2819.25943}    \\
2r213      & 2000  & 11600   & 48    & 76000      & \numb{3152.19345}    \\
2r221      & 2000  & 11600   & 50    & \nosol{}   & \notime{}            \\
2r222      & 2000  & 11600   & 50    & \nosol{}   & \notime{}            \\
2r223      & 2000  & 11600   & 49    & \nosol{}   & \notime{}            \\
2r231      & 2000  & 11600   & 50    & \nosol{}   & \notime{}            \\
2r232      & 2000  & 11600   & 49    & \nosol{}   & \notime{}            \\
2r233      & 2000  & 11600   & 47    & \nosol{}   & \notime{}            \\
}{\typeGridthreeD}


\testset{ES10FST}{
es10fst01  & 18    & 20      & 10    & 22920745   & \numb{0.00042}       \\
es10fst02  & 14    & 13      & 10    & 19134104   & \numb{0.00024}       \\
es10fst03  & 17    & 20      & 10    & 26003678   & \numb{0.00046}       \\
es10fst04  & 18    & 20      & 10    & 20461116   & \numb{0.00015}       \\
es10fst05  & 12    & 11      & 10    & 18818916   & \numb{0.00029}       \\
es10fst06  & 17    & 20      & 10    & 26540768   & \numb{0.00045}       \\
es10fst07  & 14    & 13      & 10    & 26025072   & \numb{0.00024}       \\
es10fst08  & 21    & 28      & 10    & 25056214   & \numb{0.00046}       \\
es10fst09  & 21    & 29      & 10    & 22062355   & \numb{0.00041}       \\
es10fst10  & 18    & 21      & 10    & 23936095   & \numb{0.00030}       \\
es10fst11  & 14    & 13      & 10    & 22239535   & \numb{0.00026}       \\
es10fst12  & 13    & 12      & 10    & 19626318   & \numb{0.00031}       \\
es10fst13  & 18    & 21      & 10    & 19483914   & \numb{0.00029}       \\
es10fst14  & 24    & 32      & 10    & 21856128   & \numb{0.00042}       \\
es10fst15  & 16    & 18      & 10    & 18641924   & \numb{0.00021}       \\
}{\typeFST}

\testset{ES20FST}{
es20fst01  & 29    & 28      & 20    & 33703886   & \numb{0.00414}       \\
es20fst02  & 29    & 28      & 20    & 32639486   & \numb{0.00138}       \\
es20fst03  & 27    & 26      & 20    & 27847417   & \numb{0.00143}       \\
es20fst04  & 57    & 83      & 20    & 27624394   & \numb{0.00316}       \\
es20fst05  & 54    & 77      & 20    & 34033163   & \numb{0.00165}       \\
es20fst06  & 29    & 28      & 20    & 36014241   & \numb{0.00119}       \\
es20fst07  & 45    & 59      & 20    & 34934874   & \numb{0.00179}       \\
es20fst08  & 52    & 74      & 20    & 38016346   & \numb{0.01355}       \\
es20fst09  & 36    & 42      & 20    & 36739939   & \numb{0.00252}       \\
es20fst10  & 49    & 67      & 20    & 34024740   & \numb{0.00246}       \\
es20fst11  & 33    & 36      & 20    & 27123908   & \numb{0.00117}       \\
es20fst12  & 33    & 36      & 20    & 30451397   & \numb{0.00466}       \\
es20fst13  & 35    & 40      & 20    & 34438673   & \numb{0.00309}       \\*
es20fst14  & 36    & 44      & 20    & 34062374   & \numb{0.00914}       \\*
es20fst15  & 37    & 43      & 20    & 32303746   & \numb{0.00275}       \\*
}{\typeFST}

\def\extraspace{0pt}

\testset{ES30FST}{
es30fst01  & 79    & 115     & 30    & 40692993   & \numb{0.03714}       \\
es30fst02  & 71    & 97      & 30    & 40900061   & \numb{0.02883}       \\
es30fst03  & 83    & 120     & 30    & 43120444   & \numb{0.01903}       \\
es30fst04  & 80    & 115     & 30    & 42150958   & \numb{0.00983}       \\
es30fst05  & 58    & 71      & 30    & 41739748   & \numb{0.00713}       \\
es30fst06  & 83    & 119     & 30    & 39955139   & \numb{0.05290}       \\
es30fst07  & 53    & 64      & 30    & 43761391   & \numb{0.00734}       \\
es30fst08  & 69    & 93      & 30    & 41691217   & \numb{0.00788}       \\
es30fst09  & 43    & 44      & 30    & 37133658   & \numb{0.00967}       \\
es30fst10  & 48    & 52      & 30    & 42686610   & \numb{0.00854}       \\
es30fst11  & 79    & 112     & 30    & 41647993   & \numb{0.00706}       \\
es30fst12  & 46    & 48      & 30    & 38416720   & \numb{0.01326}       \\
es30fst13  & 65    & 84      & 30    & 37406646   & \numb{0.00512}       \\
es30fst14  & 53    & 58      & 30    & 42897025   & \numb{0.02083}       \\
es30fst15  & 118   & 188     & 30    & 43035576   & \numb{0.06959}       \\
}{\typeFST}


\testset{ES40FST}{
es40fst01  & 93    & 127     & 40    & 44841522   & \numb{0.03561}       \\
es40fst02  & 82    & 105     & 40    & 46811310   & \numb{0.01641}       \\
es40fst03  & 87    & 116     & 40    & 49974157   & \numb{0.09310}       \\
es40fst04  & 55    & 55      & 40    & 45289864   & \numb{0.01348}       \\
es40fst05  & 121   & 180     & 40    & 51940413   & \numb{0.13277}       \\
es40fst06  & 92    & 123     & 40    & 49753385   & \numb{0.04592}       \\
es40fst07  & 77    & 95      & 40    & 45639009   & \numb{0.10935}       \\
es40fst08  & 98    & 137     & 40    & 48745996   & \numb{0.01618}       \\
es40fst09  & 107   & 153     & 40    & 51761789   & \numb{0.03959}       \\
es40fst10  & 107   & 152     & 40    & 57136852   & \numb{0.11675}       \\
es40fst11  & 97    & 135     & 40    & 46734214   & \numb{0.02724}       \\
es40fst12  & 67    & 75      & 40    & 43843378   & \numb{0.02872}       \\
es40fst13  & 78    & 95      & 40    & 51884545   & \numb{0.02955}       \\
es40fst14  & 98    & 134     & 40    & 49166952   & \numb{0.03014}       \\
es40fst15  & 93    & 129     & 40    & 50828067   & \numb{0.05000}       \\
}{\typeFST}

\testset{ES50FST}{
es50fst01  & 118   & 160     & 50    & 54948660   & \numb{0.05299}       \\
es50fst02  & 125   & 177     & 50    & 55484245   & \numb{0.47418}       \\
es50fst03  & 128   & 182     & 50    & 54691035   & \numb{0.07526}       \\
es50fst04  & 106   & 138     & 50    & 51535766   & \numb{0.34787}       \\
es50fst05  & 104   & 135     & 50    & 55186015   & \numb{0.28013}       \\
es50fst06  & 126   & 182     & 50    & 55804287   & \numb{0.53858}       \\
es50fst07  & 143   & 211     & 50    & 49961178   & \numb{0.04894}       \\
es50fst08  & 83    & 96      & 50    & 53754708   & \numb{0.08934}       \\
es50fst09  & 139   & 202     & 50    & 53456773   & \numb{0.53768}       \\
es50fst10  & 139   & 207     & 50    & 54037963   & \numb{3.00560}       \\
es50fst11  & 100   & 131     & 50    & 52532923   & \numb{0.01905}       \\
es50fst12  & 110   & 149     & 50    & 53409291   & \numb{0.17231}       \\
es50fst13  & 92    & 116     & 50    & 53891019   & \numb{0.03237}       \\
es50fst14  & 120   & 167     & 50    & 53551419   & \numb{0.08207}       \\
es50fst15  & 112   & 147     & 50    & 52180862   & \numb{0.10462}       \\
}{\typeFST}

\testset{ES60FST}{
es60fst01  & 123   & 159     & 60    & 53761423   & \numb{0.43599}       \\
es60fst02  & 186   & 280     & 60    & 55367804   & \numb{3.03250}       \\
es60fst03  & 113   & 142     & 60    & 56566797   & \numb{0.06604}       \\
es60fst04  & 162   & 238     & 60    & 55371042   & \numb{0.11766}       \\
es60fst05  & 119   & 148     & 60    & 54704991   & \numb{0.05909}       \\
es60fst06  & 130   & 174     & 60    & 60421961   & \numb{1.06408}       \\
es60fst07  & 188   & 280     & 60    & 58978041   & \numb{0.07281}       \\
es60fst08  & 109   & 133     & 60    & 58138178   & \numb{0.36580}       \\
es60fst09  & 151   & 216     & 60    & 55877112   & \numb{0.49299}       \\
es60fst10  & 133   & 177     & 60    & 57624488   & \numb{0.03867}       \\
es60fst11  & 121   & 154     & 60    & 56141666   & \numb{0.52339}       \\
es60fst12  & 176   & 257     & 60    & 59791362   & \numb{3.28890}       \\
es60fst13  & 157   & 226     & 60    & 61213533   & \numb{0.09080}       \\
es60fst14  & 118   & 149     & 60    & 56035528   & \numb{0.05842}       \\
es60fst15  & 117   & 151     & 60    & 56622581   & \numb{0.10420}       \\
}{\typeFST}


\testset{TSPFST}{
att48fst   & 139   & 202     & 48    & 30236      & \numb{0.45271}       \\
berlin52fst & 89    & 104    & 52    & 6760       & \numb{33.10841}      \\
eil51fst   & 181   & 289     & 51    & 409        & \numb{513.68985}     \\
}{\typeFST}

\def\extraspace{10pt}

\testset{Copenhagen14}{
ind1       & 18    & 31      & 10    & 604        & \numb{0.00060}       \\
ind2       & 31    & 57      & 10    & 9500       & \numb{0.00092}       \\
ind3       & 16    & 23      & 10    & 600        & \numb{0.00020}       \\
ind4       & 74    & 146     & 25    & 1086       & \numb{0.02335}       \\
ind5       & 114   & 228     & 33    & 1341       & \numb{0.02059}       \\
rc01       & 21    & 35      & 10    & 25980      & \numb{0.00028}       \\
rc02       & 87    & 176     & 30    & 41350      & \numb{0.02867}       \\
rc03       & 109   & 202     & 50    & 54160      & \numb{0.15333}       \\
rt01       & 262   & 740     & 10    & 2146       & \numb{0.00247}       \\
rt02       & 788   & 1938    & 50    & 45852      & \numb{0.70289}       \\
}{\typeObstacle}


\testset{WRP3}{
wrp3-11    & 128   & 227     & 11    & 1100361    & \numb{0.09948}       \\
wrp3-12    & 84    & 149     & 12    & 1200237    & \numb{0.12333}       \\
wrp3-13    & 311   & 613     & 13    & 1300497    & \numb{20.01383}      \\
wrp3-14    & 128   & 247     & 14    & 1400250    & \numb{3.77965}       \\
wrp3-15    & 138   & 257     & 15    & 1500422    & \numb{51.15627}      \\
wrp3-16    & 204   & 374     & 16    & 1600208    & \numb{11.62201}      \\
wrp3-17    & 177   & 354     & 17    & 1700442    & \numb{422.07924}     \\
wrp3-19    & 189   & 353     & 19    & 1900439    & \numb{1765.72262}    \\
wrp3-20    & 245   & 454     & 20    & \nosol{}   & \notime{}            \\
wrp3-21    & 237   & 444     & 21    & \nosol{}   & \notime{}            \\
wrp3-22    & 233   & 431     & 22    & \nosol{}   & \notime{}            \\
wrp3-23    & 132   & 230     & 23    & \nosol{}   & \notime{}            \\
wrp3-24    & 262   & 487     & 24    & \nosol{}   & \notime{}            \\
wrp3-25    & 246   & 468     & 25    & \nosol{}   & \nomem{}             \\
wrp3-26    & 402   & 780     & 26    & \nosol{}   & \nomem{}             \\
wrp3-27    & 370   & 721     & 27    & \nosol{}   & \notime{}            \\
wrp3-28    & 307   & 559     & 28    & \nosol{}   & \nomem{}             \\
wrp3-29    & 245   & 436     & 29    & \nosol{}   & \nomem{}             \\
wrp3-30    & 467   & 896     & 30    & \nosol{}   & \nomem{}             \\
wrp3-31    & 323   & 592     & 31    & \nosol{}   & \notime{}            \\
wrp3-33    & 437   & 838     & 33    & \nosol{}   & \nomem{}             \\
wrp3-34    & 1244  & 2474    & 34    & \nosol{}   & \nomem{}             \\
wrp3-36    & 435   & 818     & 36    & \nosol{}   & \nomem{}             \\
wrp3-37    & 1011  & 2010    & 37    & \nosol{}   & \nomem{}             \\
wrp3-38    & 603   & 1207    & 38    & \nosol{}   & \notime{}            \\
wrp3-39    & 703   & 1616    & 39    & \nosol{}   & \nomem{}             \\
wrp3-41    & 178   & 307     & 41    & \nosol{}   & \nomem{}             \\
wrp3-42    & 705   & 1373    & 42    & \nosol{}   & \nomem{}             \\
wrp3-43    & 173   & 298     & 43    & \nosol{}   & \nomem{}             \\
wrp3-45    & 1414  & 2813    & 45    & \nosol{}   & \nomem{}             \\
wrp3-48    & 925   & 1738    & 48    & \nosol{}   & \nomem{}             \\
wrp3-49    & 886   & 1800    & 49    & \nosol{}   & \nomem{}             \\
wrp3-50    & 1119  & 2251    & 50    & \nosol{}   & \nomem{}             \\
wrp3-52    & 701   & 1352    & 52    & \nosol{}   & \nomem{}             \\
wrp3-53    & 775   & 1471    & 53    & \nosol{}   & \nomem{}             \\
wrp3-55    & 1645  & 3186    & 55    & \nosol{}   & \nomem{}             \\
wrp3-56    & 853   & 1590    & 56    & \nosol{}   & \nomem{}             \\
wrp3-60    & 838   & 1763    & 60    & \nosol{}   & \nomem{}             \\
wrp3-62    & 670   & 1316    & 62    & \nosol{}   & \nomem{}             \\
}{\typeGroup}

\testset{WRP4}{
wrp4-11    & 123   & 233     & 11    & 1100179    & \numb{0.24070}       \\
wrp4-13    & 110   & 188     & 13    & 1300798    & \numb{0.01944}       \\
wrp4-14    & 145   & 283     & 14    & 1400290    & \numb{2.79831}       \\
wrp4-15    & 193   & 369     & 15    & 1500405    & \numb{8.19263}       \\
wrp4-16    & 311   & 579     & 16    & 1601190    & \numb{149.54248}     \\
wrp4-17    & 223   & 404     & 17    & 1700525    & \numb{69.52406}      \\
wrp4-18    & 211   & 380     & 18    & 1801464    & \numb{1212.77670}    \\
wrp4-19    & 119   & 206     & 19    & 1901446    & \numb{0.92547}       \\
wrp4-21    & 529   & 1032    & 21    & \nosol{}   & \notime{}            \\
wrp4-22    & 294   & 568     & 22    & \nosol{}   & \notime{}            \\
wrp4-23    & 257   & 515     & 23    & \nosol{}   & \nomem{}             \\
wrp4-24    & 493   & 963     & 24    & \nosol{}   & \nomem{}             \\
wrp4-25    & 422   & 808     & 25    & \nosol{}   & \nomem{}             \\
wrp4-26    & 396   & 781     & 26    & \nosol{}   & \nomem{}             \\
wrp4-27    & 243   & 497     & 27    & \nosol{}   & \nomem{}             \\
wrp4-28    & 272   & 545     & 28    & \nosol{}   & \nomem{}             \\
wrp4-29    & 247   & 505     & 29    & \nosol{}   & \nomem{}             \\
wrp4-30    & 361   & 724     & 30    & \nosol{}   & \nomem{}             \\
wrp4-31    & 390   & 786     & 31    & \nosol{}   & \nomem{}             \\
wrp4-32    & 311   & 632     & 32    & \nosol{}   & \nomem{}             \\
wrp4-33    & 304   & 571     & 33    & \nosol{}   & \notime{}            \\
wrp4-34    & 314   & 650     & 34    & \nosol{}   & \nomem{}             \\
wrp4-35    & 471   & 954     & 35    & \nosol{}   & \nomem{}             \\
wrp4-36    & 363   & 750     & 36    & \nosol{}   & \nomem{}             \\
wrp4-37    & 522   & 1054    & 37    & \nosol{}   & \nomem{}             \\
wrp4-38    & 294   & 618     & 38    & \nosol{}   & \nomem{}             \\
wrp4-39    & 802   & 1553    & 39    & \nosol{}   & \nomem{}             \\
wrp4-40    & 538   & 1088    & 40    & \nosol{}   & \nomem{}             \\
wrp4-41    & 465   & 955     & 41    & \nosol{}   & \nomem{}             \\
wrp4-42    & 552   & 1131    & 42    & \nosol{}   & \nomem{}             \\
wrp4-43    & 596   & 1148    & 43    & \nosol{}   & \notime{}            \\
wrp4-44    & 398   & 788     & 44    & \nosol{}   & \nomem{}             \\
wrp4-45    & 388   & 815     & 45    & \nosol{}   & \nomem{}             \\
wrp4-46    & 632   & 1287    & 46    & \nosol{}   & \nomem{}             \\
wrp4-47    & 555   & 1098    & 47    & \nosol{}   & \nomem{}             \\
wrp4-48    & 451   & 825     & 48    & \nosol{}   & \notime{}            \\
wrp4-49    & 557   & 1080    & 49    & \nosol{}   & \nomem{}             \\
wrp4-50    & 564   & 1112    & 50    & \nosol{}   & \nomem{}             \\
wrp4-51    & 668   & 1306    & 51    & \nosol{}   & \nomem{}             \\
wrp4-52    & 547   & 1115    & 52    & \nosol{}   & \nomem{}             \\
wrp4-53    & 615   & 1232    & 53    & \nosol{}   & \nomem{}             \\
wrp4-54    & 688   & 1388    & 54    & \nosol{}   & \nomem{}             \\
wrp4-55    & 610   & 1201    & 55    & \nosol{}   & \nomem{}             \\
wrp4-56    & 839   & 1617    & 56    & \nosol{}   & \nomem{}             \\
wrp4-58    & 757   & 1493    & 58    & \nosol{}   & \nomem{}             \\
wrp4-59    & 904   & 1806    & 59    & \nosol{}   & \nomem{}             \\
wrp4-60    & 693   & 1370    & 60    & \nosol{}   & \nomem{}             \\
wrp4-61    & 775   & 1538    & 61    & \nosol{}   & \notime{}            \\
wrp4-62    & 1283  & 2493    & 62    & \nosol{}   & \nomem{}             \\
wrp4-63    & 1121  & 2227    & 63    & \nosol{}   & \nomem{}             \\
}
{\typeGroup}


\testset{vienna-i-advanced}{
I052a      & 160   & 237     & 23    & 13309487   & \numb{0.01293}      \\
I054a      & 540   & 817     & 25    & 15841596   & \numb{0.02124}      \\
I056a      & 290   & 439     & 34    & 14171206   & \numb{0.08587}      \\
}
{\typeViennaAdvance}

\def\extraspace{-5pt}

\testset{vienna-i-simple}{
I052       & 2363  & 3761    & 40    & \nosol{}   & \notime{}           \\
I054       & 3803  & 6213    & 38    & \nosol{}   & \notime{}           \\
I056       & 1991  & 3176    & 51    & \nosol{}   & \notime{}           \\
}
{\typeViennaSimple}



\testset{X}{
berlin52   & 52    & 1326    & 16    & 1044       & \numb{0.01292}       \\
brasil58   & 58    & 1653    & 25    & 13655      & \numb{0.00515}       \\
}
{\typePfourE}


\testset{P4E}{
p455       & 100   & 4950    & 5     & 1138       & \numb{0.00118}       \\
p456       & 100   & 4950    & 5     & 1228       & \numb{0.00110}       \\
p457       & 100   & 4950    & 10    & 1609       & \numb{0.00132}       \\
p458       & 100   & 4950    & 10    & 1868       & \numb{0.00169}       \\
p459       & 100   & 4950    & 20    & 2345       & \numb{0.00308}       \\
p460       & 100   & 4950    & 20    & 2959       & \numb{0.00492}       \\
p461       & 100   & 4950    & 50    & 4474       & \numb{0.03041}       \\
p463       & 200   & 19900   & 10    & 1510       & \numb{0.00471}       \\
p464       & 200   & 19900   & 20    & 2545       & \numb{0.01233}       \\
p465       & 200   & 19900   & 40    & 3853       & \numb{0.06634}       \\
}{\typePfourE \label{p4elabel}}


\testset{P4Z}{
p401       & 100   & 4950    & 5     & 155        & \numb{0.00087}       \\
p402       & 100   & 4950    & 5     & 116        & \numb{0.00080}       \\
p403       & 100   & 4950    & 5     & 179        & \numb{0.00145}       \\
p404       & 100   & 4950    & 10    & 270        & \numb{0.00136}       \\
p405       & 100   & 4950    & 10    & 270        & \numb{0.00225}       \\
p406       & 100   & 4950    & 10    & 290        & \numb{0.00298}       \\
p407       & 100   & 4950    & 20    & 590        & \numb{0.25733}       \\
p408       & 100   & 4950    & 20    & 542        & \numb{0.15382}       \\
p409       & 100   & 4950    & 50    & 963        & \numb{4350.84904}    \\
p410       & 100   & 4950    & 50    & 1010       & \numb{245.30424}     \\
}
{\typePfourZ}


\testset{P6E}{
p619       & 100   & 180     & 5     & 7485       & \numb{0.00055}      \\
p620       & 100   & 180     & 5     & 8746       & \numb{0.00056}      \\
p621       & 100   & 180     & 5     & 8688       & \numb{0.00055}      \\
p622       & 100   & 180     & 10    & 15972      & \numb{0.00103}      \\
p623       & 100   & 180     & 10    & 19496      & \numb{0.00156}      \\
p624       & 100   & 180     & 20    & 20246      & \numb{0.00327}      \\
p625       & 100   & 180     & 20    & 23078      & \numb{0.00301}      \\
p626       & 100   & 180     & 20    & 22346      & \numb{0.01082}      \\
p627       & 100   & 180     & 50    & 40647      & \numb{0.02342}      \\
p628       & 100   & 180     & 50    & 40008      & \numb{0.05235}      \\
p629       & 100   & 180     & 50    & 43287      & \numb{0.04251}      \\
p630       & 200   & 370     & 10    & 26125      & \numb{0.00115}      \\
p631       & 200   & 370     & 20    & 39067      & \numb{0.00968}      \\
p632       & 200   & 370     & 40    & 56217      & \numb{0.03812}      \\
}
{\typePsixE}

\testset{P6Z}{
p602       & 100   & 180     & 5     & 8083       & \numb{0.00064}       \\
p603       & 100   & 180     & 5     & 5022       & \numb{0.00055}       \\
p604       & 100   & 180     & 10    & 11397      & \numb{0.00103}       \\
p605       & 100   & 180     & 10    & 10355      & \numb{0.00105}       \\
p606       & 100   & 180     & 11    & 13048      & \numb{0.00110}       \\
p607       & 100   & 180     & 21    & 15358      & \numb{0.00265}       \\
p608       & 100   & 180     & 21    & 14439      & \numb{0.00232}       \\
p609       & 100   & 180     & 20    & 18263      & \numb{0.00437}       \\
p610       & 100   & 180     & 50    & 30161      & \numb{0.03407}       \\
p611       & 100   & 180     & 50    & 26903      & \numb{0.09018}       \\
p612       & 100   & 180     & 50    & 30258      & \numb{0.10735}       \\
p613       & 200   & 370     & 10    & 18429      & \numb{0.00097}       \\
p614       & 200   & 370     & 20    & 27276      & \numb{0.01686}       \\
p615       & 200   & 370     & 40    & 42474      & \numb{0.04866}       \\
}
{\typePsixZ}



\testset{B}{
b01        & 50    & 63      & 9     & 82         & \numb{0.00097}       \\
b02        & 50    & 63      & 13    & 83         & \numb{0.00282}       \\
b03        & 50    & 63      & 25    & 138        & \numb{0.13586}       \\
b04        & 50    & 100     & 9     & 59         & \numb{0.00061}       \\
b05        & 50    & 100     & 13    & 61         & \numb{0.00758}       \\
b06        & 50    & 100     & 25    & 122        & \numb{0.25593}       \\
b07        & 75    & 94      & 13    & 111        & \numb{0.00462}       \\
b08        & 75    & 94      & 19    & 104        & \numb{0.00417}       \\
b09        & 75    & 94      & 38    & 220        & \numb{1858.38164}    \\
b10        & 75    & 150     & 13    & 86         & \numb{0.00586}       \\
b11        & 75    & 150     & 19    & 88         & \numb{0.01927}       \\
b12        & 75    & 150     & 38    & 174        & \numb{98.16452}      \\
b13        & 100   & 125     & 17    & 165        & \numb{0.11100}       \\
b14        & 100   & 125     & 25    & 235        & \numb{1.70413}       \\
b15        & 100   & 125     & 50    & 318        & \numb{176.19179}     \\
b16        & 100   & 200     & 17    & 127        & \numb{0.00757}       \\
b17        & 100   & 200     & 25    & 131        & \numb{266.81382}     \\
b18        & 100   & 200     & 50    & 218        & \numb{3228.99124}    \\
}{\typeRandom}

\testset{C}{
c01        & 500   & 625     & 5     & 85         & \numb{0.00217}       \\
c02        & 500   & 625     & 10    & 144        & \numb{0.00459}       \\
c06        & 500   & 1000    & 5     & 55         & \numb{0.00283}       \\
c07        & 500   & 1000    & 10    & 102        & \numb{0.02380}       \\
c11        & 500   & 2500    & 5     & 32         & \numb{0.00380}       \\
c12        & 500   & 2500    & 10    & 46         & \numb{0.01180}       \\
c16        & 500   & 12500   & 5     & 11         & \numb{0.00516}       \\
c17        & 500   & 12500   & 10    & 18         & \numb{0.01123}       \\
}{\typeRandom}

\testset{D}{
d01        & 1000  & 1250    & 5     & 106        & \numb{0.00380}       \\
d02        & 1000  & 1250    & 10    & 220        & \numb{0.05223}       \\
d06        & 1000  & 2000    & 5     & 67         & \numb{0.00438}       \\
d07        & 1000  & 2000    & 10    & 103        & \numb{0.00979}       \\
d11        & 1000  & 5000    & 5     & 29         & \numb{0.00688}       \\
d12        & 1000  & 5000    & 10    & 42         & \numb{0.01630}       \\
d16        & 1000  & 25000   & 5     & 13         & \numb{0.01088}       \\
d17        & 1000  & 25000   & 10    & 23         & \numb{0.07248}       \\
}{\typeRandom}

\def\extraspace{6pt}

\testset{E}{
e01        & 2500  & 3125    & 5     & 111        & \numb{0.00689}       \\*
e02        & 2500  & 3125    & 10    & 214        & \numb{0.04494}       \\*
e06        & 2500  & 5000    & 5     & 73         & \numb{0.00806}       \\
e07        & 2500  & 5000    & 10    & 145        & \numb{0.23007}       \\
e11        & 2500  & 12500   & 5     & 34         & \numb{0.00947}       \\
e12        & 2500  & 12500   & 10    & 67         & \numb{0.16077}       \\
e16        & 2500  & 62500   & 5     & 15         & \numb{0.01602}       \\*
e17        & 2500  & 62500   & 10    & 25         & \numb{0.12417}       \\*
}{\typeRandom}


\testset{I080}{
i080-001   & 80    & 120     & 6     & 1787       & \numb{0.00056}       \\
i080-002   & 80    & 120     & 6     & 1607       & \numb{0.00033}       \\
i080-003   & 80    & 120     & 6     & 1713       & \numb{0.00039}       \\
i080-004   & 80    & 120     & 6     & 1866       & \numb{0.00040}       \\
i080-005   & 80    & 120     & 6     & 1790       & \numb{0.00051}       \\
i080-011   & 80    & 350     & 6     & 1479       & \numb{0.00100}       \\
i080-012   & 80    & 350     & 6     & 1484       & \numb{0.00099}       \\
i080-013   & 80    & 350     & 6     & 1381       & \numb{0.00072}       \\
i080-014   & 80    & 350     & 6     & 1397       & \numb{0.00086}       \\
i080-015   & 80    & 350     & 6     & 1495       & \numb{0.00097}       \\
i080-021   & 80    & 3160    & 6     & 1175       & \numb{0.00378}       \\
i080-022   & 80    & 3160    & 6     & 1178       & \numb{0.00381}       \\
i080-023   & 80    & 3160    & 6     & 1174       & \numb{0.00376}       \\
i080-024   & 80    & 3160    & 6     & 1161       & \numb{0.00369}       \\
i080-025   & 80    & 3160    & 6     & 1162       & \numb{0.00374}       \\
i080-031   & 80    & 160     & 6     & 1570       & \numb{0.00081}       \\
i080-032   & 80    & 160     & 6     & 2088       & \numb{0.00081}       \\
i080-033   & 80    & 160     & 6     & 1794       & \numb{0.00062}       \\
i080-034   & 80    & 160     & 6     & 1688       & \numb{0.00056}       \\
i080-035   & 80    & 160     & 6     & 1862       & \numb{0.00068}       \\
i080-041   & 80    & 632     & 6     & 1276       & \numb{0.00096}       \\
i080-042   & 80    & 632     & 6     & 1287       & \numb{0.00118}       \\
i080-043   & 80    & 632     & 6     & 1295       & \numb{0.00115}       \\
i080-044   & 80    & 632     & 6     & 1366       & \numb{0.00135}       \\
i080-045   & 80    & 632     & 6     & 1310       & \numb{0.00100}       \\
i080-101   & 80    & 120     & 8     & 2608       & \numb{0.00069}       \\
i080-102   & 80    & 120     & 8     & 2403       & \numb{0.00142}       \\
i080-103   & 80    & 120     & 8     & 2603       & \numb{0.00180}       \\
i080-104   & 80    & 120     & 8     & 2486       & \numb{0.00218}       \\
i080-105   & 80    & 120     & 8     & 2203       & \numb{0.00057}       \\
i080-111   & 80    & 350     & 8     & 2051       & \numb{0.00796}       \\
i080-112   & 80    & 350     & 8     & 1885       & \numb{0.00631}       \\
i080-113   & 80    & 350     & 8     & 1884       & \numb{0.00440}       \\
i080-114   & 80    & 350     & 8     & 1895       & \numb{0.00481}       \\
i080-115   & 80    & 350     & 8     & 1868       & \numb{0.00366}       \\
i080-121   & 80    & 3160    & 8     & 1561       & \numb{0.02292}       \\
i080-122   & 80    & 3160    & 8     & 1561       & \numb{0.02286}       \\
i080-123   & 80    & 3160    & 8     & 1569       & \numb{0.02331}       \\
i080-124   & 80    & 3160    & 8     & 1555       & \numb{0.02245}       \\
i080-125   & 80    & 3160    & 8     & 1572       & \numb{0.02351}       \\
i080-131   & 80    & 160     & 8     & 2284       & \numb{0.00136}       \\
i080-132   & 80    & 160     & 8     & 2180       & \numb{0.00261}       \\
i080-133   & 80    & 160     & 8     & 2261       & \numb{0.00179}       \\
i080-134   & 80    & 160     & 8     & 2070       & \numb{0.00219}       \\
i080-135   & 80    & 160     & 8     & 2102       & \numb{0.00119}       \\
i080-141   & 80    & 632     & 8     & 1788       & \numb{0.00907}       \\
i080-142   & 80    & 632     & 8     & 1708       & \numb{0.00932}       \\
i080-143   & 80    & 632     & 8     & 1767       & \numb{0.01552}       \\
i080-144   & 80    & 632     & 8     & 1772       & \numb{0.01520}       \\
i080-145   & 80    & 632     & 8     & 1762       & \numb{0.01301}       \\
i080-201   & 80    & 120     & 16    & 4760       & \numb{0.12836}       \\
i080-202   & 80    & 120     & 16    & 4650       & \numb{0.11127}       \\
i080-203   & 80    & 120     & 16    & 4599       & \numb{0.89110}       \\
i080-204   & 80    & 120     & 16    & 4492       & \numb{8.02466}       \\
i080-205   & 80    & 120     & 16    & 4564       & \numb{0.61529}       \\
i080-211   & 80    & 350     & 16    & 3631       & \numb{108.65709}     \\
i080-212   & 80    & 350     & 16    & 3677       & \numb{104.79866}     \\
i080-213   & 80    & 350     & 16    & 3678       & \numb{106.88719}     \\
i080-214   & 80    & 350     & 16    & 3734       & \numb{107.15394}     \\
i080-215   & 80    & 350     & 16    & 3681       & \numb{107.42398}     \\
i080-221   & 80    & 3160    & 16    & 3158       & \numb{111.56453}     \\
i080-222   & 80    & 3160    & 16    & 3141       & \numb{113.02559}     \\
i080-223   & 80    & 3160    & 16    & 3156       & \numb{112.40187}     \\
i080-224   & 80    & 3160    & 16    & 3159       & \numb{114.26396}     \\
i080-225   & 80    & 3160    & 16    & 3150       & \numb{114.48878}     \\
i080-231   & 80    & 160     & 16    & 4354       & \numb{32.82894}      \\
i080-232   & 80    & 160     & 16    & 4199       & \numb{26.64471}      \\
i080-233   & 80    & 160     & 16    & 4118       & \numb{16.69140}      \\
i080-234   & 80    & 160     & 16    & 4274       & \numb{0.89790}       \\
i080-235   & 80    & 160     & 16    & 4487       & \numb{2.32598}       \\
i080-241   & 80    & 632     & 16    & 3538       & \numb{146.82699}     \\
i080-242   & 80    & 632     & 16    & 3458       & \numb{141.58934}     \\
i080-243   & 80    & 632     & 16    & 3474       & \numb{139.92190}     \\
i080-244   & 80    & 632     & 16    & 3466       & \numb{143.14732}     \\
i080-245   & 80    & 632     & 16    & 3467       & \numb{142.41802}     \\
i080-301   & 80    & 120     & 20    & 5519       & \numb{38.99112}      \\
i080-302   & 80    & 120     & 20    & 5944       & \numb{10.55645}      \\
i080-303   & 80    & 120     & 20    & 5777       & \numb{19.15451}      \\
i080-304   & 80    & 120     & 20    & 5586       & \numb{6.39249}       \\
i080-305   & 80    & 120     & 20    & 5932       & \numb{201.23488}     \\
i080-311   & 80    & 350     & 20    & \nosol{}   & \notime{}            \\
i080-312   & 80    & 350     & 20    & \nosol{}   & \notime{}            \\
i080-313   & 80    & 350     & 20    & \nosol{}   & \notime{}            \\
i080-314   & 80    & 350     & 20    & \nosol{}   & \notime{}            \\
i080-315   & 80    & 350     & 20    & \nosol{}   & \notime{}            \\
i080-321   & 80    & 3160    & 20    & \nosol{}   & \notime{}            \\
i080-322   & 80    & 3160    & 20    & \nosol{}   & \notime{}            \\
i080-323   & 80    & 3160    & 20    & \nosol{}   & \notime{}            \\
i080-324   & 80    & 3160    & 20    & \nosol{}   & \notime{}            \\
i080-325   & 80    & 3160    & 20    & \nosol{}   & \notime{}            \\
i080-331   & 80    & 160     & 20    & 5226       & \numb{814.50203}     \\
i080-332   & 80    & 160     & 20    & 5362       & \numb{902.23010}     \\
i080-333   & 80    & 160     & 20    & 5381       & \numb{541.44771}     \\
i080-334   & 80    & 160     & 20    & 5264       & \numb{920.97433}     \\
i080-335   & 80    & 160     & 20    & 4953       & \numb{1004.15581}    \\
i080-341   & 80    & 632     & 20    & \nosol{}   & \notime{}            \\
i080-342   & 80    & 632     & 20    & \nosol{}   & \notime{}            \\
i080-343   & 80    & 632     & 20    & \nosol{}   & \notime{}            \\
i080-344   & 80    & 632     & 20    & \nosol{}   & \notime{}            \\
i080-345   & 80    & 632     & 20    & \nosol{}   & \notime{}            \\
}{\typeIncidenceCost}

\testset{I160}{
i160-001   & 160   & 240     & 7     & 2490       & \numb{0.00440}       \\
i160-002   & 160   & 240     & 7     & 2158       & \numb{0.00145}       \\
i160-003   & 160   & 240     & 7     & 2297       & \numb{0.00206}       \\
i160-004   & 160   & 240     & 7     & 2370       & \numb{0.00211}       \\
i160-005   & 160   & 240     & 7     & 2495       & \numb{0.00241}       \\
i160-011   & 160   & 812     & 7     & 1677       & \numb{0.00709}       \\
i160-012   & 160   & 812     & 7     & 1750       & \numb{0.00750}       \\
i160-013   & 160   & 812     & 7     & 1661       & \numb{0.00465}       \\
i160-014   & 160   & 812     & 7     & 1778       & \numb{0.00794}       \\
i160-015   & 160   & 812     & 7     & 1768       & \numb{0.00938}       \\
i160-021   & 160   & 12720   & 7     & 1352       & \numb{0.05547}       \\
i160-022   & 160   & 12720   & 7     & 1365       & \numb{0.05741}       \\
i160-023   & 160   & 12720   & 7     & 1351       & \numb{0.05578}       \\
i160-024   & 160   & 12720   & 7     & 1371       & \numb{0.05781}       \\
i160-025   & 160   & 12720   & 7     & 1366       & \numb{0.05369}       \\
i160-031   & 160   & 320     & 7     & 2170       & \numb{0.00296}       \\
i160-032   & 160   & 320     & 7     & 2330       & \numb{0.00298}       \\
i160-033   & 160   & 320     & 7     & 2101       & \numb{0.00453}       \\
i160-034   & 160   & 320     & 7     & 2083       & \numb{0.00303}       \\
i160-035   & 160   & 320     & 7     & 2103       & \numb{0.00422}       \\
i160-041   & 160   & 2544    & 7     & 1494       & \numb{0.01275}       \\
i160-042   & 160   & 2544    & 7     & 1486       & \numb{0.01304}       \\
i160-043   & 160   & 2544    & 7     & 1549       & \numb{0.01609}       \\
i160-044   & 160   & 2544    & 7     & 1478       & \numb{0.01051}       \\
i160-045   & 160   & 2544    & 7     & 1554       & \numb{0.01480}       \\
i160-101   & 160   & 240     & 12    & 3859       & \numb{0.08167}       \\
i160-102   & 160   & 240     & 12    & 3747       & \numb{0.20505}       \\
i160-103   & 160   & 240     & 12    & 3837       & \numb{0.09262}       \\
i160-104   & 160   & 240     & 12    & 4063       & \numb{0.01440}       \\
i160-105   & 160   & 240     & 12    & 3563       & \numb{0.03449}       \\
i160-111   & 160   & 812     & 12    & 2869       & \numb{1.43022}       \\
i160-112   & 160   & 812     & 12    & 2924       & \numb{2.30419}       \\
i160-113   & 160   & 812     & 12    & 2866       & \numb{1.57464}       \\
i160-114   & 160   & 812     & 12    & 2989       & \numb{1.94856}       \\
i160-115   & 160   & 812     & 12    & 2937       & \numb{1.43577}       \\
i160-121   & 160   & 12720   & 12    & 2363       & \numb{5.39485}       \\
i160-122   & 160   & 12720   & 12    & 2348       & \numb{5.35543}       \\
i160-123   & 160   & 12720   & 12    & 2355       & \numb{5.41763}       \\
i160-124   & 160   & 12720   & 12    & 2352       & \numb{5.33008}       \\
i160-125   & 160   & 12720   & 12    & 2351       & \numb{5.45926}       \\
i160-131   & 160   & 320     & 12    & 3356       & \numb{0.30557}       \\
i160-132   & 160   & 320     & 12    & 3450       & \numb{0.20455}       \\
i160-133   & 160   & 320     & 12    & 3585       & \numb{0.34488}       \\
i160-134   & 160   & 320     & 12    & 3470       & \numb{0.09512}       \\
i160-135   & 160   & 320     & 12    & 3716       & \numb{0.28598}       \\
i160-141   & 160   & 2544    & 12    & 2549       & \numb{3.39813}       \\
i160-142   & 160   & 2544    & 12    & 2562       & \numb{3.43438}       \\
i160-143   & 160   & 2544    & 12    & 2557       & \numb{2.95534}       \\
i160-144   & 160   & 2544    & 12    & 2607       & \numb{3.23815}       \\
i160-145   & 160   & 2544    & 12    & 2578       & \numb{3.45143}       \\
i160-201   & 160   & 240     & 24    & \nosol{}   & \notime{}            \\
i160-202   & 160   & 240     & 24    & \nosol{}   & \notime{}            \\
i160-203   & 160   & 240     & 24    & 7243       & \numb{4173.67890}    \\
i160-204   & 160   & 240     & 24    & \nosol{}   & \notime{}            \\
i160-205   & 160   & 240     & 24    & \nosol{}   & \notime{}            \\
i160-211   & 160   & 812     & 24    & \nosol{}   & \notime{}            \\
i160-212   & 160   & 812     & 24    & \nosol{}   & \notime{}            \\
i160-213   & 160   & 812     & 24    & \nosol{}   & \notime{}            \\
i160-214   & 160   & 812     & 24    & \nosol{}   & \notime{}            \\
i160-215   & 160   & 812     & 24    & \nosol{}   & \notime{}            \\
i160-221   & 160   & 12720   & 24    & \nosol{}   & \notime{}            \\
i160-222   & 160   & 12720   & 24    & \nosol{}   & \notime{}            \\
i160-223   & 160   & 12720   & 24    & \nosol{}   & \notime{}            \\
i160-224   & 160   & 12720   & 24    & \nosol{}   & \notime{}            \\
i160-225   & 160   & 12720   & 24    & \nosol{}   & \notime{}            \\
i160-231   & 160   & 320     & 24    & \nosol{}   & \notime{}            \\
i160-232   & 160   & 320     & 24    & \nosol{}   & \notime{}            \\
i160-233   & 160   & 320     & 24    & \nosol{}   & \notime{}            \\
i160-234   & 160   & 320     & 24    & \nosol{}   & \notime{}            \\
i160-235   & 160   & 320     & 24    & \nosol{}   & \notime{}            \\
i160-241   & 160   & 2544    & 24    & \nosol{}   & \notime{}            \\
i160-242   & 160   & 2544    & 24    & \nosol{}   & \notime{}            \\
i160-243   & 160   & 2544    & 24    & \nosol{}   & \notime{}            \\
i160-244   & 160   & 2544    & 24    & \nosol{}   & \notime{}            \\
i160-245   & 160   & 2544    & 24    & \nosol{}   & \notime{}            \\
i160-301   & 160   & 240     & 40    & \nosol{}   & \nomem{}             \\
i160-302   & 160   & 240     & 40    & \nosol{}   & \nomem{}             \\
i160-303   & 160   & 240     & 40    & \nosol{}   & \nomem{}             \\
i160-304   & 160   & 240     & 40    & \nosol{}   & \nomem{}             \\
i160-305   & 160   & 240     & 40    & \nosol{}   & \nomem{}             \\
i160-311   & 160   & 812     & 40    & \nosol{}   & \nomem{}             \\
i160-312   & 160   & 812     & 40    & \nosol{}   & \nomem{}             \\
i160-313   & 160   & 812     & 40    & \nosol{}   & \nomem{}             \\
i160-314   & 160   & 812     & 40    & \nosol{}   & \nomem{}             \\
i160-315   & 160   & 812     & 40    & \nosol{}   & \nomem{}             \\
i160-321   & 160   & 12720   & 40    & \nosol{}   & \nomem{}             \\
i160-322   & 160   & 12720   & 40    & \nosol{}   & \nomem{}             \\
i160-323   & 160   & 12720   & 40    & \nosol{}   & \nomem{}             \\
i160-324   & 160   & 12720   & 40    & \nosol{}   & \nomem{}             \\
i160-325   & 160   & 12720   & 40    & \nosol{}   & \nomem{}             \\
i160-331   & 160   & 320     & 40    & \nosol{}   & \nomem{}             \\
i160-332   & 160   & 320     & 40    & \nosol{}   & \nomem{}             \\
i160-333   & 160   & 320     & 40    & \nosol{}   & \nomem{}             \\
i160-334   & 160   & 320     & 40    & \nosol{}   & \nomem{}             \\
i160-335   & 160   & 320     & 40    & \nosol{}   & \nomem{}             \\
i160-341   & 160   & 2544    & 40    & \nosol{}   & \nomem{}             \\
i160-342   & 160   & 2544    & 40    & \nosol{}   & \nomem{}             \\
i160-343   & 160   & 2544    & 40    & \nosol{}   & \nomem{}             \\
i160-344   & 160   & 2544    & 40    & \nosol{}   & \nomem{}             \\
i160-345   & 160   & 2544    & 40    & \nosol{}   & \nomem{}             \\
}{\typeIncidenceCost}

\testset{I320}{
i320-001   & 320   & 480     & 8     & 2672       & \numb{0.00262}       \\
i320-002   & 320   & 480     & 8     & 2847       & \numb{0.00399}       \\
i320-003   & 320   & 480     & 8     & 2972       & \numb{0.00411}       \\
i320-004   & 320   & 480     & 8     & 2905       & \numb{0.00727}       \\
i320-005   & 320   & 480     & 8     & 2991       & \numb{0.00570}       \\
i320-011   & 320   & 1845    & 8     & 2053       & \numb{0.03078}       \\
i320-012   & 320   & 1845    & 8     & 1997       & \numb{0.02330}       \\
i320-013   & 320   & 1845    & 8     & 2072       & \numb{0.03906}       \\
i320-014   & 320   & 1845    & 8     & 2061       & \numb{0.03501}       \\
i320-015   & 320   & 1845    & 8     & 2059       & \numb{0.03983}       \\
i320-021   & 320   & 51040   & 8     & 1553       & \numb{0.36556}       \\
i320-022   & 320   & 51040   & 8     & 1565       & \numb{0.41407}       \\
i320-023   & 320   & 51040   & 8     & 1549       & \numb{0.37291}       \\
i320-024   & 320   & 51040   & 8     & 1553       & \numb{0.36048}       \\
i320-025   & 320   & 51040   & 8     & 1550       & \numb{0.36022}       \\
i320-031   & 320   & 640     & 8     & 2673       & \numb{0.01092}       \\
i320-032   & 320   & 640     & 8     & 2770       & \numb{0.01048}       \\
i320-033   & 320   & 640     & 8     & 2769       & \numb{0.00958}       \\
i320-034   & 320   & 640     & 8     & 2521       & \numb{0.00471}       \\
i320-035   & 320   & 640     & 8     & 2385       & \numb{0.00688}       \\
i320-041   & 320   & 10208   & 8     & 1707       & \numb{0.06121}       \\
i320-042   & 320   & 10208   & 8     & 1682       & \numb{0.04977}       \\
i320-043   & 320   & 10208   & 8     & 1723       & \numb{0.05935}       \\
i320-044   & 320   & 10208   & 8     & 1681       & \numb{0.06713}       \\
i320-045   & 320   & 10208   & 8     & 1686       & \numb{0.05622}       \\
i320-101   & 320   & 480     & 17    & 5548       & \numb{16.13133}      \\
i320-102   & 320   & 480     & 17    & 5556       & \numb{10.28908}      \\
i320-103   & 320   & 480     & 17    & 6239       & \numb{178.46285}     \\
i320-104   & 320   & 480     & 17    & 5703       & \numb{108.73862}     \\
i320-105   & 320   & 480     & 17    & 5928       & \numb{84.74775}      \\
i320-111   & 320   & 1845    & 17    & 4273       & \numb{1706.67394}    \\
i320-112   & 320   & 1845    & 17    & 4213       & \numb{1872.72001}    \\
i320-113   & 320   & 1845    & 17    & 4205       & \numb{1464.60944}    \\
i320-114   & 320   & 1845    & 17    & 4104       & \numb{1581.73330}    \\
i320-115   & 320   & 1845    & 17    & 4238       & \numb{1747.82683}    \\
i320-121   & 320   & 51040   & 17    & 3321       & \numb{2381.93881}    \\
i320-122   & 320   & 51040   & 17    & 3314       & \numb{2336.60330}    \\
i320-123   & 320   & 51040   & 17    & 3332       & \numb{2379.99072}    \\
i320-124   & 320   & 51040   & 17    & 3323       & \numb{2342.52702}    \\
i320-125   & 320   & 51040   & 17    & 3340       & \numb{2353.65542}    \\
i320-131   & 320   & 640     & 17    & 5255       & \numb{343.94163}     \\
i320-132   & 320   & 640     & 17    & 5052       & \numb{33.23370}      \\
i320-133   & 320   & 640     & 17    & 5125       & \numb{44.19732}      \\
i320-134   & 320   & 640     & 17    & 5272       & \numb{374.93796}     \\
i320-135   & 320   & 640     & 17    & 5342       & \numb{211.98616}     \\
i320-141   & 320   & 10208   & 17    & 3606       & \numb{2246.12522}    \\
i320-142   & 320   & 10208   & 17    & 3567       & \numb{2305.36881}    \\
i320-143   & 320   & 10208   & 17    & 3561       & \numb{2244.35405}    \\
i320-144   & 320   & 10208   & 17    & 3512       & \numb{2244.79237}    \\
i320-145   & 320   & 10208   & 17    & 3601       & \numb{2261.94575}    \\
i320-201   & 320   & 480     & 34    & \nosol{}   & \nomem{}            \\
i320-202   & 320   & 480     & 34    & \nosol{}   & \nomem{}            \\
i320-203   & 320   & 480     & 34    & \nosol{}   & \nomem{}            \\
i320-204   & 320   & 480     & 34    & \nosol{}   & \nomem{}            \\
i320-205   & 320   & 480     & 34    & \nosol{}   & \nomem{}            \\
i320-211   & 320   & 1845    & 34    & \nosol{}   & \nomem{}            \\
i320-212   & 320   & 1845    & 34    & \nosol{}   & \nomem{}            \\
i320-213   & 320   & 1845    & 34    & \nosol{}   & \nomem{}            \\
i320-214   & 320   & 1845    & 34    & \nosol{}   & \nomem{}            \\
i320-215   & 320   & 1845    & 34    & \nosol{}   & \nomem{}            \\
i320-221   & 320   & 51040   & 34    & \nosol{}   & \nomem{}            \\
i320-222   & 320   & 51040   & 34    & \nosol{}   & \nomem{}            \\
i320-223   & 320   & 51040   & 34    & \nosol{}   & \nomem{}            \\
i320-224   & 320   & 51040   & 34    & \nosol{}   & \nomem{}            \\
i320-225   & 320   & 51040   & 34    & \nosol{}   & \nomem{}            \\
i320-231   & 320   & 640     & 34    & \nosol{}   & \nomem{}            \\
i320-232   & 320   & 640     & 34    & \nosol{}   & \nomem{}            \\
i320-233   & 320   & 640     & 34    & \nosol{}   & \nomem{}            \\
i320-234   & 320   & 640     & 34    & \nosol{}   & \nomem{}            \\
i320-235   & 320   & 640     & 34    & \nosol{}   & \nomem{}            \\
i320-241   & 320   & 10208   & 34    & \nosol{}   & \nomem{}            \\
i320-242   & 320   & 10208   & 34    & \nosol{}   & \nomem{}            \\
i320-243   & 320   & 10208   & 34    & \nosol{}   & \nomem{}            \\
i320-244   & 320   & 10208   & 34    & \nosol{}   & \nomem{}            \\
i320-245   & 320   & 10208   & 34    & \nosol{}   & \nomem{}            \\
}{\typeIncidenceCost}

\testset{I640}{
i640-001   & 640   & 960     & 9     & 4033       & \numb{0.04220}       \\
i640-002   & 640   & 960     & 9     & 3588       & \numb{0.03003}       \\
i640-003   & 640   & 960     & 9     & 3438       & \numb{0.02363}       \\
i640-004   & 640   & 960     & 9     & 4000       & \numb{0.07037}       \\
i640-005   & 640   & 960     & 9     & 4006       & \numb{0.05535}       \\
i640-011   & 640   & 4135    & 9     & 2392       & \numb{0.16333}       \\
i640-012   & 640   & 4135    & 9     & 2465       & \numb{0.27070}       \\
i640-013   & 640   & 4135    & 9     & 2399       & \numb{0.20514}       \\
i640-014   & 640   & 4135    & 9     & 2171       & \numb{0.03798}       \\
i640-015   & 640   & 4135    & 9     & 2347       & \numb{0.12587}       \\
i640-021   & 640   & 204480  & 9     & 1749       & \numb{4.42602}       \\
i640-022   & 640   & 204480  & 9     & 1756       & \numb{4.03080}       \\
i640-023   & 640   & 204480  & 9     & 1754       & \numb{4.49948}       \\
i640-024   & 640   & 204480  & 9     & 1751       & \numb{4.08132}       \\
i640-025   & 640   & 204480  & 9     & 1745       & \numb{4.45155}       \\
i640-031   & 640   & 1280    & 9     & 3278       & \numb{0.06546}       \\
i640-032   & 640   & 1280    & 9     & 3187       & \numb{0.05648}       \\
i640-033   & 640   & 1280    & 9     & 3260       & \numb{0.06201}       \\
i640-034   & 640   & 1280    & 9     & 2953       & \numb{0.01469}       \\
i640-035   & 640   & 1280    & 9     & 3292       & \numb{0.03329}       \\
i640-041   & 640   & 40896   & 9     & 1897       & \numb{0.95667}       \\
i640-042   & 640   & 40896   & 9     & 1934       & \numb{0.81063}       \\
i640-043   & 640   & 40896   & 9     & 1931       & \numb{0.88873}       \\
i640-044   & 640   & 40896   & 9     & 1938       & \numb{0.97524}       \\
i640-045   & 640   & 40896   & 9     & 1866       & \numb{0.40323}       \\
i640-101   & 640   & 960     & 25    & \nosol{}   & \notime{}            \\
i640-102   & 640   & 960     & 25    & \nosol{}   & \notime{}            \\
i640-103   & 640   & 960     & 25    & \nosol{}   & \notime{}            \\
i640-104   & 640   & 960     & 25    & \nosol{}   & \notime{}            \\
i640-105   & 640   & 960     & 25    & \nosol{}   & \notime{}            \\
i640-111   & 640   & 4135    & 25    & \nosol{}   & \nomem{}             \\
i640-112   & 640   & 4135    & 25    & \nosol{}   & \nomem{}             \\
i640-113   & 640   & 4135    & 25    & \nosol{}   & \nomem{}             \\
i640-114   & 640   & 4135    & 25    & \nosol{}   & \nomem{}             \\
i640-115   & 640   & 4135    & 25    & \nosol{}   & \nomem{}             \\
i640-121   & 640   & 204480  & 25    & \nosol{}   & \nomem{}             \\
i640-122   & 640   & 204480  & 25    & \nosol{}   & \nomem{}             \\
i640-123   & 640   & 204480  & 25    & \nosol{}   & \nomem{}             \\
i640-124   & 640   & 204480  & 25    & \nosol{}   & \nomem{}             \\
i640-125   & 640   & 204480  & 25    & \nosol{}   & \nomem{}             \\
i640-131   & 640   & 1280    & 25    & \nosol{}   & \notime{}            \\
i640-132   & 640   & 1280    & 25    & \nosol{}   & \notime{}            \\
i640-133   & 640   & 1280    & 25    & \nosol{}   & \notime{}            \\
i640-134   & 640   & 1280    & 25    & \nosol{}   & \notime{}            \\
i640-135   & 640   & 1280    & 25    & \nosol{}   & \notime{}            \\
i640-141   & 640   & 40896   & 25    & \nosol{}   & \nomem{}             \\
i640-142   & 640   & 40896   & 25    & \nosol{}   & \nomem{}             \\
i640-143   & 640   & 40896   & 25    & \nosol{}   & \nomem{}             \\
i640-144   & 640   & 40896   & 25    & \nosol{}   & \nomem{}             \\
i640-145   & 640   & 40896   & 25    & \nosol{}   & \nomem{}             \\
i640-201   & 640   & 960     & 50    & \nosol{}   & \nomem{}            \\
i640-202   & 640   & 960     & 50    & \nosol{}   & \nomem{}            \\
i640-203   & 640   & 960     & 50    & \nosol{}   & \nomem{}            \\
i640-204   & 640   & 960     & 50    & \nosol{}   & \nomem{}            \\
i640-205   & 640   & 960     & 50    & \nosol{}   & \nomem{}            \\
i640-211   & 640   & 4135    & 50    & \nosol{}   & \nomem{}            \\
i640-212   & 640   & 4135    & 50    & \nosol{}   & \nomem{}            \\
i640-213   & 640   & 4135    & 50    & \nosol{}   & \nomem{}            \\
i640-214   & 640   & 4135    & 50    & \nosol{}   & \nomem{}            \\
i640-215   & 640   & 4135    & 50    & \nosol{}   & \nomem{}            \\
i640-221   & 640   & 204480  & 50    & \nosol{}   & \nomem{}            \\
i640-222   & 640   & 204480  & 50    & \nosol{}   & \nomem{}            \\
i640-223   & 640   & 204480  & 50    & \nosol{}   & \nomem{}            \\
i640-224   & 640   & 204480  & 50    & \nosol{}   & \nomem{}            \\
i640-225   & 640   & 204480  & 50    & \nosol{}   & \nomem{}            \\
i640-231   & 640   & 1280    & 50    & \nosol{}   & \nomem{}            \\
i640-232   & 640   & 1280    & 50    & \nosol{}   & \nomem{}            \\
i640-233   & 640   & 1280    & 50    & \nosol{}   & \nomem{}            \\
i640-234   & 640   & 1280    & 50    & \nosol{}   & \nomem{}            \\
i640-235   & 640   & 1280    & 50    & \nosol{}   & \nomem{}            \\
i640-241   & 640   & 40896   & 50    & \nosol{}   & \nomem{}            \\
i640-242   & 640   & 40896   & 50    & \nosol{}   & \nomem{}            \\
i640-243   & 640   & 40896   & 50    & \nosol{}   & \nomem{}            \\
i640-244   & 640   & 40896   & 50    & \nosol{}   & \nomem{}            \\
i640-245   & 640   & 40896   & 50    & \nosol{}   & \nomem{}            \\
}{\typeIncidenceCost}


\testset{PUC}{
bipe2p     & 550   & 5013    & 50    & \nosol{}   & \nomem{}            \\
bipe2u     & 550   & 5013    & 50    & \nosol{}   & \nomem{}            \\
cc3-10p    & 1000  & 13500   & 50    & \nosol{}   & \nomem{}            \\
cc3-10u    & 1000  & 13500   & 50    & \nosol{}   & \nomem{}            \\
cc3-11p    & 1331  & 19965   & 61    & \nosol{}   & \nomem{}            \\
cc3-11u    & 1331  & 19965   & 61    & \nosol{}   & \nomem{}            \\
cc3-4p     & 64    & 288     & 8     & 2338       & \numb{0.00565}      \\
cc3-4u     & 64    & 288     & 8     & 23         & \numb{0.00819}      \\
cc3-5p     & 125   & 750     & 13    & 3661       & \numb{3.45007}      \\
cc3-5u     & 125   & 750     & 13    & 36         & \numb{4.63731}      \\
cc5-3p     & 243   & 1215    & 27    & \nosol{}   & \nomem{}            \\
cc5-3u     & 243   & 1215    & 27    & \nosol{}   & \nomem{}            \\
cc6-2p     & 64    & 192     & 12    & 3271       & \numb{0.16628}      \\
cc6-2u     & 64    & 192     & 12    & 32         & \numb{0.27873}      \\
hc6p       & 64    & 192     & 32    & \nosol{}   & \nomem{}            \\
hc6u       & 64    & 192     & 32    & \nosol{}   & \nomem{}            \\
}
{\typePUC}

\def\extraspace{0pt}

\testset{SPG-PUCN}{
cc3-10n    & 1000  & 13500   & 50    & \nosol{}   & \nomem{}             \\*
cc3-11n    & 1331  & 19965   & 61    & \nosol{}   & \nomem{}             \\*
cc3-4n     & 64    & 288     & 8     & 13         & \numb{0.00212}       \\*
cc3-5n     & 125   & 750     & 13    & 20         & \numb{0.56342}       \\*
cc5-3n     & 243   & 1215    & 27    & \nosol{}   & \nomem{}             \\*
cc6-2n     & 64    & 192     & 12    & 18         & \numb{0.04701}       \\*
}{\typePUCn}


\testset{SP}{
antiwheel5 & 10    & 15      & 5     & 7          & \numb{0.00011}       \\
design432  & 8     & 20      & 4     & 9          & \numb{0.00007}       \\
oddcycle3  & 6     & 9       & 3     & 4          & \numb{0.00008}       \\
oddwheel3  & 7     & 9       & 4     & 5          & \numb{0.00009}       \\
se03       & 13    & 21      & 4     & 12         & \numb{0.00017}       \\
}
{\typeArt}

\testset{MC}{
mc2        & 120   & 7140    & 60    & \nosol{}   & \nomem{}             \\
mc3        & 97    & 4656    & 45    & \nosol{}   & \notime{}            \\
}{\typeArt}



\testset{csd}{
csd02      & 3     & 2       & 1     & 0          & \numb{0.00004}          \\
csd03      & 6     & 6       & 3     & 4          & \numb{0.00007}          \\
csd04      & 10    & 12      & 6     & 8          & \numb{0.00017}          \\
csd05      & 15    & 20      & 10    & 13         & \numb{0.00721}          \\
csd06      & 21    & 30      & 15    & 19         & \numb{2.72433}          \\
csd07      & 28    & 42      & 21    & 26         & \numb{6138.67847}       \\
csd08      & 36    & 56      & 28    & \nosol{}   & \notime{}               \\
csd09      & 45    & 72      & 36    & \nosol{}   & \nomem{}                \\
csd10      & 55    & 90      & 45    & \nosol{}   & \nomem{}                \\
csd11      & 66    & 110     & 55    & \nosol{}   & \nomem{}                \\
}
{\typeGap}

\testset{goemans}{
g01-00     & 7     & 9       & 3     & 8          & \numb{0.00006}       \\*
g01-01     & 8     & 10      & 4     & 9          & \numb{0.00008}       \\*
g01-02     & 9     & 11      & 5     & 10         & \numb{0.00011}       \\*
g01-03     & 10    & 12      & 6     & 11         & \numb{0.00024}       \\
g01-04     & 11    & 13      & 7     & 12         & \numb{0.00029}       \\
g01-05     & 12    & 14      & 8     & 13         & \numb{0.00063}       \\
g01-06     & 13    & 15      & 9     & 14         & \numb{0.00424}       \\
g01-07     & 14    & 16      & 10    & 15         & \numb{0.00689}       \\
g01-08     & 15    & 17      & 11    & 16         & \numb{0.01592}       \\
g01-09     & 16    & 18      & 12    & 17         & \numb{0.06691}       \\
g01-10     & 17    & 19      & 13    & 18         & \numb{0.17197}       \\
g01-11     & 18    & 20      & 14    & 19         & \numb{0.49025}       \\
g01-12     & 19    & 21      & 15    & 20         & \numb{2.22014}       \\
g01-13     & 20    & 22      & 16    & 21         & \numb{18.63715}      \\
g01-14     & 21    & 23      & 17    & 22         & \numb{42.57551}      \\
g01-15     & 22    & 24      & 18    & 23         & \numb{39.91527}      \\
g02-00     & 9     & 16      & 4     & 14         & \numb{0.00010}       \\
g02-01     & 10    & 17      & 5     & 15         & \numb{0.00011}       \\
g02-02     & 11    & 18      & 6     & 16         & \numb{0.00018}       \\
g02-03     & 12    & 19      & 7     & 17         & \numb{0.00038}       \\
g02-04     & 13    & 20      & 8     & 18         & \numb{0.00131}       \\
g02-05     & 14    & 21      & 9     & 19         & \numb{0.00363}       \\
g02-06     & 15    & 22      & 10    & 20         & \numb{0.01315}       \\
g02-07     & 16    & 23      & 11    & 21         & \numb{0.03503}       \\
g02-08     & 17    & 24      & 12    & 22         & \numb{0.13502}       \\
g02-09     & 18    & 25      & 13    & 23         & \numb{0.35164}       \\
g02-10     & 19    & 26      & 14    & 24         & \numb{1.35370}       \\
g02-11     & 20    & 27      & 15    & 25         & \numb{6.34649}       \\
g02-12     & 21    & 28      & 16    & 26         & \numb{25.60194}      \\
g02-13     & 22    & 29      & 17    & 27         & \numb{114.36556}     \\
g02-14     & 23    & 30      & 18    & 28         & \numb{300.81812}     \\
g02-15     & 24    & 31      & 19    & 29         & \numb{1171.00970}    \\
g03-00     & 11    & 25      & 5     & 22         & \numb{0.00013}       \\
g03-01     & 12    & 26      & 6     & 23         & \numb{0.00019}       \\
g03-02     & 13    & 27      & 7     & 24         & \numb{0.00044}       \\
g03-03     & 14    & 28      & 8     & 25         & \numb{0.00135}       \\
g03-04     & 15    & 29      & 9     & 26         & \numb{0.00435}       \\
g03-05     & 16    & 30      & 10    & 27         & \numb{0.01566}       \\
g03-06     & 17    & 31      & 11    & 28         & \numb{0.06658}       \\
g03-07     & 18    & 32      & 12    & 29         & \numb{0.14494}       \\
g03-08     & 19    & 33      & 13    & 30         & \numb{0.41747}       \\
g03-09     & 20    & 34      & 14    & 31         & \numb{1.53606}       \\
g03-10     & 21    & 35      & 15    & 32         & \numb{6.41057}       \\
g03-11     & 22    & 36      & 16    & 33         & \numb{28.53643}      \\
g03-12     & 23    & 37      & 17    & 34         & \numb{106.67365}     \\
g03-13     & 24    & 38      & 18    & 35         & \numb{407.01956}     \\
g03-14     & 25    & 39      & 19    & 36         & \numb{1414.00988}    \\
g03-15     & 26    & 40      & 20    & 37         & \numb{4935.89724}    \\
g04-00     & 13    & 36      & 6     & 32         & \numb{0.00043}       \\
g04-01     & 14    & 37      & 7     & 33         & \numb{0.00075}       \\
g04-02     & 15    & 38      & 8     & 34         & \numb{0.00221}       \\
g04-03     & 16    & 39      & 9     & 35         & \numb{0.00711}       \\
g04-04     & 17    & 40      & 10    & 36         & \numb{0.02188}       \\
g04-05     & 18    & 41      & 11    & 37         & \numb{0.06189}       \\
g04-06     & 19    & 42      & 12    & 38         & \numb{0.15510}       \\
g04-07     & 20    & 43      & 13    & 39         & \numb{0.42395}       \\
g04-08     & 21    & 44      & 14    & 40         & \numb{1.57144}       \\
g04-09     & 22    & 45      & 15    & 41         & \numb{6.11440}       \\
g04-10     & 23    & 46      & 16    & 42         & \numb{26.87975}      \\
g04-11     & 24    & 47      & 17    & 43         & \numb{103.26262}     \\
g04-12     & 25    & 48      & 18    & 44         & \numb{379.41257}     \\
g04-13     & 26    & 49      & 19    & 45         & \numb{1380.37798}    \\
g04-14     & 27    & 50      & 20    & 46         & \numb{4854.75386}    \\
g04-15     & 28    & 51      & 21    & \nosol{}   & \notime{}            \\
g05-00     & 15    & 49      & 7     & 44         & \numb{0.00081}       \\
g05-01     & 16    & 50      & 8     & 45         & \numb{0.00214}       \\
g05-02     & 17    & 51      & 9     & 46         & \numb{0.00668}       \\
g05-03     & 18    & 52      & 10    & 47         & \numb{0.02329}       \\
g05-04     & 19    & 53      & 11    & 48         & \numb{0.06539}       \\
g05-05     & 20    & 54      & 12    & 49         & \numb{0.16358}       \\
g05-06     & 21    & 55      & 13    & 50         & \numb{0.41660}       \\
g05-07     & 22    & 56      & 14    & 51         & \numb{1.50248}       \\
g05-08     & 23    & 57      & 15    & 52         & \numb{5.83860}       \\
g05-09     & 24    & 58      & 16    & 53         & \numb{22.36820}      \\
g05-10     & 25    & 59      & 17    & 54         & \numb{89.85322}      \\
g05-11     & 26    & 60      & 18    & 55         & \numb{331.23040}     \\
g05-12     & 27    & 61      & 19    & 56         & \numb{1183.86655}    \\
g05-13     & 28    & 62      & 20    & 57         & \numb{4559.62705}    \\
g05-14     & 29    & 63      & 21    & \nosol{}   & \notime{}            \\
g05-15     & 30    & 64      & 22    & \nosol{}   & \notime{}            \\
g06-00     & 17    & 64      & 8     & 58         & \numb{0.00182}       \\
g06-01     & 18    & 65      & 9     & 59         & \numb{0.00643}       \\
g06-02     & 19    & 66      & 10    & 60         & \numb{0.02161}       \\
g06-03     & 20    & 67      & 11    & 61         & \numb{0.06072}       \\
g06-04     & 21    & 68      & 12    & 62         & \numb{0.14963}       \\
g06-05     & 22    & 69      & 13    & 63         & \numb{0.40265}       \\
g06-06     & 23    & 70      & 14    & 64         & \numb{1.44108}       \\
g06-07     & 24    & 71      & 15    & 65         & \numb{5.37414}       \\
g06-08     & 25    & 72      & 16    & 66         & \numb{20.74916}      \\
g06-09     & 26    & 73      & 17    & 67         & \numb{80.43537}      \\
g06-10     & 27    & 74      & 18    & 68         & \numb{307.30054}     \\
g06-11     & 28    & 75      & 19    & 69         & \numb{1129.36963}    \\
g06-12     & 29    & 76      & 20    & 70         & \numb{4294.38653}    \\
g06-13     & 30    & 77      & 21    & \nosol{}   & \notime{}            \\
g06-14     & 31    & 78      & 22    & \nosol{}   & \notime{}            \\
g06-15     & 32    & 79      & 23    & \nosol{}   & \notime{}            \\
g07-00     & 19    & 81      & 9     & 74         & \numb{0.00314}       \\
g07-01     & 20    & 82      & 10    & 75         & \numb{0.01147}       \\
g07-02     & 21    & 83      & 11    & 76         & \numb{0.03737}       \\
g07-03     & 22    & 84      & 12    & 77         & \numb{0.11196}       \\
g07-04     & 23    & 85      & 13    & 78         & \numb{0.39034}       \\
g07-05     & 24    & 86      & 14    & 79         & \numb{1.37036}       \\
g07-06     & 25    & 87      & 15    & 80         & \numb{5.03220}       \\
g07-07     & 26    & 88      & 16    & 81         & \numb{19.38299}      \\
g07-08     & 27    & 89      & 17    & 82         & \numb{76.22362}      \\
g07-09     & 28    & 90      & 18    & 83         & \numb{276.93132}     \\
g07-10     & 29    & 91      & 19    & 84         & \numb{1064.74504}    \\
g07-11     & 30    & 92      & 20    & 85         & \numb{4173.15536}    \\
g07-12     & 31    & 93      & 21    & \nosol{}   & \notime{}            \\
g07-13     & 32    & 94      & 22    & \nosol{}   & \notime{}            \\
g07-14     & 33    & 95      & 23    & \nosol{}   & \notime{}            \\
g07-15     & 34    & 96      & 24    & \nosol{}   & \notime{}            \\
g08-00     & 21    & 100     & 10    & 92         & \numb{0.00894}       \\
g08-01     & 22    & 101     & 11    & 93         & \numb{0.03214}       \\
g08-02     & 23    & 102     & 12    & 94         & \numb{0.10766}       \\
g08-03     & 24    & 103     & 13    & 95         & \numb{0.38327}       \\
g08-04     & 25    & 104     & 14    & 96         & \numb{1.30899}       \\
g08-05     & 26    & 105     & 15    & 97         & \numb{4.94207}       \\
g08-06     & 27    & 106     & 16    & 98         & \numb{18.08123}      \\
g08-07     & 28    & 107     & 17    & 99         & \numb{68.31416}      \\
g08-08     & 29    & 108     & 18    & 100        & \numb{262.76833}     \\
g08-09     & 30    & 109     & 19    & 101        & \numb{1019.46373}    \\
g08-10     & 31    & 110     & 20    & 102        & \numb{3956.30498}    \\
g08-11     & 32    & 111     & 21    & \nosol{}   & \notime{}            \\
g08-12     & 33    & 112     & 22    & \nosol{}   & \notime{}            \\
g08-13     & 34    & 113     & 23    & \nosol{}   & \notime{}            \\
g08-14     & 35    & 114     & 24    & \nosol{}   & \notime{}            \\
g08-15     & 36    & 115     & 25    & \nosol{}   & \notime{}            \\
g09-00     & 23    & 121     & 11    & 112        & \numb{0.02822}       \\
g09-01     & 24    & 122     & 12    & 113        & \numb{0.15044}       \\
g09-02     & 25    & 123     & 13    & 114        & \numb{0.38309}       \\
g09-03     & 26    & 124     & 14    & 115        & \numb{1.29057}       \\
g09-04     & 27    & 125     & 15    & 116        & \numb{4.62503}       \\
g09-05     & 28    & 126     & 16    & 117        & \numb{17.47903}      \\
g09-06     & 29    & 127     & 17    & 118        & \numb{66.58204}      \\
g09-07     & 30    & 128     & 18    & 119        & \numb{240.70567}     \\
g09-08     & 31    & 129     & 19    & 120        & \numb{939.55325}     \\
g09-09     & 32    & 130     & 20    & 121        & \numb{3748.63517}    \\
g09-10     & 33    & 131     & 21    & \nosol{}   & \notime{}            \\
g09-11     & 34    & 132     & 22    & \nosol{}   & \notime{}            \\
g09-12     & 35    & 133     & 23    & \nosol{}   & \notime{}            \\
g09-13     & 36    & 134     & 24    & \nosol{}   & \notime{}            \\
g09-14     & 37    & 135     & 25    & \nosol{}   & \notime{}            \\
g09-15     & 38    & 136     & 26    & \nosol{}   & \notime{}            \\
g10-00     & 25    & 144     & 12    & 134        & \numb{0.12626}       \\
g10-01     & 26    & 145     & 13    & 135        & \numb{0.35076}       \\
g10-02     & 27    & 146     & 14    & 136        & \numb{1.22802}       \\
g10-03     & 28    & 147     & 15    & 137        & \numb{4.47433}       \\
g10-04     & 29    & 148     & 16    & 138        & \numb{17.18805}      \\
g10-05     & 30    & 149     & 17    & 139        & \numb{63.64039}      \\
g10-06     & 31    & 150     & 18    & 140        & \numb{223.03530}     \\
g10-07     & 32    & 151     & 19    & 141        & \numb{878.75152}     \\
g10-08     & 33    & 152     & 20    & 142        & \numb{3612.40584}    \\
g10-09     & 34    & 153     & 21    & \nosol{}   & \notime{}            \\
g10-10     & 35    & 154     & 22    & \nosol{}   & \notime{}            \\
g10-11     & 36    & 155     & 23    & \nosol{}   & \notime{}            \\
g10-12     & 37    & 156     & 24    & \nosol{}   & \notime{}            \\
g10-13     & 38    & 157     & 25    & \nosol{}   & \notime{}            \\
g10-14     & 39    & 158     & 26    & \nosol{}   & \notime{}            \\
g10-15     & 40    & 159     & 27    & \nosol{}   & \notime{}            \\
g11-00     & 27    & 169     & 13    & 158        & \numb{0.36896}       \\
g11-01     & 28    & 170     & 14    & 159        & \numb{1.14736}       \\
g11-02     & 29    & 171     & 15    & 160        & \numb{4.52897}       \\
g11-03     & 30    & 172     & 16    & 161        & \numb{16.52451}      \\
g11-04     & 31    & 173     & 17    & 162        & \numb{61.29996}      \\
g11-05     & 32    & 174     & 18    & 163        & \numb{216.07445}     \\
g11-06     & 33    & 175     & 19    & 164        & \numb{845.01416}     \\
g11-07     & 34    & 176     & 20    & 165        & \numb{3514.33893}    \\
g11-08     & 35    & 177     & 21    & \nosol{}   & \notime{}            \\
g11-09     & 36    & 178     & 22    & \nosol{}   & \notime{}            \\
g11-10     & 37    & 179     & 23    & \nosol{}   & \notime{}            \\
g11-11     & 38    & 180     & 24    & \nosol{}   & \notime{}            \\
g11-12     & 39    & 181     & 25    & \nosol{}   & \notime{}            \\
g11-13     & 40    & 182     & 26    & \nosol{}   & \notime{}            \\
g11-14     & 41    & 183     & 27    & \nosol{}   & \notime{}            \\
g11-15     & 42    & 184     & 28    & \nosol{}   & \notime{}            \\
g12-00     & 29    & 196     & 14    & 184        & \numb{0.62605}       \\
g12-01     & 30    & 197     & 15    & 185        & \numb{4.49099}       \\
g12-02     & 31    & 198     & 16    & 186        & \numb{15.09835}      \\
g12-03     & 32    & 199     & 17    & 187        & \numb{59.82361}      \\
g12-04     & 33    & 200     & 18    & 188        & \numb{201.24998}     \\
g12-05     & 34    & 201     & 19    & 189        & \numb{790.25169}     \\
g12-06     & 35    & 202     & 20    & 190        & \numb{3361.07545}    \\
g12-07     & 36    & 203     & 21    & \nosol{}   & \notime{}            \\
g12-08     & 37    & 204     & 22    & \nosol{}   & \notime{}            \\
g12-09     & 38    & 205     & 23    & \nosol{}   & \notime{}            \\
g12-10     & 39    & 206     & 24    & \nosol{}   & \notime{}            \\
g12-11     & 40    & 207     & 25    & \nosol{}   & \notime{}            \\
g12-12     & 41    & 208     & 26    & \nosol{}   & \notime{}            \\
g12-13     & 42    & 209     & 27    & \nosol{}   & \notime{}            \\
g12-14     & 43    & 210     & 28    & \nosol{}   & \notime{}            \\
g12-15     & 44    & 211     & 29    & \nosol{}   & \notime{}            \\
g13-00     & 31    & 225     & 15    & 212        & \numb{1.81453}       \\
g13-01     & 32    & 226     & 16    & 213        & \numb{14.46799}      \\
g13-02     & 33    & 227     & 17    & 214        & \numb{56.67415}      \\
g13-03     & 34    & 228     & 18    & 215        & \numb{195.57240}     \\
g13-04     & 35    & 229     & 19    & 216        & \numb{728.65044}     \\
g13-05     & 36    & 230     & 20    & 217        & \numb{3214.44331}    \\
g13-06     & 37    & 231     & 21    & \nosol{}   & \notime{}            \\
g13-07     & 38    & 232     & 22    & \nosol{}   & \notime{}            \\
g13-08     & 39    & 233     & 23    & \nosol{}   & \notime{}            \\
g13-09     & 40    & 234     & 24    & \nosol{}   & \notime{}            \\
g13-10     & 41    & 235     & 25    & \nosol{}   & \notime{}            \\
g13-11     & 42    & 236     & 26    & \nosol{}   & \notime{}            \\
g13-12     & 43    & 237     & 27    & \nosol{}   & \notime{}            \\
g13-13     & 44    & 238     & 28    & \nosol{}   & \notime{}            \\
g13-14     & 45    & 239     & 29    & \nosol{}   & \notime{}            \\
g13-15     & 46    & 240     & 30    & \nosol{}   & \notime{}            \\
g14-00     & 33    & 256     & 16    & 242        & \numb{6.51242}       \\
g14-01     & 34    & 257     & 17    & 243        & \numb{53.63073}      \\
g14-02     & 35    & 258     & 18    & 244        & \numb{189.94231}     \\
g14-03     & 36    & 259     & 19    & 245        & \numb{696.40039}     \\
g14-04     & 37    & 260     & 20    & 246        & \numb{2918.27458}    \\
g14-05     & 38    & 261     & 21    & \nosol{}   & \notime{}            \\
g14-06     & 39    & 262     & 22    & \nosol{}   & \notime{}            \\
g14-07     & 40    & 263     & 23    & \nosol{}   & \notime{}            \\
g14-08     & 41    & 264     & 24    & \nosol{}   & \notime{}            \\
g14-09     & 42    & 265     & 25    & \nosol{}   & \notime{}            \\
g14-10     & 43    & 266     & 26    & \nosol{}   & \notime{}            \\
g14-11     & 44    & 267     & 27    & \nosol{}   & \notime{}            \\
g14-12     & 45    & 268     & 28    & \nosol{}   & \notime{}            \\
g14-13     & 46    & 269     & 29    & \nosol{}   & \notime{}            \\
g14-14     & 47    & 270     & 30    & \nosol{}   & \notime{}            \\
g14-15     & 48    & 271     & 31    & \nosol{}   & \notime{}            \\
g15-00     & 35    & 289     & 17    & 274        & \numb{22.45794}      \\
g15-01     & 36    & 290     & 18    & 275        & \numb{181.21805}     \\
g15-02     & 37    & 291     & 19    & 276        & \numb{653.42481}     \\
g15-03     & 38    & 292     & 20    & 277        & \numb{2738.78011}    \\
g15-04     & 39    & 293     & 21    & \nosol{}   & \notime{}            \\
g15-05     & 40    & 294     & 22    & \nosol{}   & \notime{}            \\
g15-06     & 41    & 295     & 23    & \nosol{}   & \notime{}            \\
g15-07     & 42    & 296     & 24    & \nosol{}   & \notime{}            \\
g15-08     & 43    & 297     & 25    & \nosol{}   & \notime{}            \\
g15-09     & 44    & 298     & 26    & \nosol{}   & \notime{}            \\
g15-10     & 45    & 299     & 27    & \nosol{}   & \notime{}            \\
g15-11     & 46    & 300     & 28    & \nosol{}   & \notime{}            \\
g15-12     & 47    & 301     & 29    & \nosol{}   & \notime{}            \\
g15-13     & 48    & 302     & 30    & \nosol{}   & \notime{}            \\
g15-14     & 49    & 303     & 31    & \nosol{}   & \notime{}            \\
g15-15     & 50    & 304     & 32    & \nosol{}   & \notime{}            \\
}
{\typeGap}

\testset{skutella}{
s1         & 15    & 35      & 8     & 10         & \numb{0.00282}       \\
s2         & 106   & 399     & 50    & \nosol{}   & \nomem{}             \\
}
{\typeGap}

\testset{smc}{
smc01      & 2     & 1       & 1     & 0          & \numb{0.00003}       \\
smc02      & 3     & 3       & 2     & 2          & \numb{0.00004}       \\
smc03      & 4     & 6       & 3     & 3          & \numb{0.00004}       \\
smc04      & 5     & 10      & 4     & 4          & \numb{0.00005}       \\
smc05      & 6     & 15      & 5     & 5          & \numb{0.00008}       \\
smc06      & 7     & 21      & 6     & 6          & \numb{0.00011}       \\
smc07      & 8     & 28      & 7     & 7          & \numb{0.00020}       \\
smc08      & 9     & 36      & 8     & 8          & \numb{0.00073}       \\
smc09      & 10    & 45      & 9     & 9          & \numb{0.00103}       \\
smc10      & 11    & 55      & 10    & 10         & \numb{0.00200}       \\
smc11      & 12    & 66      & 11    & 11         & \numb{0.00284}       \\
smc12      & 13    & 78      & 12    & 12         & \numb{0.01202}       \\
smc13      & 14    & 91      & 13    & 13         & \numb{0.04611}       \\
smc14      & 15    & 105     & 14    & 14         & \numb{0.13398}       \\
smc15      & 16    & 120     & 15    & 15         & \numb{0.31307}       \\
}
{\typeGap}

\newpage

\section{Results on Rectilinear Instances}\label{appendix:rect}
For these experiments, the setup is identical to the one described in Appendix~\ref{appendix:graph}.
For each instance, first, the Hanan grid was computed and written to an instance file
describing an instance of the Steiner tree problem in graphs. The latter instance was then solved. The reported run times only include the time to solve the graphic instance, in particular, the time needed to read in the Hanan grid is not included.

\def\extraspace{-6pt}
\vspace{5pt}

\testset{CARIOCA}
{
carioca\_3\_11\_01 & 1331  & 3630    & 11    & 311221222  & \numb{0.02015}         \\
carioca\_3\_11\_02 & 1331  & 3630    & 11    & 466149453  & \numb{0.01888}         \\
carioca\_3\_11\_03 & 1331  & 3630    & 11    & 439391117  & \numb{0.05784}         \\
carioca\_3\_11\_04 & 1331  & 3630    & 11    & 413409501  & \numb{0.05352}         \\
carioca\_3\_11\_05 & 1331  & 3630    & 11    & 387407782  & \numb{0.00877}         \\
carioca\_3\_12\_01 & 1728  & 4752    & 12    & 494141224  & \numb{0.03222}         \\
carioca\_3\_12\_02 & 1728  & 4752    & 12    & 443366694  & \numb{0.02421}         \\
carioca\_3\_12\_03 & 1728  & 4752    & 12    & 429706282  & \numb{0.08283}         \\
carioca\_3\_12\_04 & 1728  & 4752    & 12    & 486545112  & \numb{0.03070}         \\
carioca\_3\_12\_05 & 1728  & 4752    & 12    & 444438614  & \numb{0.01561}         \\
carioca\_3\_13\_01 & 2197  & 6084    & 13    & 452770450  & \numb{0.03688}         \\
carioca\_3\_13\_02 & 2197  & 6084    & 13    & 462700327  & \numb{0.23799}         \\
carioca\_3\_13\_03 & 2197  & 6084    & 13    & 474263794  & \numb{0.15298}         \\
carioca\_3\_13\_04 & 2197  & 6084    & 13    & 442802506  & \numb{0.01932}         \\
carioca\_3\_13\_05 & 2197  & 6084    & 13    & 547158862  & \numb{0.03185}         \\
carioca\_3\_14\_01 & 2744  & 7644    & 14    & 438382690  & \numb{0.14964}         \\
carioca\_3\_14\_02 & 2744  & 7644    & 14    & 495879854  & \numb{0.13565}         \\
carioca\_3\_14\_03 & 2744  & 7644    & 14    & 480652934  & \numb{0.04798}         \\
carioca\_3\_14\_04 & 2744  & 7644    & 14    & 473370979  & \numb{0.09734}         \\
carioca\_3\_14\_05 & 2744  & 7644    & 14    & 408691456  & \numb{0.04804}         \\
carioca\_3\_15\_01 & 3375  & 9450    & 15    & 603071413  & \numb{0.35178}         \\
carioca\_3\_15\_02 & 3375  & 9450    & 15    & 528575469  & \numb{0.27863}         \\
carioca\_3\_15\_03 & 3375  & 9450    & 15    & 490905559  & \numb{0.15394}         \\
carioca\_3\_15\_04 & 3375  & 9450    & 15    & 540300331  & \numb{0.20102}         \\
carioca\_3\_15\_05 & 3375  & 9450    & 15    & 535330648  & \numb{0.12554}         \\
carioca\_3\_16\_01 & 4096  & 11520   & 16    & 527653658  & \numb{0.18098}         \\
carioca\_3\_16\_02 & 4096  & 11520   & 16    & 500606262  & \numb{0.21274}         \\
carioca\_3\_16\_03 & 4096  & 11520   & 16    & 468414684  & \numb{0.13010}         \\
carioca\_3\_16\_04 & 4096  & 11520   & 16    & 539655260  & \numb{0.23914}         \\
carioca\_3\_16\_05 & 4096  & 11520   & 16    & 540126099  & \numb{0.14038}         \\
carioca\_3\_17\_01 & 4913  & 13872   & 17    & 555032119  & \numb{0.23618}         \\
carioca\_3\_17\_02 & 4913  & 13872   & 17    & 527737314  & \numb{0.93569}         \\
carioca\_3\_17\_03 & 4913  & 13872   & 17    & 562083809  & \numb{0.40511}         \\
carioca\_3\_17\_04 & 4913  & 13872   & 17    & 589257703  & \numb{0.23021}         \\
carioca\_3\_17\_05 & 4913  & 13872   & 17    & 612582645  & \numb{3.85211}         \\
carioca\_3\_18\_01 & 5832  & 16524   & 18    & 556545528  & \numb{0.29135}         \\
carioca\_3\_18\_02 & 5832  & 16524   & 18    & 549027056  & \numb{1.39242}         \\
carioca\_3\_18\_03 & 5832  & 16524   & 18    & 649390363  & \numb{0.70955}         \\
carioca\_3\_18\_04 & 5832  & 16524   & 18    & 554907444  & \numb{0.79962}         \\
carioca\_3\_18\_05 & 5832  & 16524   & 18    & 509234906  & \numb{0.40772}         \\
carioca\_3\_19\_01 & 6859  & 19494   & 19    & 561743715  & \numb{6.78763}         \\
carioca\_3\_19\_02 & 6859  & 19494   & 19    & 608506796  & \numb{1.95704}         \\
carioca\_3\_19\_03 & 6859  & 19494   & 19    & 574021347  & \numb{0.39806}         \\
carioca\_3\_19\_04 & 6859  & 19494   & 19    & 630525105  & \numb{0.34223}         \\
carioca\_3\_19\_05 & 6859  & 19494   & 19    & 650204402  & \numb{6.35446}         \\
carioca\_3\_20\_01 & 8000  & 22800   & 20    & 638376617  & \numb{1.61296}         \\
carioca\_3\_20\_02 & 8000  & 22800   & 20    & 477950448  & \numb{0.14770}         \\
carioca\_3\_20\_03 & 8000  & 22800   & 20    & 746979341  & \numb{8.89011}         \\
carioca\_3\_20\_04 & 8000  & 22800   & 20    & 653809733  & \numb{10.09628}        \\
carioca\_3\_20\_05 & 8000  & 22800   & 20    & 678171940  & \numb{1.64884}         \\
}
{Type: Random 3-d rectilinear instances with coordinates scaled by $10^8$.}

\testset{CARIOCA}
{
carioca\_4\_11\_01 & 14641 & 53240   & 11    & 627022001  & \numb{0.42525}      \\
carioca\_4\_11\_02 & 14641 & 53240   & 11    & 636772154  & \numb{0.22314}      \\
carioca\_4\_11\_03 & 14641 & 53240   & 11    & 607879790  & \numb{1.77909}      \\
carioca\_4\_11\_04 & 14641 & 53240   & 11    & 638743359  & \numb{1.16621}      \\
carioca\_4\_11\_05 & 14641 & 53240   & 11    & 545419447  & \numb{1.88278}      \\
carioca\_4\_12\_01 & 20736 & 76032   & 12    & 641297479  & \numb{0.35005}      \\
carioca\_4\_12\_02 & 20736 & 76032   & 12    & 619890840  & \numb{1.14255}      \\
carioca\_4\_12\_03 & 20736 & 76032   & 12    & 618169838  & \numb{1.15172}      \\
carioca\_4\_12\_04 & 20736 & 76032   & 12    & 573580734  & \numb{0.79095}      \\
carioca\_4\_12\_05 & 20736 & 76032   & 12    & 690707456  & \numb{0.90659}      \\
carioca\_4\_13\_01 & 28561 & 105456  & 13    & 668457902  & \numb{1.77587}      \\
carioca\_4\_13\_02 & 28561 & 105456  & 13    & 732093506  & \numb{5.15086}      \\
carioca\_4\_13\_03 & 28561 & 105456  & 13    & 667816953  & \numb{1.36938}      \\
carioca\_4\_13\_04 & 28561 & 105456  & 13    & 618023300  & \numb{1.46494}      \\
carioca\_4\_13\_05 & 28561 & 105456  & 13    & 816676045  & \numb{9.27424}      \\
carioca\_4\_14\_01 & 38416 & 142688  & 14    & 678015050  & \numb{1.96684}      \\
carioca\_4\_14\_02 & 38416 & 142688  & 14    & 770189931  & \numb{32.55969}     \\
carioca\_4\_14\_03 & 38416 & 142688  & 14    & 774041179  & \numb{7.65624}      \\
carioca\_4\_14\_04 & 38416 & 142688  & 14    & 753394327  & \numb{15.25636}     \\
carioca\_4\_14\_05 & 38416 & 142688  & 14    & 668149161  & \numb{6.90750}      \\
carioca\_4\_15\_01 & 50625 & 189000  & 15    & 867047412  & \numb{3.50454}      \\
carioca\_4\_15\_02 & 50625 & 189000  & 15    & 747749078  & \numb{2.62394}      \\
carioca\_4\_15\_03 & 50625 & 189000  & 15    & 797161324  & \numb{41.34427}     \\
carioca\_4\_15\_04 & 50625 & 189000  & 15    & 682701539  & \numb{9.47503}      \\
carioca\_4\_15\_05 & 50625 & 189000  & 15    & 816563142  & \numb{5.40698}      \\
carioca\_4\_16\_01 & 65536 & 245760  & 16    & 877600183  & \numb{53.43973}     \\
carioca\_4\_16\_02 & 65536 & 245760  & 16    & 840957543  & \numb{11.18493}     \\
carioca\_4\_16\_03 & 65536 & 245760  & 16    & 769487137  & \numb{7.46912}      \\
carioca\_4\_16\_04 & 65536 & 245760  & 16    & 883810994  & \numb{40.60337}     \\
carioca\_4\_16\_05 & 65536 & 245760  & 16    & 844364805  & \numb{26.79172}     \\
carioca\_4\_17\_01 & 83521 & 314432  & 17    & 778812798  & \numb{46.32228}     \\
carioca\_4\_17\_02 & 83521 & 314432  & 17    & 885687619  & \numb{96.92606}     \\
carioca\_4\_17\_03 & 83521 & 314432  & 17    & 866496021  & \numb{345.27471}    \\
carioca\_4\_17\_04 & 83521 & 314432  & 17    & 961783573  & \numb{22.30273}     \\
carioca\_4\_17\_05 & 83521 & 314432  & 17    & 926268906  & \numb{107.80819}    \\
carioca\_4\_18\_01 & 104976 & 396576  & 18    &951149709  & \numb{99.35341}      \\
carioca\_4\_18\_02 & 104976 & 396576  & 18    &842404966  & \numb{98.72303}      \\
carioca\_4\_18\_03 & 104976 & 396576  & 18    &848447422  & \numb{1431.37407}    \\
carioca\_4\_18\_04 & 104976 & 396576  & 18    &956109307  & \numb{212.73207}     \\
carioca\_4\_18\_05 & 104976 & 396576  & 18    &893733310  & \numb{143.34199}     \\
carioca\_4\_19\_01 & 130321 & 493848  & 19    & 917950857  & \numb{186.38252}     \\
carioca\_4\_19\_02 & 130321 & 493848  & 19    & 825014078  & \numb{49.40921}      \\
carioca\_4\_19\_03 & 130321 & 493848  & 19    & 945521812  & \numb{25.66043}      \\
carioca\_4\_19\_04 & 130321 & 493848  & 19    & 912019383  & \numb{160.59101}     \\
carioca\_4\_19\_05 & 130321 & 493848  & 19    & 963415391  & \numb{1625.46514}    \\
carioca\_4\_20\_01 & 160000 & 608000  & 20    & 889180827  & \numb{82.72426}      \\
carioca\_4\_20\_02 & 160000 & 608000  & 20    & 822698792  & \numb{101.10654}     \\
carioca\_4\_20\_03 & 160000 & 608000  & 20    & 884633836  & \numb{125.99493}     \\
carioca\_4\_20\_04 & 160000 & 608000  & 20    & 948878450  & \numb{2618.04508}    \\
carioca\_4\_20\_05 & 160000 & 608000  & 20    & 984006649  & \numb{82.54523}      \\
}
{Type: Random 4-d rectilinear instances with coordinates scaled by $10^8$.}

\testset{CARIOCA}
{
carioca\_5\_11\_01 & 161051 & 732050  & 11    & 925163690  & \numb{34.54747}      \\
carioca\_5\_11\_02 & 161051 & 732050  & 11    & 844673618  & \numb{13.01897}      \\
carioca\_5\_11\_03 & 161051 & 732050  & 11    & 867510918  & \numb{6.29731}       \\
carioca\_5\_11\_04 & 161051 & 732050  & 11    & 906103201  & \numb{14.04304}      \\
carioca\_5\_11\_05 & 161051 & 732050  & 11    & 795198510  & \numb{30.52142}      \\
carioca\_5\_12\_01 & 248832 & 1140480 & 12    & 953491398  & \numb{59.40753}      \\
carioca\_5\_12\_02 & 248832 & 1140480 & 12    & 985601088  & \numb{102.10617}     \\
carioca\_5\_12\_03 & 248832 & 1140480 & 12    & 844385082  & \numb{57.49894}      \\
carioca\_5\_12\_04 & 248832 & 1140480 & 12    & 879014839  & \numb{93.15067}      \\
carioca\_5\_12\_05 & 248832 & 1140480 & 12    & 815604529  & \numb{15.59538}      \\
carioca\_5\_13\_01 & 371293 & 1713660 & 13    & 881473517  & \numb{132.66157}     \\
carioca\_5\_13\_02 & 371293 & 1713660 & 13    & 873559091  & \numb{177.28734}     \\
carioca\_5\_13\_03 & 371293 & 1713660 & 13    & 1005775838 & \numb{184.83810}     \\
carioca\_5\_13\_04 & 371293 & 1713660 & 13    & 922258018  & \numb{237.88241}     \\
carioca\_5\_13\_05 & 371293 & 1713660 & 13    & 879174698  & \numb{46.43080}      \\
carioca\_5\_14\_01 & 537824 & 2497040 & 14    & 1080307930 & \numb{1038.20928}    \\
carioca\_5\_14\_02 & 537824 & 2497040 & 14    & 1082279116 & \numb{150.91310}     \\
carioca\_5\_14\_03 & 537824 & 2497040 & 14    & 931463937  & \numb{284.22063}     \\
carioca\_5\_14\_04 & 537824 & 2497040 & 14    & 1037634219 & \numb{821.46846}     \\
carioca\_5\_14\_05 & 537824 & 2497040 & 14    & 1072793454 & \numb{1224.01345}    \\
carioca\_5\_15\_01 & 759375 & 3543750 & 15    & 1011895745 & \numb{1046.36109}    \\
carioca\_5\_15\_02 & 759375 & 3543750 & 15    & 1067623193 & \numb{888.80817}     \\
carioca\_5\_15\_03 & 759375 & 3543750 & 15    & 1093631593 & \numb{1258.17189}    \\
carioca\_5\_15\_04 & 759375 & 3543750 & 15    & 890715927  & \numb{66.47183}      \\
carioca\_5\_15\_05 & 759375 & 3543750 & 15    & 1112392828 & \numb{638.51534}     \\
carioca\_5\_16\_01 & 1048576 & 4915200 & 16   & 1140155635 & \numb{3259.37667}    \\
carioca\_5\_16\_02 & 1048576 & 4915200 & 16   & 1114675222 & \numb{2047.70027}    \\
carioca\_5\_16\_03 & 1048576 & 4915200 & 16   & 1097447396 & \numb{464.25447}     \\
carioca\_5\_16\_04 & 1048576 & 4915200 & 16   & \nosol{}   & \notime{}            \\
carioca\_5\_16\_05 & 1048576 & 4915200 & 16   & 1034250551 & \numb{815.79351}     \\
carioca\_5\_17\_01 & 1419857 & 6681680 & 17   & 1084906998 & \numb{803.61294}     \\
carioca\_5\_17\_02 & 1419857 & 6681680 & 17   & \nosol{}   & \nomem{}             \\
carioca\_5\_17\_03 & 1419857 & 6681680 & 17   & 1030965254 & \numb{845.50173}     \\
carioca\_5\_17\_04 & 1419857 & 6681680 & 17   & 1154984533 & \numb{3195.59585}    \\
carioca\_5\_17\_05 & 1419857 & 6681680 & 17   & \nosol{}   & \notime{}            \\
carioca\_5\_18\_01 & 1889568 & 8922960 & 18   & \nosol{}   & \nomem{}             \\
carioca\_5\_18\_02 & 1889568 & 8922960 & 18   & \nosol{}   & \nomem{}             \\
carioca\_5\_18\_03 & 1889568 & 8922960 & 18   & 1177091608 & \numb{1081.01490}    \\
carioca\_5\_18\_04 & 1889568 & 8922960 & 18   & \nosol{}   & \notime{}            \\
carioca\_5\_18\_05 & 1889568 & 8922960 & 18   & \nosol{}   & \nomem{}             \\
}
{Type: Random 5-d rectilinear instances with coordinates scaled by $10^8$. Instances with 19 and 20 terminals omitted.}

\testset{bonn-3d}
{
bonn\_3\_21\_1 & 8820  & 25179   & 21    & 6217       & \numb{1.54232}       \\
bonn\_3\_21\_2 & 8820  & 25179   & 21    & 6729       & \numb{1.94449}       \\
bonn\_3\_21\_3 & 8820  & 25179   & 21    & 5738       & \numb{0.94769}       \\
bonn\_3\_22\_1 & 10648 & 30492   & 22    & 6681       & \numb{842.79205}     \\
bonn\_3\_22\_2 & 10648 & 30492   & 22    & 6797       & \numb{3.24157}       \\
bonn\_3\_22\_3 & 10164 & 29084   & 22    & 6941       & \numb{4.23854}       \\
bonn\_3\_23\_1 & 12167 & 34914   & 23    & 6195       & \numb{1.53011}       \\
bonn\_3\_23\_2 & 12167 & 34914   & 23    & 6094       & \numb{2.27521}       \\
bonn\_3\_23\_3 & 11638 & 33373   & 23    & 6398       & \numb{12.39354}      \\
bonn\_3\_24\_1 & 13824 & 39744   & 24    & 6622       & \numb{4.94802}       \\
bonn\_3\_24\_2 & 13248 & 38064   & 24    & 7136       & \numb{14.76881}      \\
bonn\_3\_24\_3 & 13248 & 38064   & 24    & 7014       & \numb{7.01045}       \\
bonn\_3\_25\_1 & 15625 & 45000   & 25    & 6928       & \numb{12.24094}      \\
bonn\_3\_25\_2 & 15625 & 45000   & 25    & 7249       & \numb{9.19540}       \\
bonn\_3\_25\_3 & 15625 & 45000   & 25    & 7504       & \numb{36.02245}      \\
bonn\_3\_26\_1 & 17576 & 50700   & 26    & 8184       & \numb{11.84886}      \\
bonn\_3\_26\_2 & 17576 & 50700   & 26    & 7172       & \numb{6.53602}       \\
bonn\_3\_26\_3 & 16900 & 48724   & 26    & 7732       & \numb{7.60750}       \\
bonn\_3\_27\_1 & 18252 & 52676   & 27    & 7504       & \numb{12.94104}      \\
bonn\_3\_27\_2 & 19683 & 56862   & 27    & 7032       & \numb{20.36483}      \\
bonn\_3\_27\_3 & 18954 & 54729   & 27    & 7688       & \numb{21.09389}      \\
bonn\_3\_28\_1 & 21168 & 61208   & 28    & 8106       & \numb{1795.93082}    \\
bonn\_3\_28\_2 & 21168 & 61208   & 28    & 7774       & \numb{384.24007}     \\
bonn\_3\_28\_3 & 21952 & 63504   & 28    & 7307       & \numb{22.96362}      \\
bonn\_3\_29\_1 & 22736 & 65800   & 29    & 8231       & \numb{248.49025}     \\
bonn\_3\_29\_2 & 23548 & 68179   & 29    & 8542       & \numb{37.71728}      \\
bonn\_3\_29\_3 & 22736 & 65800   & 29    & 8053       & \numb{23.26990}      \\
bonn\_3\_30\_1 & 26100 & 75660   & 30    & 8141       & \numb{181.24496}     \\
bonn\_3\_30\_2 & 26100 & 75660   & 30    & 7897       & \numb{25.81434}      \\
bonn\_3\_30\_3 & 26100 & 75660   & 30    & 8527       & \numb{63.45128}      \\
bonn\_3\_31\_1 & 27869 & 80848   & 31    & 8664       & \numb{16.91324}      \\
bonn\_3\_31\_2 & 28830 & 83669   & 31    & 8409       & \numb{328.87897}     \\
bonn\_3\_31\_3 & 28830 & 83669   & 31    & 7583       & \numb{45.30390}      \\
bonn\_3\_32\_1 & 31744 & 92224   & 32    & 9483       & \numb{59.04668}      \\
bonn\_3\_32\_2 & 31744 & 92224   & 32    & 8828       & \numb{2371.74264}    \\
bonn\_3\_32\_3 & 31744 & 92224   & 32    & 8261       & \numb{8.53939}       \\
bonn\_3\_33\_1 & 33792 & 98240   & 33    & 9588       & \numb{46.32895}      \\
bonn\_3\_33\_2 & 35937 & 104544  & 33    & 9658       & \numb{465.67574}     \\
bonn\_3\_33\_3 & 31744 & 92224   & 33    & 8902       & \numb{83.40442}      \\
bonn\_3\_34\_1 & 37026 & 107745  & 34    & 9045       & \numb{410.33255}     \\
bonn\_3\_34\_2 & 38148 & 111044  & 34    & 9664       & \numb{1822.08714}    \\
bonn\_3\_34\_3 & 37026 & 107745  & 34    & 9105       & \numb{959.83567}     \\
bonn\_3\_35\_1 & 41650 & 121345  & 35    & \nosol{}   & \notime{}            \\
bonn\_3\_35\_2 & 40460 & 117844  & 35    & 9372       & \numb{513.59911}     \\
bonn\_3\_35\_3 & 40460 & 117844  & 35    & 9803       & \numb{528.18720}     \\
bonn\_3\_36\_1 & 42768 & 124632  & 36    & 9353       & \numb{2906.58774}    \\
bonn\_3\_36\_2 & 45360 & 132264  & 36    & 9118       & \numb{87.81291}      \\
bonn\_3\_36\_3 & 44100 & 128555  & 36    & \nosol{}   & \notime{}            \\
bonn\_3\_37\_1 & 49284 & 143819  & 37    & 9379       & \numb{811.91222}     \\
bonn\_3\_37\_2 & 47915 & 139786  & 37    & 9596       & \numb{1602.80635}    \\
bonn\_3\_37\_3 & 49284 & 143819  & 37    & 8768       & \numb{113.70363}     \\
bonn\_3\_38\_1 & 53428 & 156028  & 38    & \nosol{}   & \notime{}            \\
bonn\_3\_38\_2 & 52022 & 151885  & 38    & 9895       & \numb{317.96666}     \\
bonn\_3\_38\_3 & 52022 & 151885  & 38    & \nosol{}   & \notime{}            \\
bonn\_3\_39\_1 & 54834 & 160171  & 39    & 9757       & \numb{581.36245}     \\
bonn\_3\_39\_2 & 54834 & 160171  & 39    & \nosol{}   & \notime{}            \\
bonn\_3\_39\_3 & 57798 & 168909  & 39    & \nosol{}   & \notime{}            \\
bonn\_3\_40\_1 & 64000 & 187200  & 40    & 9024       & \numb{25.43975}      \\
bonn\_3\_40\_2 & 57798 & 168909  & 40    & 9633       & \numb{710.09162}     \\
bonn\_3\_40\_3 & 59280 & 173278  & 40    & 10696      & \numb{1183.49634}    \\
bonn\_3\_45\_1 & 89100  & 261315 & 45    & \nosol{}   & \notime{}            \\
bonn\_3\_45\_2 & 89100  & 261315 & 45    & \nosol{}   & \notime{}            \\
bonn\_3\_45\_3 & 85140  & 249613 & 45    & \nosol{}   & \notime{}            \\
bonn\_3\_50\_1 & 110544 & 324721 & 50    & \nosol{}   & \nomem{}             \\
bonn\_3\_50\_2 & 117600 & 345598 & 50    & \nosol{}   & \notime{}            \\
bonn\_3\_50\_3 & 112800 & 331394 & 50    & \nosol{}   & \nomem{}             \\
bonn\_3\_55\_1 & 160380 & 472284 & 55    & \nosol{}   & \notime{}            \\
bonn\_3\_55\_2 & 151470 & 445881 & 55    & \nosol{}   & \notime{}            \\
bonn\_3\_55\_3 & 154548 & 455004 & 55    & 12138      & \numb{6201.10443}    \\
bonn\_3\_60\_1 & 198417 & 585044 & 60    & \nosol{}   & \notime{}            \\
bonn\_3\_60\_2 & 201898 & 595369 & 60    & \nosol{}   & \notime{}            \\
bonn\_3\_60\_3 & 198476 & 585220 & 60    & \nosol{}   & \nomem{}             \\

}
{Type: Random 3-d rectilinear instances. Coordinates were chosen uniformly at random from $\{0, 1, \ldots, 999\}$.}

\end{document}